\documentclass[10pt]{article}
\usepackage{graphicx,latexsym,amsmath,amsthm,amssymb,color,verbatim,mathrsfs,bbm,amscd}
\usepackage[latin1] {inputenc}
\usepackage[margin=3cm]{geometry}
\usepackage[english]{babel}
\usepackage{setspace,paralist}
\usepackage[enableskew]{youngtab}
\usepackage{leftidx}
\usepackage[all]{xy}
\usepackage{tikz,url}
\usetikzlibrary{matrix,decorations.pathreplacing}
\usepackage{floatrow}
\usepackage[all]{xy}
\usepackage{longtable}
\newcommand{\Sym}{\textup{Sym}}
\DeclareMathOperator{\mult}{mult}

\SelectTips{eu}{10}
\newcommand{\inj}{\ar@{^{(}->}}
\newcommand{\surj}{\ar@{->>}}

\newcommand{\la}{\lambda}
\renewcommand{\det}{\mathrm{det}}
\def\al{\alpha}
\newcommand{\IC}{\ensuremath{\mathbb{C}}}
\newcommand{\IA}{\ensuremath{\mathbb{A}}}
\newcommand{\IN}{\ensuremath{\mathbb{N}}}

\newcommand{\aS}{\ensuremath{\mathfrak{S}}}

\DeclareMathOperator{\diag}{diag}
\DeclareMathOperator{\sgn}{sgn}
\newcommand{\pl}[3]{a_{#1}({#2[#3]})}

\newcommand{\Ch}{\textup{\textsf{Ch}}}
\newcommand{\Pow}{\textup{\textsf{Pow}}}

\newcommand{\tensor}{\smash{\textstyle\bigotimes}}
\newcommand{\GL}{\mathsf{GL}}

\newcommand{\HWV}{\mathsf{HWV}}

\newcommand{\at}{\,@\;}

\newcommand{\partinto}[1][]{\smash{\mathord{\mathchoice{%
  \xymatrix@=0.4em@1{%
  \ar@{|-}[rr]_-*--{\scriptstyle #1}
  &*{\phantom{\scriptstyle{#1}}}&}
}{
  \xymatrix@=0.25em@1{%
  \ar@{|-}[rr]_-*--{\scriptstyle #1}
  &*{\phantom{\scriptstyle{#1}}}&}
}{
  \xymatrix@=0.2em@1{%
  \ar@{|-}[rr]_-*--{\scriptscriptstyle #1}
  &*{\phantom{\scriptscriptstyle{#1}}}&}
}{}}}}
\newcommand{\partintonosmash}[1][]{\mathord{\mathchoice{%
  \xymatrix@=0.4em@1{%
  \ar@{|-}[rr]_-*--{\scriptstyle #1}
  &*{\phantom{\scriptstyle{#1}}}&}
}{
  \xymatrix@=0.25em@1{%
  \ar@{|-}[rr]_-*--{\scriptstyle #1}
  &*{\phantom{\scriptstyle{#1}}}&}
}{
  \xymatrix@=0.2em@1{%
  \ar@{|-}[rr]_-*--{\scriptscriptstyle #1}
  &*{\phantom{\scriptscriptstyle{#1}}}&}
}{}}}

\swapnumbers
\numberwithin{equation}{section}
\newtheorem{theorem}[equation]{Theorem}
\newtheorem{corollary}[equation]{Corollary}
\newtheorem{lemma}[equation]{Lemma}

\newtheorem{proposition}[equation]{Proposition}

\theoremstyle{definition}
\newtheorem{remark}[equation]{Remark}

\title{On geometric complexity theory: Multiplicity obstructions are stronger than occurrence obstructions}
\author{Julian D\"orfler\footnote{Saarland University, s8judoer$@$stud.uni-saarland.de}, Christian Ikenmeyer\footnote{Max Planck Institute for Informatics, Simons Institute for the Theory of Computing, cikenmey$@$mpi-inf.mpg.de}, Greta Panova\footnote{University of Southern California, University of Pennsylvania, Simons Institute for the Theory of Computing, gpanova$@$usc.edu}}
\date{\today}

\begin{document}
\sloppy

\maketitle
\begin{abstract}
Geometric Complexity Theory as initiated by Mulmuley and Sohoni in two papers (SIAM J Comput 2001, 2008) aims to separate algebraic complexity classes via representation theoretic multiplicities in coordinate rings of specific group varieties.
The papers also conjecture that the vanishing behavior of these multiplicities would be sufficient to separate complexity classes (so-called \emph{occurrence obstructions}).
The existence of such strong occurrence obstructions has been recently disproven in 2016 in two successive papers, Ikenmeyer-Panova (Adv.\ Math.) and B\"urgisser-Ikenmeyer-Panova (J.\ AMS).

This raises the question whether separating group varieties via representation theoretic multiplicities is stronger than separating them via occurrences.
This paper provides for the first time a setting where separating with multiplicities can be achieved, while the separation with occurrences is provably impossible.
Our setting is surprisingly simple and natural: We study the variety of products of homogeneous linear forms (the so-called Chow variety) and the variety of polynomials of bounded border Waring rank (i.e.\ a higher secant variety of the Veronese variety).

As a side result we prove a slight generalization of Hermite's reciprocity theorem, which proves Foulkes' conjecture for a new infinite family of cases.
\end{abstract}

\section{Introduction}
In two landmark papers \cite{gct1, gct2} Mulmuley and Sohoni suggested the use of representation theoretic multiplicities
to separate group varieties that correspond to complexity classes.
The goal of this approach, which is called \emph{geometric complexity theory}, is to achieve complexity lower bounds that lead to the separation of algebraic complexity classes such as VP and VNP (see \cite{bcs:97} or \cite{sap:16} for the precise definitions, which will not be important in this paper).
At the heart of the approach was the hope that so-called \emph{occurrence obstructions} (see Section~\ref{sec:reprthobs}) would be sufficient to separate VP and VNP.
In \cite{IP:17, BIP:19} it was shown that occurrence obstructions are too weak to provide the necessary separation, at least for the group varieties that were originally proposed by Mulmuley and Sohoni.
But representation theoretic multiplicities might still be able to separate VP and VNP when we look at the finer separation criterion via \emph{multiplicity obstructions} (see also Section~\ref{sec:reprthobs}).
Unfortunately, so far all known separations of group varieties via multiplicity obstructions could also in fact be obtained via occurrence obstructions,
or at least there is no setting in which multiplicity obstructions are provably stronger than occurrence obstructions, see e.g.~\cite{bi:10, bi:13}.
Indeed, little is known about multiplicity obstructions in general, as the required multiplicities are often \#P-hard to compute, see e.g.\ \cite{nara:06, bi:08, bor:09},
which implies that a polynomial time algorithm for their computation can only exist if P=NP.

Scott Aaronson raised the question about the existence of a setting where multiplicity obstructions are provably more powerful than occurrence obstructions.
In this paper we give the first example of such a situation, see Theorem~\ref{thm:main} below.
Our separation result is extremely modest, which indicates that such behaviour is expected in many more cases.

As a side result we prove a slight generalization of Hermite's reciprocity theorem, which proves Foulkes' conjecture (see \eqref{eq:foulkes}) for a new infinite family of cases, see Theorem~\ref{thm:plethineq}.

\section{Representation theoretic obstructions}\label{sec:reprthobs}
In this section we review how to separate group varieties via representation theoretic multiplicities.
The setup is in complete analogy to the geometric complexity theory approach of Mulmuley and Sohoni.
We then list our main result, see Theorem~\ref{thm:main}.

Consider the space $\IA_m^n := \IC[x_1,\ldots,x_m]_n$ of complex homogeneous polynomials of degree $n$ in $m$ variables.
Let $V := \IA_m^1$ be the space of homogeneous degree 1 polynomials.
In this paper we compare two subvarieties of $\IA_m^n$.
The first is the so-called \emph{Chow variety}
\[
\Ch_m^n := \{\ell_1 \cdots \ell_n \mid \ell_i \in V \} \subseteq \IA_m^n,
\]
which is the set of polynomials that can be written as a product of homogeneous linear forms, see e.g.\ \cite[\S8.6]{lan:11}.
The second variety is called a \emph{higher secant variety of the Veronese variety} and can be written as
\[
\Pow_{m,k}^n := \overline{\{\ell_1^n +  \cdots + \ell_k^n \mid \ell_i \in V \}} \subseteq \IA_m^n,
\]
which is the closure of the set of all sums of $k$ powers of homogeneous linear forms.
Note that from a general principle it follows that the Zariski closure equals the Euclidean closure in this case,
see e.g.\ \cite[\S 2.C]{mum:complprojvars} where this is shown for every constructible set.
The polynomials in $\Pow_{m,k}^n$ are exactly those that have \emph{border Waring rank at most $k$}, see e.g.~\cite[\S5.4]{lan:11}.

$\IA_m^n$ is generated as a vector space by the powers $v^n$, $v \in V$, see e.g.~\cite[Ex.~2.6.6.2]{lan:11}.
Given two elements $g_1, g_2 \in \GL_m := \GL(V)$, and given $v \in V$, we clearly have $g_1(g_2 v) = (g_1 g_2) v$.
Thus we say that $V$ admits a $\GL_m$-action.
This natural action of $\GL_m$ on $V$ lifts canonically to $\IA_m^n$ via $g(v^n) := (gv)^n$, $g \in \GL_m$, $v \in V$, and linear continuation.
Both varieties $\Ch_m^n$ and $\Pow_{m,k}^n$ are closed under this action, i.e., for $g \in \GL_m$ and $v \in \Ch_m^n$ we have $gv \in \Ch_m^n$,
and analogously $v \in \Pow_{m,k}^n$ implies $gv \in \Pow_{m,k}^n$.
A variety that is closed under the action of $\GL_m$ is called a $\GL_m$-variety.

Let $\IC[\IA_m^n]$ denote the coordinate ring of $\IA_m^n$, i.e., the polynomial ring in $\dim \IA_m^n = \binom{n+m-1}{n}$ many variables,
where these variables are in 1:1 correspondence to the monomials in $\IA_m^n$.
The action of $\GL_m$ on $\IA_m^n$ lifts to a linear action of $\GL_m$ on $\IC[\IA_m^n]$ via the canonical pullback as follows:
\[
(gf)(h) := f(g^{-1} h), \ g \in \GL_m, \ f \in \IC[\IA_m^n], \ h \in \IA_m^n.
\]
Moreover, the action respects the natural grading of $\IC[\IA_m^n]$, so that each homogeneous degree $d$ part $\IC[\IA_m^n]_d$ is a finite dimensional vector space that is closed under the action of $\GL_m$.

Recall that a finite dimensional vector space $W$ that is closed under a linear action of $\GL_m$ is called a $\GL_m$-\emph{representation}.
This is equivalent to the existence of a group homomorphism $\varrho:\GL_m \to \GL(W)$.
If we choose bases, then we can interpret
$\GL_m \subseteq \IC^{m \times m}$ and $\GL(W) \subseteq \IC^{\dim W \times \dim W}$ and $\varrho$ is described by $(\dim W)^2$ many coordinate functions,
which are functions in $m^2$ many variables.
If these functions are polynomials, then we call $W$ a \emph{polynomial representation}.
Our main representation of interest, $\IC[\IA_m^n]_d$, is a polynomial representation.
A linear subspace of $W$ that is closed under the action of $\GL_m$ is called a subrepresentation.
Subrepresentations of polynomial representations are clearly polynomial representations again.
For every $\GL_m$-representation $W$ we have that $W$ and $0$ are two trivial subrepresentations.
If $W$ has no other subrepresentations, then we call $W$ \emph{irreducible}.
A linear map $\varphi: W_1 \to W_2$ between two $\GL_m$-representations is called \emph{equivariant} if $g \varphi(f) = \varphi(gf)$ for all $f\in W_1$, $g\in\GL_m$.
If there exists an equivariant vector space isomorphism from $W_1$ to $W_2$, then we say that $W_1$ and $W_2$ are \emph{isomorphic} $\GL_m$-representations.
An \emph{$m$-partition} of $D$ is a nonincreasing list of $m$ nonnegative integers that sum up to $D$.
Every irreducible polynomial $\GL_m$-representation has an associated isomorphism type, which is an $m$-partition, see e.g.\ \cite[Ch.~8]{fult:97}.
Two irreducible $\GL_m$-representations are isomorphic iff their isomorphism types coincide.
We denote by $\{\la\}_m$ the irreducible $\GL_m$-representation corresponding to the $m$-partition $\la$.
We write $\{\la\} = \{\la\}_m$ is $m$ is clear from the context.

The group $\GL_m$ is \emph{linearly reductive}, which means that
every $\GL_m$-representation $W$ decomposes into a direct sum of irreducible $\GL_m$-representations, see e.g.\ \cite[AII.5, Satz 4]{Kra:85}.
The number of times an irreducible representation of type $\la$ occurs in the decomposition is called the \emph{multiplicity} of $\la$ in $W$, written $\mult_\la(W)$.
Even though this decomposition is usually not unique, the notation $\mult_\la(W)$ makes sense,
because the multiplicities are independent of the actual decompositions.

The multiplicity $\pl \la d n := \mult_\la(\IC[\IA_m^n]_d)$ is the infamous \emph{plethysm coefficient},
which is the object of study in Foulkes' conjecture and also in Problem 9 in Stanley's famous list of open problems \cite{sta:00}.
If we pad an $m$-partition $\la$ with $m'-m$ many zeros to obtain the $m'$-partitions $\la'=(\la_1,\ldots,\la_m,0,\ldots,0)$,
then $\mult_\la(\IC[\IA_m^n]_d) = \mult_{\la'}(\IC[\IA_{m'}^n]_d)$, see e.g.~\cite[Lem.~4.3.2]{ike:12b}.
For the sake of simplicity we identify $m$-partitions with $m'$-partitions that arise from padding zeros.
This justifies leaving out the parameter $m$ in the notation $\pl \la d n$ by assuming that $m$ is large enough.
Foulkes' conjecture states that
\begin{equation}\label{eq:foulkes}
\text{Conjecture}: \quad  \pl \la n d \leq \pl \la d n \ \text{ for all $d \geq n$}.
\end{equation}
Conjecture \eqref{eq:foulkes} is known to be true (moreover, equality holds: $\pl \la d n = \pl \la n d$) for all 2-partitions $\la$,
which is often called \emph{Hermite reciprocity}~\cite{her:1854}.
We make modest progress on this conjecture by proving it for many families of 3-partitions, see Corollary~\ref{cor:keyinequality}.

Let $Z$ be a $\GL_m$-variety, e.g., $Z=\Ch_m^n$ or $Z=\Pow_{m,k}^n$.
Then the vanishing ideal $I(Z):=\{f \in \IC[\IA_m^n] \mid \forall h \in Z : f(h)=0 \}$ is also closed under the action of $\GL_m$,
which is easy to verify: If $f(h) = 0$ for all $h \in Z$, then also $(gf)(h) = f(g^{-1} h) = 0$, because $g^{-1} h \in Z$.
Since the action respects the grading, each homogeneous degree $d$ part $I(Z)_d$ is a $\GL_m$-representation.
The coordinate ring $\IC[Z]$ is defined as the quotient algebra $\IC[\IA_m^n]/I(Z)$ and each homogeneous part $\IC[Z]_d = \IC[\IA_m^n]_d/I(Z)_d$ is a $\GL_m$-representation.
Equivalently, we can define $\IC[Z]$ as the set of restrictions of functions in $\IC[\IA_m^n]$ to $Z$.

For most sets of parameters we have $\Pow_{m,k}^n \not\subseteq \Ch_m^n$, but there are some exceptions.
Clearly $\Pow_{m,1}^n \subseteq \Ch_m^n$.
Moreover, $\Pow_{1,k}^n = \Ch_1^n$ for all $n\geq 1$, $k \geq 1$; and $\Pow_{m,k}^1 = \Ch_m^1$ for all $m\geq 1$, $k \geq 1$.
It is also easy to see that $\Pow_{2,2}^2 \subseteq \Ch_2^2$, because $\ell_1^2 + \ell_2^2 = (\ell_1 + i \ell_2)(\ell_1 - i \ell_2)$, where $i^2 = -1$.
More generally, $(\ell_1 + \zeta \ell_2)(\ell_1 + \zeta^2 \ell_2)\cdots (\ell_1 + \zeta^n \ell_2) = \ell_1^n +\zeta^{\frac{n(n+1)}{2}}\ell_2^n$ for $\zeta^n=1$,
which implies $\Pow_{m,2}^n \subseteq \Ch_m^n$.
For $m=2$, $k\geq 1$, $n\geq 1$, we have $\Pow_{m,k}^n \subseteq \Ch_m^n$ by the fundamental theorem of algebra.
These are the only exceptions, as
for $n\geq 2$, $m\geq 3$, $k\geq 3$ we have $\Pow_{m,k}^n \not\subseteq \Ch_m^n$:
the polynomial $x^n + y^n + z^n$ of the Fermat curve is in $\Pow_{m,k}^n$ and its irreducibility implies (since $n \geq 2$) that $x^n + y^n + z^n \notin \Ch_m^n$.

We will see that for specific settings of parameters
there exist \emph{multiplicity obstructions} that prove $\Pow_{m,k}^n \not\subseteq \Ch_m^n$,
but there do not exist \emph{occurrence obstructions} that prove this fact (see the definitions below).
Our approach works as follows and is in complete analogy to the approach proposed in \cite{gct1,gct2} to separate group varieties arising from algebraic complexity theory.
If $\Pow_{m,k}^n \subseteq \Ch_m^n$, then the restriction of functions gives a canonical $\GL_m$-equivariant surjection
\[
\IC[\Ch_m^n]_d \twoheadrightarrow \IC[\Pow_{m,k}^n]_d.
\]
In this case, Schur's lemma (e.g.~\cite[Lemma 4.1.4]{gw:09}) implies that
\begin{equation}\label{eq:obsineq}
\mult_\la(\IC[\Ch_m^n]_d) \geq \mult_\la(\IC[\Pow_{m,k}^n]_d).
\end{equation}
for all $m$-partitions $\la$.
Therefore, a partition $\la$ that violates \eqref{eq:obsineq} proves that $\Pow_{m,k}^n \not\subseteq \Ch_m^n$. Such a $\la$ is called a \emph{multiplicity obstruction}.
If additionally $\mult_\la(\IC[\Ch_m^n]_d)=0$, then $\la$ is called an \emph{occurrence obstruction}.

Since $\Ch_m^n$ and $\Pow_{m,k}^n$ are subvarieties of $\IA_m^n$ and since all $\la$ for which $\mult_\la(\IC[\IA_m^n]_d)>0$ are $m$-partitions of $dn$,
it follwos that if $\mult_\la(\IC[\Ch_m^n]_d)>0$ or $\mult_\la(\IC[\Pow_{m,k}^n]_d)>0$, then $\la$ is an $m$-partition of $dn$.

\begin{theorem}[Main Theorem]\label{thm:main}
~\\\textup{(1)} Let $m\geq 3$, $n \geq 2$, $k=d=n+1$, $\la=(n^2-2,n,2)$. We have
$\mult_\la(\IC[\Ch_m^n]_d) < \mult_\la(\IC[\Pow_{m,k}^n]_d)$,
i.e., $\la$ is a multiplicity obstruction that shows $\Pow_{m,k}^n \not\subseteq \Ch_m^n$.

\noindent\textup{(2)} In two finite settings we can show a slightly stronger separation:
\begin{compactitem}
\item[\textup{(a)}] Let $k=4$, $n=6$, $m=3$, $d=7$, $\la=(n^2-2,n,2)=(34,6,2)$. Then
$\mult_\la(\IC[\Ch_m^n]_d) = 7 < 8 = \mult_\la(\IC[\Pow_{m,k}^n]_d)$, i.e., $\la$ is a multiplicity obstruction that shows $\Pow_{m,k}^n \not\subseteq \Ch_m^n$.
\item[\textup{(b)}] Similarly, for $k=4$, $n = 7$, $m = 4$, $d = 8$, $\la=(n^2-2,n,2)=(47,7,2)$ we have $\mult_\la(\IC[\Ch_m^n]_d) < 11 = \mult_\la(\IC[\Pow_{m,k}^n]_d)$, i.e., $\la$ is a multiplicity obstruction that shows $\Pow_{m,k}^n \not\subseteq \Ch_m^n$.
\end{compactitem}
Both separations \textup{(a)} and \textup{(b)} cannot be achieved using occurrence obstructions, even for arbitrary $k$:
for all $m$-partitions $\mu$ 
that satisfy $\pl \mu d n >0$
we have $\mult_\la(\IC[\Ch_m^n]_{d'})>0$ in these settings.
\end{theorem}

One would like to show that there are no occurrence obstructions in all cases (1),
but this is wrong if $n$ is not large enough with respect to $m$, see Prop.~\ref{pro:occobsdoexist}.
Even for $m=3$ or $m=4$ ruling out occurrence obstructions as in (2) is done by a large-scale computer calculation which is only suitable for a finite case, but not for sequences as in (1).
The papers \cite{IP:17, BIP:19} rule out occurrence obstructions for families,
but only in ranges where they would give very strong new algebraic circuit lower bounds,
so that we expect it to be difficult to find multiplicity obstructions in those cases.
Note also that \cite{IP:17, BIP:19} are only dealing with \emph{padded polynomials},
for which \cite{kl:14} guarantees $\la$ to have a very restricted shape.

We expect multiplicity obstructions to be more powerful than occurrence obstructions in most cases relevant for geometric complexity theory,
and Theorem~\ref{thm:main} resolves the challenge of finding a setting in which
the corresponding multiplicities and occurrences could actually be computed in a reasonable amout of time,
while the setting is also involved enough so that a difference between occurrence obstructions and multiplicity obstructions could be witnessed.

\begin{remark}
The partition $(n^2-2,n,2)$ is known to be the type of one of Brill's classical set-theoretic equations for $\Ch_m^n$, see \cite{Gua:18}.
\end{remark}

\section{Proof of the main theorem}
The main theorem (Theorem~\ref{thm:main})
makes a statement about the finite situations $k=4$, $n=6$, $m=3$, $d=7$ and $k=4$, $n=7$, $m=4$, $d=8$,
as well as the general situation $m\geq 3$, $n \geq 2$, $k=d=n+1$.
As a first step, in all these cases we show that
\begin{equation}\label{eq:equality}
\mult_\la(\IC[\Pow_{m,k}^n]_d) = \pl \la d n.
\end{equation}
In the finite cases the following computer calculation suffices to prove \eqref{eq:equality}.
\begin{proposition}\label{pro:smallcomputercalc}
$\mult_{(34,6,2)}(\IC[\Pow_{3,4}^6]_7) = 8 = \pl {(34,6,2)} 7 6$ and
$\mult_{(47,7,2)}(\IC[\Pow_{3,4}^7]_8) = 11 = \pl {(47,7,2)} 8 7$.
\end{proposition}
\begin{proof}
The plethysm coefficient computations were performed with the \textsc{LiE} software.
The rest is a small computer calculation completely analogous to the ones in \cite[Sec.~6]{BIP:19}.
The details can be found in Section~\ref{sec:computer_resultssmallcomputercalc}.
\end{proof}
For the general situation the equality \eqref{eq:equality} is a consequence of the following result on power sums proved in \cite[Prop.~3.2]{BIP:19}:
\begin{proposition}\label{pro:lowdegreepowersums}
If $\la$ is an $m$-partition of $dn$ and $k\geq d$, then
$\mult_\la(\IC[\Pow_{m,k}^n]_d) = \pl \la d n$.
\end{proposition}
As a second step we will use the following lemma for $\la=(n^2-2,n,2)$.
\begin{lemma}[{see also \cite[Sec.~9.2.3]{Lan:17}}]\label{lem:chowupperbound}
Let $\la$ be an $m$-partition and $n\geq m$. Then $\mult_{\la}(\IC[\Ch_m^n]_{d}) \leq \pl \la n d$.
\end{lemma}
\begin{proof}
Let $\GL_n(x_1 \cdots x_n) := \{g(x_1\cdots x_n) \mid g \in \GL_n\} \subseteq \IA_n^n$ denote the $\GL_n$-orbit of $x_1 \cdots x_n$.
We denote by $\overline{\GL_n(x_1 \cdots x_n)}$ the Zariski closure of this orbit, which equals its Euclidean closure by the same principles as in Section~\ref{sec:reprthobs}.
Choose bases and embed $\IA_m^n \subseteq \IA_n^n$, so that $\Ch_m^n$ is the intersection of $\IA_m^n$ and $\overline{\GL_n(x_1 \cdots x_n)}$.
This implies (via argument analogous to that for the plethysm coefficient (\cite[Lem.~4.3.2]{ike:12b}))
that the multiplicity of the irreducible $\GL_m$-representation $\{\la\}_m$
in $\IC[\Ch_m^n]_{d}$ equals the multiplicity of the irreducible $\GL_n$-representation $\{\la\}_n$
in $\IC[\overline{\GL_n (x_1 \cdots x_n)}]$. In other words
\[
\mult_\la(\IC[\Ch_m^n]_{d}) = \mult_\la(\IC[\overline{\GL_n (x_1 \cdots x_n)}]_d).
\]
The vector space $\IC[\overline{\GL_n (x_1 \cdots x_n)}]_d$ consists of exactly the restrictions of polynomials in $\IC[\IA_n^n]_d$ to the orbit $\GL_n (x_1 \cdots x_n)$.
The coordinate ring $\IC[\GL_n (x_1 \cdots x_n)]$ is also graded
and its homogeneous degree $d$ part $\IC[\GL_n (x_1 \cdots x_n)]_d$
consists of \emph{all} homogeneous degree $d$ regular functions on $\GL_n (x_1 \cdots x_n)$, in particular
\[
\mult_\la(\IC[\overline{\GL_n (x_1 \cdots x_n)}]_d) \leq \mult_\la(\IC[\GL_n (x_1 \cdots x_n)]_d).
\]
The right-hand side can be understood via geometric invariant theory as follows (see \cite[Sec.~3.4(A)]{ike:12b}):
\[
\mult_\la(\IC[\GL_n (x_1 \cdots x_n)]_d) = \mult_{\la^*}(\IC[\GL_n]^{H}_d),
\]
where $H = \{\diag(\alpha_1,\ldots,\alpha_n)\mid \prod_{i=1}^n\alpha_1 = 1\} \rtimes \aS_n \subseteq \GL_n$
is the stabilizer of $x_1 \cdots x_n$.
The algebraic Peter-Weyl theorem (see e.g.\ \cite[II.3.1 Satz 3]{Kra:85}, \cite[Thm.~4.2.7]{gw:09}, or~\cite[Ch.~7, 3.1 Thm.]{Pro:07}) states that
\[
\IC[\GL_n] = \bigoplus_{\la} \{\la\} \otimes \{\la^*\}
\]
and we conclude
\[
\mult_\la(\IC[\GL_n]^{H}_d) = \dim\{\la\}^H.
\]
There are several ways of seeing that $\dim\{\la\}^H = \pl \la n d$, see e.g.\ \cite[Sec.~9.2.3]{Lan:17} or \cite[Prop.~3.3]{ike:18}.
This proves the lemma.
\end{proof}
Now an argument using symmetric functions is used to prove the following theorem.
\begin{theorem}\label{thm:plethineq}
$\pl {(n^2-2,n,2)}{n+1}n=1+\pl {(n^2-2,n,2)}n{n+1}$.
\end{theorem}
Theorem~\ref{thm:plethineq} is a corollary of more general results, see Corollary~\ref{cor:keyinequality} in the appendix.

This finishes the proof that $(n^2-2,n,2)$ is a multiplicity obstruction in all cases of Theorem~\ref{thm:main}.

\subsection*{No occurrence obstructions}
To finish the proof of Theorem~\ref{thm:main}(2), it remains to show that there are no occurrence obstructions in the finite situation $n=6$, $m=3$ and $n=7$, $m=4$.
We will primarily go into more detail for the first case and the second one will be proven similarly.
We will do this by showing that
\begin{equation}\label{eq:posplethimpliesposchow}
\pl \mu d n >0 \text{ implies } \mult_\mu(\IC[\Ch_m^n]_d)>0 \text{ \ for } n=6\text{, } m=3.
\end{equation}
Note that this claim is independent of $k$.
We start proving \eqref{eq:posplethimpliesposchow} by giving a complete classification of when $\pl \mu d n>0$ for the case $n=6$, $m=3$.

First, the following lemma states that for a few special $\mu$ the plethysm coefficient always vanishes.
\begin{lemma}\label{lem:vanishingpleth}
Let $\bar\la := (\la_2,\la_3,\ldots)$ denote $\la$ without its first row.
If $\la$ is an $m$-partition of $dn$ and $\bar\la \in \{(3,3),(3,1),(2,1),(1,1),(1)\}$, then $\pl \la d n=0$.
\end{lemma}
\begin{proof}
This is proved by a finite calculation for all cases but $(3,3)$ as Thm~1.10(a) in \cite{IP:17}.
Exactly the same calculation can be used to also prove the result for the additional partition $(3,3)$.
\end{proof}

For characterizing the set of all $\mu$ for which $\pl \mu d n$ is positive,
we observe that they form a finitely generated \emph{semigroup} and hence we only need to find the semigroup's generators:
\begin{equation}\label{eq:plethsemigroup}
\text{If $\pl \mu d n>0$ and $\pl \nu {d'} n > 0$, then $\pl {\mu+\nu} {d+d'} n >0$.}
\end{equation}
A detailed proof of \eqref{eq:plethsemigroup} can be found for example in \cite[Prop.~21.2.6]{BI:18July25}.

\begin{proposition}\label{pro:listofgenerators}
Define the set
\begin{equation*}
\begin{minipage}{15cm}
$X := \{(6)$, $(6,6)$, $(8,4)$, $(10,2)$, $(6,6,6)$, $(8,6,4)$, $(10,4,4)$, $(9,6,3)$, $(8,8,2)$, $(10,6,2)$, $(11,5,2)$, $(10,7,1)$, $(12,4,2)$, $(11,6,1)$, $(10,8)$, $(14,2,2)$, $(13,4,1)$, $(13,5)$, $(15,3)$, $(8,8,8)$, $(10,8,6)$, $(11,7,6)$, $(10,9,5)$, $(11,8,5)$, $(10,10,4)$, $(12,7,5)$, $(11,9,4)$, $(13,6,5)$, $(12,8,4)$, $(11,10,3)$, $(13,7,4)$, $(12,9,3)$, $(13,8,3)$, $(12,10,2)$, $(15,5,4)$, $(14,7,3)$, $(13,9,2)$, $(13,10,1)$, $(16,5,3)$, $(15,7,2)$, $(14,9,1)$, $(17,4,3)$, $(15,8,1)$, $(15,9)$, $(19,3,2)$, $(18,5,1)$, $(17,7)$, $(10,10,10)$, $(11,10,9)$, $(12,10,8)$, $(13,9,8)$, $(12,11,7)$, $(13,10,7)$, $(14,9,7)$, $(13,11,6)$, $(15,8,7)$, $(13,12,5)$, $(16,7,7)$, $(15,9,6)$, $(14,11,5)$, $(13,13,4)$, $(15,10,5)$, $(15,11,4)$, $(14,13,3)$, $(16,11,3)$, $(15,13,2)$, $(15,14,1)$, $(17,13)$, $(13,12,11)$, $(14,11,11)$, $(13,13,10)$, $(15,11,10)$, $(14,13,9)$, $(16,11,9)$, $(15,13,8)$, $(15,14,7)$, $(18,9,9)$, $(15,15,6)$, $(17,17,2)$, $(18,17,1)$, $(26,5,5)$, $(15,14,13)$, $(16,13,13)$, $(15,15,12)$, $(17,17,8)$, $(18,15,15)$, $(17,17,14)$, $(25,23)$, $(45,45)\}$.
\end{minipage}
\end{equation*}
Here we truncated trailing zeros from the 3-partitions.
The set $X$ is the set of generators of the semigroup of 3-partitions $\mu$ that have \mbox{$\pl \mu d 6 >0$}.
\end{proposition}
The proof of Proposition~\ref{pro:listofgenerators} proceeds in several steps.

A direct computation with the \textsc{LiE} software
verifies $\pl \mu d 6 >0$ for all $\mu \in X \setminus \{(45,45)\}$.
The case $d=15$ runs into memory problems when using \textsc{LiE}. Other software such as \textsc{Schur} stops working when $d=8$.
We used the formula \cite[Cor.~4.2.8]{Stu:08} to verify $\pl {(45,45)} {15} 6 >0$.

We call the number of nonzero parts the \emph{length} of a partition.
We use a brute-force computer verification and a direct computation with \textsc{LiE} to show that for $d\leq 26$
every partition $\mu$ of length $\leq 2$ with $\pl \mu d 6 >0$ is a sum of partitions from the set $X$.
The same computation is done for all 3-partitions, but only up to $d \leq 14$.
The following proposition states that these finite computations completely describe all cases.
\begin{proposition}\label{pro:generators}
If $\la$ is a $3$-partition of $6d$, $d \geq 15$,
and $\bar\la \notin\{(3,3),(3,1),(2,1),(1,1),(1)\}$, then $\la$ is a sum of partitions from $X$.
\end{proposition}
\begin{proof}
For $15 \leq d \leq 17$ we use a computer calculation to show that we can write every such partition $\la$ as a sum of partitions from $X$.
For $d > 17$ we prove this inductively by showing that we can write every $3$-partition $\la$ of $6d$ with $\bar\la \notin\{(3,3),(3,1),(2,1),(1,1),(1)\}$ as a sum of one of the partitions $(6), (6,6) \text{\ or\ } (6,6,6)$ and a smaller $\la'$ with again $\bar{\la'} \notin\{(3,3),(3,1),(2,1),(1,1),(1)\}$.

Let $c_i$ denote the number of columns in $\la$ with exactly $i$ boxes for $i \in \{1,2,3\}$.
Since we have at least $108$ boxes in $\la$,
the pigeonhole principle implies that at least one must be true: $c_1 \geq 6$, $c_2 \geq 10$ or $c_3 \geq 10$.

In the case $c_1 \geq 6$ we have $\la = \la' + (6)$ with $\la'$ being a sum of elements from $X$ since $\bar{\la'} = \bar{\la}$.
In the case $c_2 \geq 10$ we have $\la = \la' + (6,6)$ with $\la'$ being a sum of elements from $X$ as $\la'_2 \geq 4$.
In the case $c_3 \geq 10$ we have $\la = \la' + (6,6,6)$ with $\la'$ being a sum of elements from $X$ as $\la'_3 \geq 4$.
\end{proof}
This finishes the proof of Proposition~\ref{pro:listofgenerators}.

To prove \eqref{eq:posplethimpliesposchow} it is sufficient (and necessary) to show that $\mult_\mu(\IC[\Ch_m^n]_d)>0$ for all $\mu \in X$,
because a semigroup property analogous to \eqref{eq:plethsemigroup} holds (the same proof applies, e.g.\ \cite[Prop.~21.2.6]{BI:18July25}):
\begin{equation}
\text{If $\mult_\mu(\IC[\Ch_m^n]_d)>0$ and $\mult_\nu(\IC[\Ch_m^n]_{d'})>0$, then $\mult_{\mu+\nu}(\IC[\Ch_m^n]_{d+d'})>0$.}
\end{equation}
If the length of $\mu$ is at most 2, we use the following general result.
\begin{proposition}\label{pro:fewrows}
Let $\mu$ be a 3-partition of length at most 2.
If $\pl \mu d n >0$, then $\mult_{\mu}(\IC[\Ch_m^n]_d)>0$.
\end{proposition}
\begin{proof}
We use an inheritance result: If for a 2-partition $\mu$ we have $\mult_{\mu}(\IC[\Ch_2^n]_d)>0$,
and $\nu$ is the 3-partition that arises from $\mu$ by adding a single 0,
then $\mult_{\nu}(\IC[\Ch_3^n]_d)>0$. The proof is completely analogous to other inheritance results, see e.g.\ \cite[Lemma 4.3.2 or Sec.~5.3]{ike:12b}.
Now for 2-partitions $\mu$ we have $\pl \mu d n = \mult_{\mu}(\IC[\Ch_2^n]_d)$, because every homogeneous polynomial in 2 variables decomposes as a product of homogeneous linear polynomials by the fundamental theorem of algebra,
see also e.g.~\cite[Exa.\ 9.1.1.8]{Lan:17}.
This is how the Hermite reciprocity can be proved.
An even simpler argument works if $\mu$ has length 1.
\end{proof}

We finish the proof of \eqref{eq:posplethimpliesposchow} by
using a computer calculation to verify that for all 3-partitions $\mu \in X$ of length 3 we have $\mult_{\mu}(\IC[\Ch_3^6])>0$, see Proposition~\ref{pro:computercalc}.

This finishes the proof of Theorem~\ref{thm:main}(2a).
The proof of Theorem~\ref{thm:main}(2b) is completely analogous as follows. Let $m=4$, $n=7$.

\begin{lemma}\label{lem:vanishingpleth74}
Let $\bar\la := (\la_2,\la_3,\ldots)$ denote $\la$ without its first row.
If $\la$ is an $m$-partition of $dn$ and $\bar\la \in Y$
for $Y := \{(1)$, $(1, 1)$, $(1, 1, 1)$, $(2, 1)$, $(2, 1, 1)$, $(2, 2, 1)$, $(3, 1)$, $(3, 1, 1)$, $(3, 2, 1)$, $(3, 3)$, $(3, 3, 1)$, $(3, 3, 2)$, $(3, 3, 3)$, $(4, 1, 1)$, $(4, 3, 3)$, $(5, 1, 1)$, $(5, 5, 5)$, $(6, 1, 1) \}$, then $\pl \la d n=0$.
\end{lemma}
\begin{proof}
    This is proven exactly like Lemma~\ref{lem:vanishingpleth}.
\end{proof}

The semigroup of $4$-partitions $\la$ that have $\pl \la d 7 > 0$ has 948 generators, listed in Proposition~\ref{pro:listofgenerators74}.
They form a set that we call $X$.

We again use a direct computation with the \textsc{LiE} software to
verify $\pl \mu d 7 >0$ for all \mbox{$\mu \in X \setminus \{(49,49), (24,24,23,23)\}$}.
For both the remaining partitions $\mu \in \{(49, 49), (24,24,23,23)\}$ we prove $\mult_\mu(\IC[\Ch_4^7]_d)>0$ using our computer calculations which also implies $\pl \mu d 7 > 0$.

To prove those are all the generators we use the following proposition which is proved completely analogously to Proposition~\ref{pro:generators}.
\begin{proposition}\label{pro:generators74}
If $\la$ is a $4$-partition of $7d$, $d \geq 14$,
and $\bar\la \notin Y$, then $\la$ is a sum of partitions from $X$.
\end{proposition}

For the next finite case ($n=7$, $k=d=8$, $m=5$) we reached the computational limit of our implementation.
Here we were able to find $5016$ generating partitions of the semigroup of $4$-partitions $\mu$ that have $\pl \mu d 7 > 0$. Unfortunately these do not generate everything excluding the exceptions yet.
We were able to verify for $5000$ generating partitions $\mu$ that $\mult_\mu(\IC[\Ch_m^n]_d) > 0$.
For the remaining ones, we used up to 200 GB of RAM, but this was not sufficient.

\subsection*{Some occurrence obstructions}
As we degenerate the parameter settings and let $n$ get closer to $m$, multiplicity obstructions tend to become occurrence obstructions.
More precisely,
for $m=3$ and values of $n<6$, and for $(m,n)=(4,6)$, some multiplicity obstructions are actually also occurrence obstructions, as the following proposition shows.
\begin{proposition}\label{pro:occobsdoexist}
The following partitions give occurrence obstructions that show $\Pow_{m,d}^n \not\subseteq \Ch_m^n$.
\end{proposition}
\begin{center}
\begin{tabular}{c|c|c|c|c|c}
$m$ & $n$ & $\la$ & $d$ & $\pl \la d n$ & $\pl \la n d$\\\hline
3 & 2 & $(2,2,2)$ & 3 & 1 & 0\\
3 & 3 & $(7,3,2)$ & 4 & 1 & 0\\
3 & 4 & $(11,9,8)$ & 7 & 1 & 0\\
3 & 5 & $(12,9,9)$ & 6 & 1 & 0\\
4 & 6 & $(14,14,13,13)$ & 9 & 11 & 0
\end{tabular}
\end{center}
\begin{proof}
The plethysm coefficient computations were performed with the \textsc{LiE} software.
Lemma~\ref{lem:chowupperbound} implies that $\mult_\la(\IC[\Ch_m^n]_d) \leq \pl \la n d = 0$.
Proposition~\ref{pro:lowdegreepowersums} implies $\mult_\la(\IC[\Pow_{m,d}^n]_d)>0$.
\end{proof}

See \cite[Prop.~4]{bhi:15} for additional occurrence obstructions in the case $n=3$.

\section{Plethysm inequalities}
We are interested in the plethysm coefficients $a_\la(d[m])$ for certain values of $\la$ and $d,m$. Here we compute such values for infinite families of parameters and in particular, prove Theorem~\ref{thm:plethineq}.

We will work over the ring of symmetric functions $\Lambda$, defined as the ring of formal power series (in finitely or infinitely many variables) which are invariant under any transposition of the variables. For the definitions and main identities see e.g.~\cite{sta:99}. Plethysms of symmetric functions are described also there in Appendix 2 of Chapter 7, here we review the necessary definitions. 

The characters of the irreducible $GL_r$--module $W_\lambda$ are the Schur functions $s_\lambda(x_1,\ldots,x_r)$, where $x_1,\ldots,x_r$ correspond to the eigenvalues of the conjugacy class representative from $GL_r$. Their combinatorial interpretation is as the generating function over all semi-standard Young tableaux with entries $1,...,r$, but we will use certain determinantal formulas as described below. The complete homogeneous symmetric functions $h_\ell$ are defined as $s_{(\ell)}$ and are the characters of the $\Sym^\ell$ module. The $\Sym^d(\Sym^n(\mathbb{C}^r))$ module is obtained as the composition of the two representations. The image in $\Sym^n(\mathbb{C}^r)$ of a diagonal matrix from $GL_r$ with entries (i.e. eigenvalues) $x_1,\ldots,x_r$ on the diagonal has eigenvalues all the $N:=\binom{n+r-1}{r-1}$ degree $n$ monomials in $x_1,\ldots,x_r$. Hence, the character of the representation $\Sym^d(\Sym^n(\mathbb{C}^r))$ of $GL_{N}$ can be obtained by evaluating the character $h_d$ of $\Sym^d$ at the monomials, i.e. the eigenvalues above.  This gives us the definition of the \emph{symmetric function plethysm} $h_d[h_n(x_1,\ldots,x_r)]$, that is, the evaluation of $h_d$ on the variables consisting of all degree $n$ monomials, i.e.
$$h_d[h_n(x_1,\ldots,x_r)] := h_d(x_1^n,x_1^{n-1}x_2,x_1^{n-1}x_3,\ldots,x_1^{\alpha_1}\cdots x_r^{\alpha_r},\ldots),$$
where $\alpha=(\alpha_1,\ldots,\alpha_r)$ runs over all compositions of $n$. 

In general, knowing the character of a representation contains all the information to obtain the multiplicities of the irreducible decomposition via the inner product of characters. As the Schur functions $s_\lambda$ are the irreducible characters for $GL_r$, the inner product is equivalent to an inner product in the ring $\Lambda$, where $\{ s_{\lambda}\}_{\lambda}$ is an orthonormal basis. 
In other words, the multiplicity of the Weyl module of weight $\lambda$ is given by the multiplicity of the Schur function $s_\lambda$ in the expansion of $h_d[h_m]$. We will now compute this via the inner product in the ring $\Lambda$ of symmetric functions, using some basic properties of this ring as found in~\cite{sta:99} and~\cite{macdonald:95}.

We have that $a_\la(d[n])$ is the multiplicity of $\{\lambda\}$ in $\Sym^d\Sym^n$, translated into characters this is also the coefficient at $s_\lambda$ of the expansion of $h_d[h_n]$ in Schur function. By their orthonormalitiy, this is the same as
\begin{equation}\label{eq:a_inner}
a_\la(d[n])=\langle s_\la,h_d[h_n]\rangle
\end{equation}
We now invoke various symmetric function identities in order to compute the above inner product.
The Schur functions $s_\la$ can be expressed via the Jacobi-Trudi formula (see again~\cite[Ch. 7]{sta:99}) as a signed sums of homogeneous symmetric functions, namely
\begin{equation}\label{eq:Jacobi-Trudi}
s_\la = \det\left[ h_{\la_i-i+j} \right]_{i,j=1}^{\ell(\la)},
\end{equation}
the inner product \eqref{eq:a_inner} can then be computed via a signed sum of inner products of the form $\langle h_\mu,h_d[h_n]\rangle$.
We remark that the orthogonal dual basis for the complete homogeneous symmetric functions  is the monomial symmetric functions, i.e. 
$$\langle h_\mu,m_\nu\rangle = \delta_{\mu,\nu},$$ so we need to express $h_d[h_n]$ in terms of the monomial symmetric functions, defined by 
$$m_\nu(x_1,\ldots,x_r):= \sum_{\sigma \in S_r(\nu)} x_1^{\nu_{\sigma(1)}}x_2^{\nu_{\sigma(2)}}\cdots x_r^{\nu_{\sigma(r)}},$$
where the sum ranges over all distinct permutations of $(\nu_1,\nu_2,\ldots,\nu_r)$ and $\nu$ is completed with 0s to the length $r$. 
Since he monomial symmetric functions form a basis for $\Lambda$, we can expand any symmetric function in it uniquely. 
Let
$$h_d[h_m] = \sum_{\nu} c_\nu m_\nu$$, for some constants $c_\nu$ (i.e. the coefficients in this expansion). Since each $m_\nu$ has a unique leading monomial (in the lexicographic order) $x_1^{\nu_1}x_2^{\nu_2}\cdots$, finding $c_\nu$ is equivalent to extracting the coefficient at the single monomial $x_1^{\nu_1}\cdots$ from the monomial expansion of the corresponding symmetric function as a polynomial, i.e.
$$c_\nu = (x_1^{\nu_1} x_2^{\nu_2}\ldots) @ h_d[h_n(x_1,x_2,\ldots)],$$ 
where to avoid confusion with the plethysm notation we denote by $(X)@f$ the coefficient of the monomial $X$ in the monomial expansion of the polynomial $f$. 

Let $\nu$ be a partition of length $\ell$. By the above remarks we need to consider only the truncated expansion $h_d[h_n(x_1,\ldots,x_\ell)]$ as only the monomials in $x_1,\ldots,x_\ell$ will be relevant.

We have the following formula for the $h$'s, see e.g. \cite{sta:99}:

\begin{equation*}
h_N(x_1,\ldots,x_r) = \sum_{(b): b_1+b_2+\cdots = N} x_1^{b_1}x_2^{b_2}\cdots,
\end{equation*}
where $(b)=(b_1,b_2,\ldots,b_r)$ runs over all (weak) compositions of $N$.
Hence, assuming some total ordering for compositions $\alpha^i $ of $n$, we have
\begin{equation*}
h_d[h_n(x_1,\ldots,x_r)] = h_d[ \ldots, x^{\al^i},\ldots ] = \sum_{(b): |b|=d}  x^{\sum_i b_i \alpha^i} 
\end{equation*}
Thus for the coefficients $c_\nu$ we have:
\begin{equation}\label{eq:mon_coef}
c_\nu(d,n) := (x^\nu)@ h_d[h_n] = \langle h_\nu,h_d[h_n]\rangle = \# \{ (b): |b| = d, \sum_i b_i \alpha^i = \nu\}
\end{equation}
By the Jacobi-Trudi identity \eqref{eq:Jacobi-Trudi} this gives a formula for computing the plethysm coefficients as
\begin{equation}\label{eq:pleth_det}
a_\la(d[n])= \langle \det\left[ h_{\la_i-i+j} \right]_{i,j=1}^{\ell(\la)}, h_d[h_n] \rangle = \sum_{\pi \in S_{\ell(\la)} } sgn(\pi) c_{ \la +\pi - (1,2,\ldots)}(d,n),
\end{equation}
where the permutations $\pi$ are viewed as vectors with entries $1,2,\ldots,\ell(\la)$

We now turn towards the proof of Theorem~\ref{thm:plethineq} and consider $s_\la$ for $\la = (\la_1,\la_2,2)$ for some $k\geq 2$. By the Jacobi-Trudi identity~\eqref{eq:pleth_det} we need to compute only $c_\nu$ for $\nu$ having at most 3 parts, with $\nu_3 = 0,1,2$. 
Let $p_r(a,b)$ denote the number of partitions of $r$ which fit inside an $a\times b$ rectangle, it's generating function is the $q$-binomial coefficient (see~\cite{sta:11}):
$$\binom{a+b}{a}_q = \frac{ (1-q) \cdots (1-q^{a+b})}{(1-q)\cdots(1-q^a) (1-q)\cdots (1-q^b)}= \sum_{r=0}^{ab} p_r(a,b)q^r$$

\begin{proposition}\label{prop:c_coeffs}
We have the following generating function identities for $c_\nu(d,n)$, where $\ell(\nu)\leq 3$ and $\nu_3 \leq 2$:

$c_{(L,k,2)} = (q^k)@ \left(  \binom{n}{1}_q\binom{n+d-2}{n}_q +  \binom{n-1}{1}_q \binom{n+d-1}{n}_q +q \binom{n}{2}_q \binom{n+d-2}{n}_q\right) $

$c_{(L,k,1)} = (q^k)@\binom{n}{1}_q\binom{n+d-1}{n}_q$

$c_{(L,k,0)} =(q^k)@ \binom{n+d}{n}_q =p_k(n,d)$
 
\end{proposition}

\begin{proof}
By formula \eqref{eq:mon_coef}, we have the following
$$c_{(L,k,0)} = \# \{ (b): |b|=d, \sum b_i \alpha^i = (L,k)\}$$

Hence, the only $\alpha^i$ involved are of the form $\alpha^i=(n-a_i,a_i)$, and after renumerating, we can assume $a_i=i$. So we are counting compositions $b$ of $d$, s.t. $\sum_i b_i i =k$ for $i=0\ldots n$. This is exactly the same as specifying an integer partition $\gamma$ of $k$ by the number of its parts, i.e. $\gamma=(0^{b_0},1^{b_1},\ldots,n^{b_n})$, such that $b_0+\cdots+b_n=d$. These restrictions are equivalent to $\gamma$ fitting inside an $n \times d$ box, and the number of such $\gamma$ is exactly $p_k(d,n)$.

Next, when the second part in $\nu$ is 1, we have the following. Since $\nu_3=1$, the condition
$\sum_i b_i \alpha^i_3 = \nu_3=1$ implies that there is a single $i$, such that $b_i \alpha^i_3 \neq 0$, and in fact must be 1, so $b_i=\alpha^i_3 =1$. After renumeration, we can assume that $i=0$ (for separation purposes) with $b_0=\alpha^0_3=1$ and $\alpha^0=(n-1-r,r,1)$ for $r=0\ldots n-1$. For the remaining $b$'s and $\alpha$s we have 
$$\sum_i b_i \alpha^i = (L,k) - (n-1-r,r) = (L+r-n+1,k-r)$$
 with $b_1+\cdots = d-1$, and $|\alpha^i|=n$. This number is now, by the previous case, $(q^{k-r})@\binom{n+d-1}{n}_q$. The total number is thus
 
$$c_{(L,k,1)} = \sum_{r=0}^{n-1} (q^{k-r})@ \binom{n+d-1}{n}_q = (q^{k})@ \sum_{r=0}^{n-1}q^r \binom{n+d-1}{n}_q = (q^k)@ \binom{n}{1}_q\binom{n+d-1}{n}_q.$$

Finally, when $\nu_3=2$ we have the following two distinct options:

Either there is an index $i$, such that $b_i \alpha^i_3 =\nu_3=2$, or $i<j$ with $b_i\alpha^i_3=1$ and $b_j\alpha^j_3=1$.

In the first case we have $b_i \alpha^3_i=2$ -- either $b_i=2$, in which case $\alpha^i = (n-1-r,1)$ and the rest of the $b$'s sum to $d-b_i=d-2$, which brings us to the previous case (of $(L,k,1)$), so the number is $$(q^k)@ \binom{n}{1}_q\binom{n+d-2}{n}_q.$$

Otherwise, $b_i=1$ and $\alpha^i_3=2$. As in the case $\nu_3=1$, let $i=0$ and $\alpha^0 = (n-2-r,r,2)$, $b_0=1$, so we are looking for the number of $(b_1,\ldots)$ with $|b|=d-1$ and such that 
$\sum_i b_i \alpha^i = (L-n+r+2,k-r)$ for all possible $r=0,\ldots,n-2$. 
So this is $$\sum_{r=0}^{n-2} (q^{k-r})@ \binom{n +d-1}{n}_q =(q^k)@ \binom{n-1}{1}_q \binom{n+d-1}{n}_q$$

Last, when there are $i<j$ with $b_i \alpha^i_3=1$ and $b_j \alpha^j_3 =1$, let $i=-1$, $j-0$ (again, renumerating for simplicity), with $\alpha^{-1} = (n -1 -r_1,r_1,1)$ and $\alpha^0 = (n-1-r_2,r_2,1)$ with $0\leq r_1<r_2\leq n-1$. We thus have for the remaining $\alpha$ and $b$s that $b_1+\cdots = d-2$, and $\sum_i b_i \alpha^i = (L - (n-1-r_1) - (n-1-r_2), k-r_1-r_2)$. By the first case, this is 

$(q^{k-r_1-r_2})@ \binom{n+d-2}{n}_q$. Summing over all possible $0<r_1<r_2\leq n-1$, we have
\begin{eqnarray*}
(q^k)\at \sum_{0 \leq r_1 < r_2 \leq n-1} q^{r_1+r_2} \binom{n+d-2}{n}_q &=& (q^k)@ q \sum_{0 \leq r_1 \leq r_2-1 \leq n-2} q^{r_1+(r_2-1)} \binom{n+d-2}{n}_q\\
&=& (q^k)@ q \binom{n-2+2}{2}_q \binom{n+d-2}{n}_q,
\end{eqnarray*}
where the last identity follows from interpreting $(r_2-1,r_1)$ as a partition in the $2 \times n-2$ rectangle. Summing over all the cases considered here, we get the desired total coefficient.
\end{proof}

\begin{proposition}\label{prop:pleth_3row}
The plethysm coefficient for $\la = (L,r,2)$ is equal to
\begin{align*}
a_\la(d[n]) =(q^{r+1})\at \bigg(  \binom{n+d-2}{n}_q \frac{q(1-q^n)(1-q^2+q-q^n)}{1-q^2} \\
+ \binom{n+d-1}{n}_q (q^{n+1}-1)  + (1-q)\binom{n+d}{n}_q \bigg)
\end{align*}
\end{proposition}

\begin{proof}
Following equation~\eqref{eq:pleth_det}, we have that
$$a_\la(d[n]) = c_{(L,r,2)} - c_{(L,r+1,1)} - c_{(L+1,r-1,2)}+c_{(L+1,r+1,0)} +c_{(L+2,r-1,1)}-c_{(L+2,r,0)}.$$
Substituting the formulas for the $c$'s from Proposition~\ref{prop:c_coeffs}, and observing that $(q^{r+j})\at f = (q^r)\at q^{-j}f$ for any $j$, we have that 
\begin{align*}
a_{\la}(d[n]) = (q^{r+1})\at \bigg( (q-q^2) \left(  \binom{n}{1}_q\binom{n+d-2}{n}_q +  \binom{n-1}{1}_q \binom{n+d-1}{n}_q +q \binom{n}{2}_q \binom{n+d-2}{n}_q\right) \\
+ (q^2-1)\binom{n}{1}_q\binom{n+d-1}{n}_q +(1-q)\binom{n+d}{n}_q \bigg)
\end{align*}
Simplifying the above expression by grouping terms for the same binomial coefficients together we obtain
\begin{align*}
a_\la(d[n]) =(q^{r+1})\at  \bigg(  \binom{n+d-2}{n}_q \frac{q(1-q^n)(1-q^2+q-q^n)}{1-q^2} \\
+ \binom{n+d-1}{n}_q (q^{n+1}-1)  + (1-q)\binom{n+d}{n}_q \bigg)
\end{align*}
\end{proof}

\begin{proposition}\label{prop:difference}
Let $\la = (L,r,2)$. We have that 
$$a_\la(d[n]) - a_\la(n[d]) = (q^r)\at \binom{n+d-2}{n-1}_q(q^n-q^d) \frac{(1-q^{d-1})(1-q^{n-1})}{(1-q^d)(1-q^n)}.$$
\end{proposition}

\begin{proof}
Set $[a]!_q:=(1-q)\cdots(1-q^a)$, a  variant of the usual factorial $q$-analogue but multiplied by $(1-q)^a$, and consider the desired difference via the formula in Proposition~\ref{prop:pleth_3row}:
\begin{align*}
a_\la(d[n]) - a_\la(n[d]) &= (q^{r+1}) \at \\
&\bigg\{ \; \bigg(  \binom{n+d-2}{n}_q \frac{q(1-q^n)(1-q^2+q-q^n)}{1-q^2} 
+ \binom{n+d-1}{n}_q (q^{n+1}-1)  \\
&\quad -  \binom{n+d-2}{d}_q \frac{q(1-q^d)(1-q^2+q-q^d)}{1-q^2} 
- \binom{n+d-1}{d}_q (q^{d+1}-1) \bigg)\\
&\quad = \frac{ [n+d-2]!_q}{[n-2]!_q[d-2]!_q} \frac{q(q^{n-1}- q^{d-1})}{(1-q^{n-1})(1-q^{d-1})}
- \frac{ [n+d-1]!_q}{[n-1]!_q[d-1]!_q} \frac{(1-q)(q^n-q^d)}{(1-q^n)(1-q^d)} \\
&\quad = \binom{n+d-2}{n-1}_q (q^n-q^d) \left( 1  - \frac{(1-q^{n+d-1})(1-q)}{(1-q^n)(1-q^d)} \right)\\
&\quad=\binom{n+d-2}{n-1}_q (q^n-q^d) \frac{ q(1-q^{d-1})(1-q^{n-1})}{(1-q^d)(1-q^n)}  \; \bigg \}
\end{align*}
Finally, observe that the RHS is a polynomial divisible by $q$, so the coefficient at $q^{r+1}$ is the same as the coefficient at $q^{r}$ after dividing by $q$. 
\end{proof}

We are now ready to prove Theorem~\ref{thm:plethineq} as a Corollary of the above computations:

\begin{corollary}\label{cor:keyinequality} [Theorem~\ref{thm:plethineq}]
Let  $d=n+1$ and $\la = (n^2+n-2-r,r,2)$. Then
$a_\la((n+1)[n]) - a_\la(n[n+1]) \geq 0,$
with $$a_\la((n+1)[n]) - a_\la(n[n+1]) = \begin{cases} 0, &\text{ when }r<n, \\ 1, & \text{ when } r=n,\\ >0, &\text{ when $r>n$ and $n\geq 7$}, \\
\end{cases}$$
with the exception in the last case when $n=8$, and $r=35$ when $a_{(35,35,2)}(9[8])=a_{(35,35,2)}(8[9])$.
\end{corollary}

\begin{proof}
Then by the Proposition~\ref{prop:difference} we have
\begin{align*}
a_\la(n+1[n]) - a_\la(n[n+1]) = (q^{r})\at  \binom{ 2n-1}{n-1}_q (q^{n}-q^{n+1}) \frac{ (1-q^n)(1-q^{n-1})}{(1-q^{n+1})(1-q^n)}\\
=(q^{r})\at \binom{2n-1}{n-1}_q (q^{n} - q^{n+1}) \frac{(1-q^{n-1})}{1-q^{n+1}} = (q^{r})\at 
\binom{2n-1}{n-2}_q (q^{n}-q^{n+1}) 
\end{align*}

The last line follows by absorbing the fraction into the $q$-binomial coefficient. 
It is now evident, that since the $q$-binomial coefficient expands into a polynomial of $q$ (with coefficients given by $p_*(n-2,n+1)$), multiplying it with $q^{n}$ or $q^{n+1}$ gives two polynomials whose lowest order terms are $q^{n}$ and $q^{n+1}$ respectively.
So if $r<n$, there is no term of such degree, and the coefficient is 0. when $r=n$ we see that
 such term can only come from the first polynomial's first (lowest order) term, which is exactly $q^{n}$ since $\binom{2n-1}{n-2}_qq^{n} = q^{n}(1+q+2q^2+\cdots)=q^n+ O(q^{n+1})$. Therefore we obtain the case $r=n$.
 
 Let now $r>n$, and set $r= n + k+1$ for some $k \geq 0$. 
 We have that 
 \begin{align*}
 a_\la((n+1)[n]) - a_\la(n[n+1]) = (q^{k+1})\at \binom{ 2n-1}{n-2}_q - (q^{k})\at \binom{2n-1}{n-2}_q \\
 = p_{k+1}(n+1,n-2) - p_k(n+1,n-2) 
 \\=g((n^2-n-3-k,k+1),(n+1)^{n-2},(n+1)^{n-2})>0,
 \end{align*}
 where $g$ denotes the Kronecker coefficient for the symmetric group $\mathfrak{S}_n$ for the 3 given partititons, and the last identity and the strict positivity are shown to hold for $n\geq 9$ in~\cite{pak-panova:13}, and the other cases are verified by direct expansion of the $q$-binomial coefficients. In particular, we have that $p_{26}(9,6) =227 = p_{27}(9,6)$ which gives the only exceptional 0 plethysm. 
\end{proof}

\section{Computer calculations}\label{sec:computer}
The following computer calculation for Proposition~\ref{pro:computercalc} is a refinement and speedup of the computation performed in \cite{cim:15}.
Indeed, a run of the method from \cite{cim:15} would take significantly too long to prove Proposition~\ref{pro:computercalc} in any reasonable time.
Our new method makes extensive use of memory resources, while the method from \cite{cim:15} uses almost no memory.

\begin{proposition}\label{pro:computercalc}
If $X$ is defined as in Proposition~\ref{pro:listofgenerators}, then
for all $\mu \in X$ of length 3 we have $\mult_{\mu}(\IC[\Ch_3^6])>0$.

If $X$ is defined as in Proposition~\ref{pro:listofgenerators74}, then
for all $\mu \in X$ we have $\mult_{\mu}(\IC[\Ch_4^7])>0$.
\end{proposition}
\begin{proof}
For a vector space $U$ let $\tensor^\delta U$ denote its $\delta$-th tensor power.
The computation that verifies Proposition~\ref{pro:computercalc}
is based on the famous Schur-Weyl duality (see e.g.~\protect\cite[Ch.~9, eq.~(3.1.4)]{Pro:07} or \protect\cite[eq.~(9.1)]{gw:09}):
\begin{equation}\label{eq:schurweyl}
\tensor^{dn} \IC^m \simeq \bigoplus_{\la} \{\la\} \otimes [\la],
\end{equation}
where the sum is over all $m$-partitions of $dn$,
$\{\la\}$ is the irreducible $\GL_m$-representation of type $\la$,
and $[\la]$ is the irreducible $\aS_{dn}$-representation of type $\la$, which is called the Specht module.

For a $\GL_m$-representation $W$,
a highest weight vector of type $\la$ is a vector $f \in W$ such that two properties hold:
(1) $\diag(\alpha_1,\alpha_2,\ldots,\alpha_m) f = \alpha_1^{\la_1}\cdots\alpha_m^{\la_m} f$,
where $\diag(\alpha_1,\alpha_2,\ldots,\alpha_m)$ denotes the diagonal matrix with $\alpha_i$ on the main diagonal,
and (2) $g f = f$ for every upper triangular matrix $g \in \GL_m$ with 1s on the main diagonal.
The highest weight vectors of type $\la$ in $W$ form a vector space, which we call $\HWV_\la(W)$.
Its dimension conveniently coincides with the multiplicity of $\la$ in $W$ (see e.g.~\cite[Prop~12.2.5]{BI:18July25}):
\begin{equation}\label{eq:multdimHWV}
\mult_\la(W) = \dim \HWV_\la(W).
\end{equation}
Let $\Sym^\delta U \subseteq \tensor^\delta U$ denote the $\aS_\delta$-invariant subspace
and let $U^*$ denote the vector space dual to $U$.
There are canonical isomorphisms $\IC[U]_d \simeq \Sym^d (U^*) \simeq (\Sym^d U)^*$.
Observe that there are canonical $\GL_m$-equivariant surjections:
\[
\tensor^{dn} \IC^m \simeq \tensor^{d}(\tensor^n \IC^{m*})^* \twoheadrightarrow \underbrace{\Sym^d(\Sym^n(\IC^{m*}))^*}_{=\IC[\IA_m^n]_d} \twoheadrightarrow \IC[\Ch_m^n]_d,
\]
where the first surjection is the symmetrization $S_{d,n} := \frac{1}{d!(n!)^d}\sum_{\sigma\in\aS_n \wr \aS_d}\sigma$ over the wreath product $\aS_n \wr \aS_d$
and the second surjection is the restriction of functions from $\IA_m^n$ to the subvariety $\Ch_m^n$.
Now, restricting to the highest weight vector space of type $\la$, we obtain surjections
\begin{equation}\label{eq:HWVsurjections}
\HWV_\la(\tensor^{dn} \IC^m) \twoheadrightarrow \HWV_\la(\IC[\IA_m^n]_d) \twoheadrightarrow \HWV_\la(\IC[\Ch_m^n]_d).
\end{equation}

To prove that $\mult_\la(\IC[\Ch_m^n]_d)>0$ we combine \eqref{eq:multdimHWV} with \eqref{eq:HWVsurjections},
so our goal is to find a nonzero vector in $\HWV_\la(\tensor^{dn} \IC^m)$
that does not vanish under the composition of both surjections in \eqref{eq:HWVsurjections}.
The vector space $\HWV_\la(\tensor^{dn} \IC^m)$ is well known and we construct its elements as follows.
For a partition $\la = (\la_1,\ldots,\la_m)$ let $\mu$ be its transposed partition, i.e., $\mu_i := |\{j \mid \la_i \geq j\}|$.
If we depict a partition by its Young diagram, which is a top-left justified array of boxes, $\la_i$ in each row, then $\mu$ is obtained by reflecting $\la$ at its diagonal, hence the name ``transposed partition''.
It is straightforward to verify that the following vector $v_\la$ is contained in $\HWV_\la(\tensor^{dn} \IC^m)$:
\[
v_\la := v_{\mu_1} \otimes v_{\mu_2} \otimes \cdots \otimes v_{\mu_{\la_1}},
\]
where $v_i := e_1 \wedge e_2 \wedge \cdots \wedge e_i := \frac{1}{i!}\sum_{\sigma \in \aS_i}\sgn(\sigma) e_{\sigma(1)} \otimes \cdots \otimes e_{\sigma(i)}$ is a highest weight vector of type $(\underbrace{1,1,\ldots,1}_{i \text{ times}})$.
Equation~\eqref{eq:schurweyl} combined with the fact that $\{\la\}$ contains a single highest weight vector line of type $\la$
and no highest weight vector of any other type (see e.g.~\cite[III.1.4, Satz 1]{Kra:85})
implies that $\{\pi v_\la \mid \pi \in \aS_{dn}\}$
is a generating set of the vector space $\HWV_\la(\tensor^{dn} \IC^m)$ (cp.~\cite[Claim 4.2.13]{ike:12b}).
Thus $\{S_{d,n} \pi v_\la \mid \pi \in \aS_{dn}\}$ is a generating set of
$\HWV_\la(\IC[\IA_m^n]_d)$.
The evaluation of $S_{d,n} \pi v_\la$ at a point $p \in \IA_m^n = \Sym^n\IC^{m*}$ is known to equal the tensor contraction
\begin{equation}\label{eq:contractI}
(p^{\otimes d}) S_{d,n} \pi v_\la,
\end{equation}
see e.g.~\cite[Sec.~4.2(A)]{ike:12b}.
Since $p^{\otimes d} S_{d,n} = p^{\otimes d}$, \eqref{eq:contractI} equals
\begin{equation}\label{eq:contractII}
(p^{\otimes d}) \pi v_\la,
\end{equation}
as observed e.g.\ in \cite[p.~39]{ike:12b}.
Therefore, the statement $\dim \HWV_\la(\IC[\Ch_m^n]_d)>0$
is equivalent to the existence of a $\pi \in \aS_{dn}$ and a $p \in \Ch_m^n$ such that \eqref{eq:contractII} is nonzero.
The search for $(\pi,p)$ is an algorithmic challenge that we tackle as follows.
We choose a random $p \in \Ch_m^n$, i.e., $p = \ell_1 \cdots \ell_n = \frac{1}{n!}\sum_{\sigma \in \aS_n} \ell_{\sigma(1)} \otimes \cdots \otimes \ell_{\sigma(i)}$, $\ell_i \in \IC^{m*}$.
To compute \eqref{eq:contractII} we aim to \emph{not} expand $\pi v_\la$,
because a tensor contraction 
\begin{equation}\label{eq:rankonecontraction}
(\ell_{i_1} \otimes \ell_{i_2} \otimes \cdots \otimes \ell_{i_{dn}}) \pi v_\la
\end{equation}
is just a product of determinants of matrices of size $\leq m$, and such determinants are efficiently computable.
If we expand $p^{\otimes d}$ into a sum of summands of the form $\ell_{i_1} \otimes \ell_{i_2} \otimes \cdots \otimes \ell_{i_{dn}}$,
then we obtain $(n!)^d$ many summands, which for $n=6$ quickly exceeds our computational recources, even for reasonably low $d$.
Therefore we use a \emph{dynamic programming} approach that is based on the
combinatorial interpretation of the summation \eqref{eq:contractII} from \cite{cim:15}, which we describe next.

We identify $\la$ with its Young diagram, which we interpret as a cardinality $|\la|$ subset of $\IN\times\IN$.
Let $\mu$ be the transpose of $\la$, so $\mu_i$ denotes the length of the $i$-th column of $\la$.
A \emph{placement $\vartheta$ on $\la$} is a map $\la \to \{1,\ldots,n\}$.
To each placement $\vartheta$ and each column index $i$ the corresponding determinant $\det_{\vartheta,i}$ is defined as the
determinant of the top $\mu_i \times \mu_i$ submatrix of the $m \times \mu_i$ matrix that is given by the linear forms
$\ell_{\vartheta(1,i)},\ell_{\vartheta(2,i)},\ldots,\ell_{\vartheta(\mu_i,i)}$.
We number the positions in $\la$ columnwise from left to right, top to bottom, so that each position $b\in\la$ gets a number $j(b) \in \{1,\ldots,nd\}$.
Given $\pi\in\aS_{nd}$ we construct the Young tableau $T$ of shape $\la$ by filling $\la$ at position $b$ with the number $T(b) := \lceil \pi(j(b))/n\rceil$.
For example, if $\la=(2,2)$, $d=n=2$, and $\pi$ is the transposition $(2 \ 3)$,
then $T = {\Yvcentermath1\tiny\young(11,22)}$.
The tableau $T$ contains each number from $\{1,\ldots,d\}$ exactly $n$ times (and each tableau $T$ which contains each number from $1$ to $d$ exactly $n$ times can be otained from some $\pi$).
A short calculation, which for example is done in \cite{cim:15}, shows that \eqref{eq:contractII} equals
\begin{equation}\label{eq:summation}
\sum_{\text{proper } \vartheta} \ \prod_{i=1}^{\la_1}\det_{\vartheta,i},
\end{equation}
where a placement $\vartheta$ is \emph{proper} if $\vartheta$ places each number onto each number in $T$ exactly once,
i.e., for every $i \in \{1,\ldots,n\}$ and $j \in \{1,\ldots,d\}$ there exists exactly one $b \in \la$ with $(\vartheta(b),T(b))=(i,j)$.

The above description was used to perform the computations in \cite{cim:15}.
We now discuss some adjustments for the computations that we use.
For a placement $\vartheta$ on $\la$ and for a partition $\nu \subseteq \la$
we denote by $\vartheta|_\nu$ the restriction of $\vartheta$ to $\nu$.
We say that a placement $\vartheta$ on $\la$ \emph{extends} a placement $\psi$ on $\nu \subseteq \la$
iff $\vartheta|_\nu = \psi$.
Let $\la^{\leq k}$ denote the set of boxes in the first $k$ columns of the Young diagram $\la$,
and $\la^{> k}$ denote the set of boxes of $\la$ that are in columns $>k$.
For a placement $\psi$ on $\la^{\leq k}$ and a placement $\varphi$ on $\la^{>k}$ let $\psi\varphi$ denote the unique placement on $\la$ that extends both $\psi$ and $\varphi$.

Our algorithm constructs all proper placements $\vartheta$ on $\la$ using a standard breadth-first search in a columnwise manner from left to right, top to bottom.
In this way we first obtain $\vartheta|_{\la^{\leq 1}}$, then $\vartheta|_{\la^{\leq 2}}$, and so on.
For each placement $\psi$ on $\la^{\leq k}$ we observe that
\begin{equation}\label{eq:summationII}
\sum_{\text{proper } \vartheta \text{ extending } \psi} \ \prod_{i=1}^{\la_1}\det_{\vartheta,i} =
\left(\prod_{i=1}^{k}\det_{\psi,i} \right)\underbrace{\left( \sum_{\text{proper } \vartheta \text{ extending } \psi} \ \prod_{i=k+1}^{\la_1}\det_{\vartheta,i}\right)}_{=: \alpha(\psi)}.
\end{equation}
Two placements $\psi$ and $\psi'$ on $\la^{\leq k}$ are called \emph{equivalent} if they
can be obtained from each other by permuting entries between positions that have the same number in $T$.
The crucial observation is that $\alpha(\psi) = \alpha(\psi')$ if $\psi$ is equivalent to $\psi'$.
Therefore we can store and reuse each $\alpha(\psi)$ that we encounter throughout the algorithm without computing it again.
Although this requires a significant amount of memory, it enables us to crucially cut down the computation time.
\footnote{Since $\alpha(\psi)$ for any placement $\psi$ on $\la^{\leq k}$ only depends on $\alpha(\psi')$ for some placements $\psi'$ on $\la^{\leq k+1}$ we don't have to keep the results for all $k$ at the same point in time, which again cuts down our memory usage by a small factor.
This was especially relevant for the case $m = 4, n = 7$, although not sufficient for all partitions.}

Define $\kappa_i(T^{\leq k})$ as the number of times the number $i$ appears in the first $k$ columns of $T$.
Then the number of non-equivalent $\psi$ can easily be calculated as
\begin{equation}\label{eq:cachesize}
    \sum_{k=1}^{\la_1}\prod_{i=1}^{d}\binom{n}{\kappa_i(T^{\leq k})}
\end{equation}
Note that depending on $T$ the number of non-equivalent $\psi$ can wildly vary and with it also our running time and memory usage.
Generally, it seems that semistandard tableaux, i.e., tableaux with non-decreasing rows and strictly increasing columns, can be evaluated faster.
It is sufficient to restrict our attention to semistandard $T$ only, see e.g.\ \cite[Sec.~4.3(A)]{ike:12b}.
For many partitions there were too many semistandard tableaux in order to generate all of them, so a set of random semistandard tableaux was chosen in these cases.
In either case we tested the chosen tableaux by increasing value of $\eqref{eq:cachesize}$.
Additionally, as soon as we find $\prod_{i=1}^{k}\det_{\psi,i}$ to be zero we do not have to evaluate the corresponding $\alpha(\psi)$.
Numerical problems were avoided by working over a finite field.
\end{proof}
The actual tableaux found in both these calculations can be found in Section~\ref{sec:computer_resultscomputercalc}.

\section{Tableau computation results for Proposition~\ref{pro:smallcomputercalc}} \label{sec:computer_resultssmallcomputercalc}
\begin{proof}[Proof of Proposition~\ref{pro:smallcomputercalc}]
We computed the following 8 tableaux that index a basis of $\HWV_{34,6,2}(\IC[\IA_3^6]_7)$,
in complete analogy to \cite[Sec.~6]{BIP:19}.
The corresponding functions can be readily evaluated at 8 random points in $\Pow_{3,4}^6$ to obtain an $8 \times 8$ matrix whose non-singularity proves the first part of Proposition~\ref{pro:smallcomputercalc}.

\begin{longtable}{cc}
{\def\lr#1{\multicolumn{1}{|@{\hspace{.6ex}}c@{\hspace{.6ex}}|}{\raisebox{-.3ex}{$#1$}}}
\raisebox{-.6ex}{\tiny$\begin{array}[b]{*{34}c}\cline{1-34}
\lr{1}&\lr{1}&\lr{1}&\lr{1}&\lr{1}&\lr{1}&\lr{2}&\lr{2}&\lr{2}&\lr{2}&\lr{2}&\lr{3}&\lr{3}&\lr{4}&\lr{4}&\lr{4}&\lr{4}&\lr{5}&\lr{5}&\lr{5}&\lr{5}&\lr{5}&\lr{6}&\lr{6}&\lr{6}&\lr{6}&\lr{6}&\lr{6}&\lr{7}&\lr{7}&\lr{7}&\lr{7}&\lr{7}&\lr{7}\\\cline{1-34}
\lr{2}&\lr{3}&\lr{3}&\lr{3}&\lr{4}&\lr{5}\\\cline{1-6}
\lr{3}&\lr{4}\\\cline{1-2}
\end{array}$}
} \\
{\def\lr#1{\multicolumn{1}{|@{\hspace{.6ex}}c@{\hspace{.6ex}}|}{\raisebox{-.3ex}{$#1$}}}
\raisebox{-.6ex}{\tiny$\begin{array}[b]{*{34}c}\cline{1-34}
\lr{1}&\lr{1}&\lr{1}&\lr{1}&\lr{1}&\lr{1}&\lr{2}&\lr{2}&\lr{2}&\lr{2}&\lr{3}&\lr{3}&\lr{3}&\lr{3}&\lr{3}&\lr{3}&\lr{4}&\lr{4}&\lr{4}&\lr{4}&\lr{5}&\lr{5}&\lr{5}&\lr{5}&\lr{6}&\lr{6}&\lr{6}&\lr{6}&\lr{7}&\lr{7}&\lr{7}&\lr{7}&\lr{7}&\lr{7}\\\cline{1-34}
\lr{2}&\lr{2}&\lr{4}&\lr{4}&\lr{5}&\lr{5}\\\cline{1-6}
\lr{6}&\lr{6}\\\cline{1-2}
\end{array}$}
} \\
{\def\lr#1{\multicolumn{1}{|@{\hspace{.6ex}}c@{\hspace{.6ex}}|}{\raisebox{-.3ex}{$#1$}}}
\raisebox{-.6ex}{\tiny$\begin{array}[b]{*{34}c}\cline{1-34}
\lr{1}&\lr{1}&\lr{1}&\lr{1}&\lr{1}&\lr{1}&\lr{2}&\lr{2}&\lr{2}&\lr{3}&\lr{3}&\lr{3}&\lr{3}&\lr{3}&\lr{3}&\lr{4}&\lr{4}&\lr{4}&\lr{4}&\lr{5}&\lr{5}&\lr{5}&\lr{6}&\lr{6}&\lr{6}&\lr{6}&\lr{6}&\lr{6}&\lr{7}&\lr{7}&\lr{7}&\lr{7}&\lr{7}&\lr{7}\\\cline{1-34}
\lr{2}&\lr{2}&\lr{2}&\lr{5}&\lr{5}&\lr{5}\\\cline{1-6}
\lr{4}&\lr{4}\\\cline{1-2}
\end{array}$}
} \\
{\def\lr#1{\multicolumn{1}{|@{\hspace{.6ex}}c@{\hspace{.6ex}}|}{\raisebox{-.3ex}{$#1$}}}
\raisebox{-.6ex}{\tiny$\begin{array}[b]{*{34}c}\cline{1-34}
\lr{1}&\lr{1}&\lr{1}&\lr{1}&\lr{1}&\lr{1}&\lr{2}&\lr{2}&\lr{2}&\lr{2}&\lr{2}&\lr{3}&\lr{3}&\lr{3}&\lr{3}&\lr{3}&\lr{3}&\lr{4}&\lr{4}&\lr{5}&\lr{5}&\lr{5}&\lr{5}&\lr{6}&\lr{6}&\lr{6}&\lr{6}&\lr{6}&\lr{6}&\lr{7}&\lr{7}&\lr{7}&\lr{7}&\lr{7}\\\cline{1-34}
\lr{2}&\lr{4}&\lr{4}&\lr{4}&\lr{4}&\lr{7}\\\cline{1-6}
\lr{5}&\lr{5}\\\cline{1-2}
\end{array}$}
} \\
{\def\lr#1{\multicolumn{1}{|@{\hspace{.6ex}}c@{\hspace{.6ex}}|}{\raisebox{-.3ex}{$#1$}}}
\raisebox{-.6ex}{\tiny$\begin{array}[b]{*{34}c}\cline{1-34}
\lr{1}&\lr{1}&\lr{1}&\lr{1}&\lr{1}&\lr{1}&\lr{2}&\lr{2}&\lr{3}&\lr{3}&\lr{3}&\lr{3}&\lr{3}&\lr{4}&\lr{4}&\lr{4}&\lr{4}&\lr{5}&\lr{5}&\lr{5}&\lr{5}&\lr{5}&\lr{5}&\lr{6}&\lr{6}&\lr{6}&\lr{6}&\lr{6}&\lr{7}&\lr{7}&\lr{7}&\lr{7}&\lr{7}&\lr{7}\\\cline{1-34}
\lr{2}&\lr{2}&\lr{2}&\lr{2}&\lr{4}&\lr{4}\\\cline{1-6}
\lr{3}&\lr{6}\\\cline{1-2}
\end{array}$}
} \\
{\def\lr#1{\multicolumn{1}{|@{\hspace{.6ex}}c@{\hspace{.6ex}}|}{\raisebox{-.3ex}{$#1$}}}
\raisebox{-.6ex}{\tiny$\begin{array}[b]{*{34}c}\cline{1-34}
\lr{1}&\lr{1}&\lr{1}&\lr{1}&\lr{1}&\lr{1}&\lr{2}&\lr{2}&\lr{2}&\lr{2}&\lr{2}&\lr{3}&\lr{3}&\lr{3}&\lr{3}&\lr{3}&\lr{4}&\lr{4}&\lr{4}&\lr{4}&\lr{5}&\lr{5}&\lr{5}&\lr{5}&\lr{5}&\lr{6}&\lr{6}&\lr{6}&\lr{7}&\lr{7}&\lr{7}&\lr{7}&\lr{7}&\lr{7}\\\cline{1-34}
\lr{2}&\lr{3}&\lr{5}&\lr{6}&\lr{6}&\lr{6}\\\cline{1-6}
\lr{4}&\lr{4}\\\cline{1-2}
\end{array}$}
} \\
{\def\lr#1{\multicolumn{1}{|@{\hspace{.6ex}}c@{\hspace{.6ex}}|}{\raisebox{-.3ex}{$#1$}}}
\raisebox{-.6ex}{\tiny$\begin{array}[b]{*{34}c}\cline{1-34}
\lr{1}&\lr{1}&\lr{1}&\lr{1}&\lr{1}&\lr{1}&\lr{2}&\lr{2}&\lr{2}&\lr{2}&\lr{2}&\lr{2}&\lr{3}&\lr{3}&\lr{3}&\lr{3}&\lr{4}&\lr{4}&\lr{4}&\lr{4}&\lr{5}&\lr{5}&\lr{5}&\lr{5}&\lr{6}&\lr{6}&\lr{6}&\lr{6}&\lr{6}&\lr{7}&\lr{7}&\lr{7}&\lr{7}&\lr{7}\\\cline{1-34}
\lr{3}&\lr{3}&\lr{4}&\lr{5}&\lr{5}&\lr{7}\\\cline{1-6}
\lr{4}&\lr{6}\\\cline{1-2}
\end{array}$}
} \\
{\def\lr#1{\multicolumn{1}{|@{\hspace{.6ex}}c@{\hspace{.6ex}}|}{\raisebox{-.3ex}{$#1$}}}
\raisebox{-.6ex}{\tiny$\begin{array}[b]{*{34}c}\cline{1-34}
\lr{1}&\lr{1}&\lr{1}&\lr{1}&\lr{1}&\lr{1}&\lr{2}&\lr{2}&\lr{2}&\lr{2}&\lr{3}&\lr{3}&\lr{3}&\lr{3}&\lr{4}&\lr{4}&\lr{4}&\lr{4}&\lr{4}&\lr{5}&\lr{5}&\lr{5}&\lr{5}&\lr{5}&\lr{6}&\lr{6}&\lr{6}&\lr{6}&\lr{6}&\lr{7}&\lr{7}&\lr{7}&\lr{7}&\lr{7}\\\cline{1-34}
\lr{2}&\lr{2}&\lr{3}&\lr{3}&\lr{6}&\lr{7}\\\cline{1-6}
\lr{4}&\lr{5}\\\cline{1-2}
\end{array}$}
} \\
\end{longtable}

The second part is proved analogously by studying $\HWV_{47,7,2}(\IC[\IA_3^7]_8)$ and using the following 11 tableaux.

\begin{longtable}{c}
\raisebox{-.6ex}{\tiny\young(11111112222233333334444445555556666666777888888,2247778,57)}\\
\raisebox{-.6ex}{\tiny\young(11111112222333333344444455556666667777778888888,2225557,46)}\\
\raisebox{-.6ex}{\tiny\young(11111112223333333444444555555566666677777778888,2222468,88)}\\
\raisebox{-.6ex}{\tiny\young(11111112223333344444445555556666666777777788888,2222338,58)}\\
\raisebox{-.6ex}{\tiny\young(11111112222233333334444455555566666667777777888,2248888,45)}\\
\raisebox{-.6ex}{\tiny\young(11111112222222334444444555566666667777778888888,3333355,57)}\\
\raisebox{-.6ex}{\tiny\young(11111112222222333334444444555555566666777888888,3366777,78)}\\
\raisebox{-.6ex}{\tiny\young(11111112222222333333344445555566666667777777888,4445588,88)}\\
\raisebox{-.6ex}{\tiny\young(11111112222233333334444455555666666677777888888,2244557,78)}\\
\raisebox{-.6ex}{\tiny\young(11111112222233333444444555555666666777777888888,2233578,46)}\\
{\tiny\young(11111112222223333444444555555666666777777888888,2334578,36)}
\end{longtable}
\end{proof}

\section{Generators for $m=4$, $n=7$}
\begin{proposition}\label{pro:listofgenerators74}
Define the set

\begin{minipage}{\textwidth}
\tiny{$X := \{\left(7\right) $}, 
\tiny{$ \left(8, 6\right) $}, 
\tiny{$ \left(10, 4\right) $}, 
\tiny{$ \left(12, 2\right) $}, 
\tiny{$ \left(8, 8, 5\right) $}, 
\tiny{$ \left(9, 6, 6\right) $}, 
\tiny{$ \left(10, 7, 4\right) $}, 
\tiny{$ \left(10, 8, 3\right) $}, 
\tiny{$ \left(10, 10, 1\right) $}, 
\tiny{$ \left(11, 6, 4\right) $}, 
\tiny{$ \left(11, 8, 2\right) $}, 
\tiny{$ \left(12, 6, 3\right) $}, 
\tiny{$ \left(12, 7, 2\right) $}, 
\tiny{$ \left(12, 8, 1\right) $}, 
\tiny{$ \left(12, 9\right) $}, 
\tiny{$ \left(13, 4, 4\right) $}, 
\tiny{$ \left(13, 6, 2\right) $}, 
\tiny{$ \left(13, 7, 1\right) $}, 
\tiny{$ \left(13, 8\right) $}, 
\tiny{$ \left(14, 5, 2\right) $}, 
\tiny{$ \left(14, 6, 1\right) $}, 
\tiny{$ \left(14, 7\right) $}, 
\tiny{$ \left(15, 4, 2\right) $}, 
\tiny{$ \left(16, 4, 1\right) $}, 
\tiny{$ \left(16, 5\right) $}, 
\tiny{$ \left(17, 2, 2\right) $}, 
\tiny{$ \left(18, 3\right) $}, 
\tiny{$ \left(10, 6, 6, 6\right) $}, 
\tiny{$ \left(9, 8, 6, 5\right) $}, 
\tiny{$ \left(8, 8, 8, 4\right) $}, 
\tiny{$ \left(10, 8, 6, 4\right) $}, 
\tiny{$ \left(10, 10, 4, 4\right) $}, 
\tiny{$ \left(11, 7, 6, 4\right) $}, 
\tiny{$ \left(11, 8, 5, 4\right) $}, 
\tiny{$ \left(10, 8, 7, 3\right) $}, 
\tiny{$ \left(10, 9, 6, 3\right) $}, 
\tiny{$ \left(12, 6, 6, 4\right) $}, 
\tiny{$ \left(12, 8, 4, 4\right) $}, 
\tiny{$ \left(11, 8, 6, 3\right) $}, 
\tiny{$ \left(11, 9, 5, 3\right) $}, 
\tiny{$ \left(10, 8, 8, 2\right) $}, 
\tiny{$ \left(11, 10, 4, 3\right) $}, 
\tiny{$ \left(10, 10, 6, 2\right) $}, 
\tiny{$ \left(13, 7, 4, 4\right) $}, 
\tiny{$ \left(12, 7, 6, 3\right) $}, 
\tiny{$ \left(12, 8, 5, 3\right) $}, 
\tiny{$ \left(12, 9, 4, 3\right) $}, 
\tiny{$ \left(11, 8, 7, 2\right) $}, 
\tiny{$ \left(11, 9, 6, 2\right) $}, 
\tiny{$ \left(11, 10, 5, 2\right) $}, 
\tiny{$ \left(10, 10, 7, 1\right) $}, 
\tiny{$ \left(14, 6, 4, 4\right) $}, 
\tiny{$ \left(13, 6, 6, 3\right) $}, 
\tiny{$ \left(13, 7, 5, 3\right) $}, 
\tiny{$ \left(13, 8, 4, 3\right) $}, 
\tiny{$ \left(13, 9, 3, 3\right) $}, 
\tiny{$ \left(12, 8, 6, 2\right) $}, 
\tiny{$ \left(12, 9, 5, 2\right) $}, 
\tiny{$ \left(11, 8, 8, 1\right) $}, 
\tiny{$ \left(12, 10, 4, 2\right) $}, 
\tiny{$ \left(11, 9, 7, 1\right) $}, 
\tiny{$ \left(12, 11, 3, 2\right) $}, 
\tiny{$ \left(11, 10, 6, 1\right) $}, 
\tiny{$ \left(12, 12, 2, 2\right) $}, 
\tiny{$ \left(11, 11, 5, 1\right) $}, 
\tiny{$ \left(10, 10, 8\right) $}, 
\tiny{$ \left(14, 7, 4, 3\right) $}, 
\tiny{$ \left(13, 7, 6, 2\right) $}, 
\tiny{$ \left(13, 8, 5, 2\right) $}, 
\tiny{$ \left(13, 9, 4, 2\right) $}, 
\tiny{$ \left(12, 8, 7, 1\right) $}, 
\tiny{$ \left(13, 10, 3, 2\right) $}, 
\tiny{$ \left(12, 9, 6, 1\right) $}, 
\tiny{$ \left(12, 10, 5, 1\right) $}, 
\tiny{$ \left(12, 11, 4, 1\right) $}, 
\tiny{$ \left(11, 10, 7\right) $}, 
\tiny{$ \left(16, 4, 4, 4\right) $}, 
\tiny{$ \left(15, 6, 4, 3\right) $}, 
\tiny{$ \left(14, 6, 6, 2\right) $}, 
\tiny{$ \left(14, 7, 5, 2\right) $}, 
\tiny{$ \left(14, 8, 4, 2\right) $}, 
\tiny{$ \left(13, 7, 7, 1\right) $}, 
\tiny{$ \left(14, 9, 3, 2\right) $}, 
\tiny{$ \left(13, 8, 6, 1\right) $}, 
\tiny{$ \left(14, 10, 2, 2\right) $}, 
\tiny{$ \left(13, 9, 5, 1\right) $}, 
\tiny{$ \left(12, 8, 8\right) $}, 
\tiny{$ \left(13, 10, 4, 1\right) $}, 
\tiny{$ \left(12, 9, 7\right) $}, 
\tiny{$ \left(13, 11, 3, 1\right) $}, 
\tiny{$ \left(12, 10, 6\right) $}, 
\tiny{$ \left(13, 12, 2, 1\right) $}, 
\tiny{$ \left(12, 11, 5\right) $}, 
\tiny{$ \left(13, 13, 1, 1\right) $}, 
\tiny{$ \left(12, 12, 4\right) $}, 
\tiny{$ \left(15, 6, 5, 2\right) $}, 
\tiny{$ \left(15, 7, 4, 2\right) $}, 
\tiny{$ \left(15, 8, 3, 2\right) $}, 
\tiny{$ \left(14, 7, 6, 1\right) $}, 
\tiny{$ \left(15, 9, 2, 2\right) $}, 
\tiny{$ \left(14, 8, 5, 1\right) $}, 
\tiny{$ \left(14, 9, 4, 1\right) $}, 
\tiny{$ \left(13, 8, 7\right) $}, 
\tiny{$ \left(14, 10, 3, 1\right) $}, 
\tiny{$ \left(13, 9, 6\right) $}, 
\tiny{$ \left(14, 11, 2, 1\right) $}, 
\tiny{$ \left(13, 10, 5\right) $}, 
\tiny{$ \left(13, 11, 4\right) $}, 
\tiny{$ \left(13, 12, 3\right) $}, 
\tiny{$ \left(16, 6, 4, 2\right) $}, 
\tiny{$ \left(16, 7, 3, 2\right) $}, 
\tiny{$ \left(15, 6, 6, 1\right) $}, 
\tiny{$ \left(16, 8, 2, 2\right) $}, 
\tiny{$ \left(15, 7, 5, 1\right) $}, 
\tiny{$ \left(15, 8, 4, 1\right) $}, 
\tiny{$ \left(15, 9, 3, 1\right) $}, 
\tiny{$ \left(14, 8, 6\right) $}, 
\tiny{$ \left(15, 10, 2, 1\right) $}, 
\tiny{$ \left(14, 9, 5\right) $}, 
\tiny{$ \left(15, 11, 1, 1\right) $}, 
\tiny{$ \left(14, 10, 4\right) $}, 
\tiny{$ \left(14, 11, 3\right) $}, 
\tiny{$ \left(14, 12, 2\right) $}, 
\tiny{$ \left(14, 13, 1\right) $}, 
\tiny{$ \left(14, 14\right) $}, 
\tiny{$ \left(17, 5, 4, 2\right) $}, 
\tiny{$ \left(17, 6, 3, 2\right) $}, 
\tiny{$ \left(17, 7, 2, 2\right) $}, 
\tiny{$ \left(16, 6, 5, 1\right) $}, 
\tiny{$ \left(16, 7, 4, 1\right) $}, 
\tiny{$ \left(16, 8, 3, 1\right) $}, 
\tiny{$ \left(15, 7, 6\right) $}, 
\tiny{$ \left(16, 9, 2, 1\right) $}, 
\tiny{$ \left(16, 10, 1, 1\right) $}, 
\tiny{$ \left(15, 9, 4\right) $}, 
\tiny{$ \left(15, 10, 3\right) $}, 
\tiny{$ \left(15, 11, 2\right) $}, 
\tiny{$ \left(15, 12, 1\right) $}, 
\tiny{$ \left(18, 4, 4, 2\right) $}, 
\tiny{$ \left(18, 6, 2, 2\right) $}, 
\tiny{$ \left(17, 6, 4, 1\right) $}, 
\tiny{$ \left(17, 7, 3, 1\right) $}, 
\tiny{$ \left(17, 8, 2, 1\right) $}, 
\tiny{$ \left(16, 7, 5\right) $}, 
\tiny{$ \left(17, 9, 1, 1\right) $}, 
\tiny{$ \left(16, 8, 4\right) $}, 
\tiny{$ \left(16, 9, 3\right) $}, 
\tiny{$ \left(16, 10, 2\right) $}, 
\tiny{$ \left(16, 11, 1\right) $}, 
\tiny{$ \left(19, 5, 2, 2\right) $}, 
\tiny{$ \left(18, 5, 4, 1\right) $}, 
\tiny{$ \left(18, 6, 3, 1\right) $}, 
\tiny{$ \left(18, 7, 2, 1\right) $}, 
\tiny{$ \left(17, 6, 5\right) $}, 
\tiny{$ \left(17, 9, 2\right) $}, 
\tiny{$ \left(17, 11\right) $}, 
\tiny{$ \left(20, 4, 2, 2\right) $}, 
\tiny{$ \left(19, 4, 4, 1\right) $}, 
\tiny{$ \left(19, 5, 3, 1\right) $}, 
\tiny{$ \left(19, 6, 2, 1\right) $}, 
\tiny{$ \left(19, 7, 1, 1\right) $}, 
\tiny{$ \left(18, 7, 3\right) $}, 
\tiny{$ \left(18, 9, 1\right) $}, 
\tiny{$ \left(20, 5, 2, 1\right) $}, 
\tiny{$ \left(19, 5, 4\right) $}, 
\tiny{$ \left(22, 2, 2, 2\right) $}, 
\tiny{$ \left(21, 4, 2, 1\right) $}, 
\tiny{$ \left(20, 5, 3\right) $}, 
\tiny{$ \left(21, 4, 3\right) $}, 
\tiny{$ \left(22, 5, 1\right) $}, 
\tiny{$ \left(23, 3, 2\right) $}, 
\tiny{$ \left(11, 8, 8, 8\right) $}, 
\tiny{$ \left(10, 10, 8, 7\right) $}, 
\tiny{$ \left(11, 9, 8, 7\right) $}, 
\tiny{$ \left(11, 10, 7, 7\right) $}, 
\tiny{$ \left(10, 10, 9, 6\right) $}, 
\tiny{$ \left(12, 8, 8, 7\right) $}, 
\tiny{$ \left(11, 10, 8, 6\right) $}, 
\tiny{$ \left(11, 11, 7, 6\right) $}, 
\tiny{$ \left(10, 10, 10, 5\right) $}, 
\tiny{$ \left(13, 8, 7, 7\right) $}, 
\tiny{$ \left(12, 9, 8, 6\right) $}, 
\tiny{$ \left(12, 10, 7, 6\right) $}, 
\tiny{$ \left(12, 11, 6, 6\right) $}, 
\tiny{$ \left(11, 10, 9, 5\right) $}, 
\tiny{$ \left(11, 11, 8, 5\right) $}, 
\tiny{$ \left(13, 8, 8, 6\right) $}, 
\tiny{$ \left(13, 9, 7, 6\right) $}, 
\tiny{$ \left(13, 10, 6, 6\right) $}, 
\tiny{$ \left(12, 9, 9, 5\right) $}, 
\tiny{$ \left(12, 10, 8, 5\right) $}, 
\tiny{$ \left(12, 11, 7, 5\right) $}, 
\tiny{$ \left(11, 10, 10, 4\right) $}, 
\tiny{$ \left(12, 12, 6, 5\right) $}, 
\tiny{$ \left(11, 11, 9, 4\right) $}, 
\tiny{$ \left(14, 8, 7, 6\right) $}, 
\tiny{$ \left(14, 9, 6, 6\right) $}, 
\tiny{$ \left(13, 9, 8, 5\right) $}, 
\tiny{$ \left(13, 10, 7, 5\right) $}, 
\tiny{$ \left(13, 11, 6, 5\right) $}, 
\tiny{$ \left(12, 10, 9, 4\right) $}, 
\tiny{$ \left(13, 12, 5, 5\right) $}, 
\tiny{$ \left(12, 11, 8, 4\right) $}, 
\tiny{$ \left(12, 12, 7, 4\right) $}, 
\tiny{$ \left(11, 11, 10, 3\right) $}, 
\tiny{$ \left(15, 7, 7, 6\right) $}, 
\tiny{$ \left(15, 8, 6, 6\right) $}, 
\tiny{$ \left(14, 8, 8, 5\right) $}, 
\tiny{$ \left(14, 9, 7, 5\right) $}, 
\tiny{$ \left(14, 10, 6, 5\right) $}, 
\tiny{$ \left(13, 9, 9, 4\right) $}, 
\tiny{$ \left(14, 11, 5, 5\right) $}, 
\tiny{$ \left(13, 10, 8, 4\right) $}, 
\tiny{$ \left(13, 11, 7, 4\right) $}, 
\tiny{$ \left(12, 10, 10, 3\right) $}, 
\tiny{$ \left(13, 12, 6, 4\right) $}, 
\tiny{$ \left(12, 11, 9, 3\right) $}, 
\tiny{$ \left(13, 13, 5, 4\right) $}, 
\tiny{$ \left(12, 12, 8, 3\right) $}, 
\tiny{$ \left(16, 7, 6, 6\right) $}, 
\tiny{$ \left(15, 8, 7, 5\right) $}, 
\tiny{$ \left(15, 9, 6, 5\right) $}, 
\tiny{$ \left(15, 10, 5, 5\right) $}, 
\tiny{$ \left(14, 9, 8, 4\right) $}, 
\tiny{$ \left(14, 10, 7, 4\right) $}, 
\tiny{$ \left(14, 11, 6, 4\right) $}, 
\tiny{$ \left(13, 10, 9, 3\right) $}, 
\tiny{$ \left(14, 12, 5, 4\right) $}, 
\tiny{$ \left(13, 11, 8, 3\right) $}, 
\tiny{$ \left(14, 13, 4, 4\right) $}, 
\tiny{$ \left(13, 12, 7, 3\right) $}, 
\tiny{$ \left(12, 11, 10, 2\right) $}, 
\tiny{$ \left(13, 13, 6, 3\right) $}, 
\tiny{$ \left(12, 12, 9, 2\right) $}, 
\tiny{$ \left(16, 7, 7, 5\right) $}, 
\tiny{$ \left(16, 9, 5, 5\right) $}, 
\tiny{$ \left(15, 9, 7, 4\right) $}, 
\tiny{$ \left(15, 10, 6, 4\right) $}, 
\tiny{$ \left(14, 9, 9, 3\right) $}, 
\tiny{$ \left(15, 11, 5, 4\right) $}, 
\tiny{$ \left(14, 10, 8, 3\right) $}, 
\tiny{$ \left(15, 12, 4, 4\right) $}, 
\tiny{$ \left(14, 11, 7, 3\right) $}, 
\tiny{$ \left(13, 10, 10, 2\right) $}, 
\tiny{$ \left(14, 12, 6, 3\right) $}, 
\tiny{$ \left(13, 11, 9, 2\right) $}, 
\tiny{$ \left(14, 13, 5, 3\right) $}, 
\tiny{$ \left(13, 12, 8, 2\right) $}, 
\tiny{$ \left(14, 14, 4, 3\right) $}, 
\tiny{$ \left(13, 13, 7, 2\right) $}, 
\tiny{$ \left(12, 12, 10, 1\right) $}, 
\tiny{$ \left(17, 7, 6, 5\right) $}, 
\tiny{$ \left(17, 8, 5, 5\right) $}, 
\tiny{$ \left(16, 8, 7, 4\right) $}, 
\tiny{$ \left(16, 9, 6, 4\right) $}, 
\tiny{$ \left(16, 10, 5, 4\right) $}, 
\tiny{$ \left(15, 9, 8, 3\right) $}, 
\tiny{$ \left(16, 11, 4, 4\right) $}, 
\tiny{$ \left(15, 10, 7, 3\right) $}, 
\tiny{$ \left(15, 11, 6, 3\right) $}, 
\tiny{$ \left(14, 10, 9, 2\right) $}, 
\tiny{$ \left(15, 12, 5, 3\right) $}, 
\tiny{$ \left(14, 11, 8, 2\right) $}, 
\tiny{$ \left(15, 13, 4, 3\right) $}, 
\tiny{$ \left(14, 12, 7, 2\right) $}, 
\tiny{$ \left(13, 11, 10, 1\right) $}, 
\tiny{$ \left(15, 14, 3, 3\right) $}, 
\tiny{$ \left(14, 13, 6, 2\right) $}, 
\tiny{$ \left(13, 12, 9, 1\right) $}, 
\tiny{$ \left(14, 14, 5, 2\right) $}, 
\tiny{$ \left(13, 13, 8, 1\right) $}, 
\tiny{$ \left(12, 12, 11\right) $}, 
\tiny{$ \left(18, 6, 6, 5\right) $}, 
\tiny{$ \left(18, 7, 5, 5\right) $}, 
\tiny{$ \left(17, 7, 7, 4\right) $}, 
\tiny{$ \left(17, 9, 5, 4\right) $}, 
\tiny{$ \left(16, 8, 8, 3\right) $}, 
\tiny{$ \left(16, 9, 7, 3\right) $}, 
\tiny{$ \left(16, 10, 6, 3\right) $}, 
\tiny{$ \left(15, 9, 9, 2\right) $}, 
\tiny{$ \left(16, 11, 5, 3\right) $}, 
\tiny{$ \left(15, 10, 8, 2\right) $}, 
\tiny{$ \left(16, 12, 4, 3\right) $}, 
\tiny{$ \left(15, 11, 7, 2\right) $}, 
\tiny{$ \left(14, 10, 10, 1\right) $}, 
\tiny{$ \left(16, 13, 3, 3\right) $}, 
\tiny{$ \left(15, 12, 6, 2\right) $}, 
\tiny{$ \left(14, 11, 9, 1\right) $}, 
\tiny{$ \left(15, 13, 5, 2\right) $}, 
\tiny{$ \left(14, 12, 8, 1\right) $}, 
\tiny{$ \left(15, 14, 4, 2\right) $}, 
\tiny{$ \left(14, 13, 7, 1\right) $}, 
\tiny{$ \left(13, 12, 10\right) $}, 
\tiny{$ \left(15, 15, 3, 2\right) $}, 
\tiny{$ \left(14, 14, 6, 1\right) $}, 
\tiny{$ \left(13, 13, 9\right) $}, 
\tiny{$ \left(18, 9, 4, 4\right) $}, 
\tiny{$ \left(17, 10, 5, 3\right) $}, 
\tiny{$ \left(16, 9, 8, 2\right) $}, 
\tiny{$ \left(17, 11, 4, 3\right) $}, 
\tiny{$ \left(16, 10, 7, 2\right) $}, 
\tiny{$ \left(17, 12, 3, 3\right) $}, 
\tiny{$ \left(16, 11, 6, 2\right) $}, 
\tiny{$ \left(15, 10, 9, 1\right) $}, 
\tiny{$ \left(16, 12, 5, 2\right) $}, 
\tiny{$ \left(15, 11, 8, 1\right) $}, 
\tiny{$ \left(16, 13, 4, 2\right) $}, 
\tiny{$ \left(15, 12, 7, 1\right) $}, 
\tiny{$ \left(14, 11, 10\right) $}, 
\tiny{$ \left(16, 14, 3, 2\right) $}, 
\tiny{$ \left(15, 13, 6, 1\right) $}, 
\tiny{$ \left(14, 12, 9\right) $}, 
\tiny{$ \left(16, 15, 2, 2\right) $}, 
\tiny{$ \left(15, 14, 5, 1\right) $}, 
\tiny{$ \left(14, 13, 8\right) $}, 
\tiny{$ \left(15, 15, 4, 1\right) $}, 
\tiny{$ \left(14, 14, 7\right) $}, 
\tiny{$ \left(19, 7, 5, 4\right) $}, 
\tiny{$ \left(18, 7, 7, 3\right) $}, 
\tiny{$ \left(17, 9, 7, 2\right) $}, 
\tiny{$ \left(18, 11, 3, 3\right) $}, 
\tiny{$ \left(16, 9, 9, 1\right) $}, 
\tiny{$ \left(17, 11, 5, 2\right) $}, 
\tiny{$ \left(16, 10, 8, 1\right) $}, 
\tiny{$ \left(17, 12, 4, 2\right) $}, 
\tiny{$ \left(16, 11, 7, 1\right) $}, 
\tiny{$ \left(15, 10, 10\right) $}, 
\tiny{$ \left(17, 13, 3, 2\right) $}, 
\tiny{$ \left(16, 12, 6, 1\right) $}, 
\tiny{$ \left(15, 11, 9\right) $}, 
\tiny{$ \left(17, 14, 2, 2\right) $}, 
\tiny{$ \left(16, 13, 5, 1\right) $}, 
\tiny{$ \left(15, 12, 8\right) $}, 
\tiny{$ \left(16, 14, 4, 1\right) $}, 
\tiny{$ \left(15, 13, 7\right) $}, 
\tiny{$ \left(16, 15, 3, 1\right) $}, 
\tiny{$ \left(15, 14, 6\right) $}, 
\tiny{$ \left(16, 16, 2, 1\right) $}, 
\tiny{$ \left(15, 15, 5\right) $}, 
\tiny{$ \left(20, 6, 5, 4\right) $}, 
\tiny{$ \left(19, 10, 3, 3\right) $}, 
\tiny{$ \left(17, 9, 8, 1\right) $}, 
\tiny{$ \left(18, 11, 4, 2\right) $}, 
\tiny{$ \left(18, 12, 3, 2\right) $}, 
\tiny{$ \left(17, 11, 6, 1\right) $}, 
\tiny{$ \left(16, 10, 9\right) $}, 
\tiny{$ \left(18, 13, 2, 2\right) $}, 
\tiny{$ \left(17, 12, 5, 1\right) $}, 
\tiny{$ \left(16, 11, 8\right) $}, 
\tiny{$ \left(17, 13, 4, 1\right) $}, 
\tiny{$ \left(16, 12, 7\right) $}, 
\tiny{$ \left(17, 14, 3, 1\right) $}, 
\tiny{$ \left(16, 13, 6\right) $}, 
\tiny{$ \left(17, 15, 2, 1\right) $}, 
\tiny{$ \left(17, 16, 1, 1\right) $}, 
\tiny{$ \left(16, 15, 4\right) $}, 
\tiny{$ \left(16, 16, 3\right) $}, 
\tiny{$ \left(19, 7, 7, 2\right) $}, 
\tiny{$ \left(17, 9, 9\right) $}, 
\tiny{$ \left(18, 12, 4, 1\right) $}, 
\tiny{$ \left(17, 11, 7\right) $}, 
\tiny{$ \left(18, 13, 3, 1\right) $}, 
\tiny{$ \left(18, 14, 2, 1\right) $}, 
\tiny{$ \left(17, 13, 5\right) $}, 
\tiny{$ \left(18, 15, 1, 1\right) $}, 
\tiny{$ \left(17, 14, 4\right) $}, 
\tiny{$ \left(17, 15, 3\right) $}, 
\tiny{$ \left(17, 16, 2\right) $}, 
\tiny{$ \left(22, 5, 4, 4\right) $}, 
\tiny{$ \left(21, 6, 5, 3\right) $}, 
\tiny{$ \left(21, 8, 3, 3\right) $}, 
\tiny{$ \left(20, 11, 2, 2\right) $}, 
\tiny{$ \left(18, 9, 8\right) $}, 
\tiny{$ \left(19, 12, 3, 1\right) $}, 
\tiny{$ \left(18, 11, 6\right) $}, 
\tiny{$ \left(19, 13, 2, 1\right) $}, 
\tiny{$ \left(19, 14, 1, 1\right) $}, 
\tiny{$ \left(18, 15, 2\right) $}, 
\tiny{$ \left(18, 17\right) $}, 
\tiny{$ \left(22, 7, 3, 3\right) $}, 
\tiny{$ \left(19, 13, 3\right) $}, 
\tiny{$ \left(19, 15, 1\right) $}, 
\tiny{$ \left(19, 16\right) $}, 
\tiny{$ \left(23, 5, 4, 3\right) $}, 
\tiny{$ \left(23, 6, 3, 3\right) $}, 
\tiny{$ \left(21, 12, 1, 1\right) $}, 
\tiny{$ \left(24, 4, 4, 3\right) $}, 
\tiny{$ \left(23, 5, 5, 2\right) $}, 
\tiny{$ \left(21, 7, 7\right) $}, 
\tiny{$ \left(25, 5, 3, 2\right) $}, 
\tiny{$ \left(24, 5, 5, 1\right) $}, 
\tiny{$ \left(26, 4, 3, 2\right) $}, 
\tiny{$ \left(25, 8, 1, 1\right) $}, 
\tiny{$ \left(28, 3, 2, 2\right) $}, 
\tiny{$ \left(27, 4, 3, 1\right) $}, 
\tiny{$ \left(12, 10, 10, 10\right) $}, 
\tiny{$ \left(11, 11, 11, 9\right) $}, 
\tiny{$ \left(12, 11, 10, 9\right) $}, 
\tiny{$ \left(12, 12, 9, 9\right) $}, 
\tiny{$ \left(13, 10, 10, 9\right) $}, 
\tiny{$ \left(13, 11, 9, 9\right) $}, 
\tiny{$ \left(12, 11, 11, 8\right) $}, 
\tiny{$ \left(12, 12, 10, 8\right) $}, 
\tiny{$ \left(14, 10, 9, 9\right) $}, 
\tiny{$ \left(13, 11, 10, 8\right) $}, 
\tiny{$ \left(13, 12, 9, 8\right) $}, 
\tiny{$ \left(13, 13, 8, 8\right) $}, 
\tiny{$ \left(12, 12, 11, 7\right) $}, 
\tiny{$ \left(15, 9, 9, 9\right) $}, 
\tiny{$ \left(14, 10, 10, 8\right) $}, 
\tiny{$ \left(14, 11, 9, 8\right) $}, 
\tiny{$ \left(14, 12, 8, 8\right) $}, 
\tiny{$ \left(13, 11, 11, 7\right) $}, 
\tiny{$ \left(13, 12, 10, 7\right) $}, 
\tiny{$ \left(13, 13, 9, 7\right) $}, 
\tiny{$ \left(12, 12, 12, 6\right) $}, 
\tiny{$ \left(15, 10, 9, 8\right) $}, 
\tiny{$ \left(15, 11, 8, 8\right) $}, 
\tiny{$ \left(14, 11, 10, 7\right) $}, 
\tiny{$ \left(14, 12, 9, 7\right) $}, 
\tiny{$ \left(14, 13, 8, 7\right) $}, 
\tiny{$ \left(13, 12, 11, 6\right) $}, 
\tiny{$ \left(14, 14, 7, 7\right) $}, 
\tiny{$ \left(13, 13, 10, 6\right) $}, 
\tiny{$ \left(16, 9, 9, 8\right) $}, 
\tiny{$ \left(16, 10, 8, 8\right) $}, 
\tiny{$ \left(15, 10, 10, 7\right) $}, 
\tiny{$ \left(15, 11, 9, 7\right) $}, 
\tiny{$ \left(15, 12, 8, 7\right) $}, 
\tiny{$ \left(14, 11, 11, 6\right) $}, 
\tiny{$ \left(15, 13, 7, 7\right) $}, 
\tiny{$ \left(14, 12, 10, 6\right) $}, 
\tiny{$ \left(14, 13, 9, 6\right) $}, 
\tiny{$ \left(13, 12, 12, 5\right) $}, 
\tiny{$ \left(14, 14, 8, 6\right) $}, 
\tiny{$ \left(13, 13, 11, 5\right) $}, 
\tiny{$ \left(17, 9, 8, 8\right) $}, 
\tiny{$ \left(16, 10, 9, 7\right) $}, 
\tiny{$ \left(16, 11, 8, 7\right) $}, 
\tiny{$ \left(16, 12, 7, 7\right) $}, 
\tiny{$ \left(15, 11, 10, 6\right) $}, 
\tiny{$ \left(15, 12, 9, 6\right) $}, 
\tiny{$ \left(15, 13, 8, 6\right) $}, 
\tiny{$ \left(14, 12, 11, 5\right) $}, 
\tiny{$ \left(15, 14, 7, 6\right) $}, 
\tiny{$ \left(14, 13, 10, 5\right) $}, 
\tiny{$ \left(15, 15, 6, 6\right) $}, 
\tiny{$ \left(14, 14, 9, 5\right) $}, 
\tiny{$ \left(13, 13, 12, 4\right) $}, 
\tiny{$ \left(17, 9, 9, 7\right) $}, 
\tiny{$ \left(17, 11, 7, 7\right) $}, 
\tiny{$ \left(16, 10, 10, 6\right) $}, 
\tiny{$ \left(16, 11, 9, 6\right) $}, 
\tiny{$ \left(16, 12, 8, 6\right) $}, 
\tiny{$ \left(15, 11, 11, 5\right) $}, 
\tiny{$ \left(16, 13, 7, 6\right) $}, 
\tiny{$ \left(15, 12, 10, 5\right) $}, 
\tiny{$ \left(16, 14, 6, 6\right) $}, 
\tiny{$ \left(15, 13, 9, 5\right) $}, 
\tiny{$ \left(14, 12, 12, 4\right) $}, 
\tiny{$ \left(15, 14, 8, 5\right) $}, 
\tiny{$ \left(14, 13, 11, 4\right) $}, 
\tiny{$ \left(15, 15, 7, 5\right) $}, 
\tiny{$ \left(14, 14, 10, 4\right) $}, 
\tiny{$ \left(13, 13, 13, 3\right) $}, 
\tiny{$ \left(17, 11, 8, 6\right) $}, 
\tiny{$ \left(17, 12, 7, 6\right) $}, 
\tiny{$ \left(16, 11, 10, 5\right) $}, 
\tiny{$ \left(17, 13, 6, 6\right) $}, 
\tiny{$ \left(16, 12, 9, 5\right) $}, 
\tiny{$ \left(16, 13, 8, 5\right) $}, 
\tiny{$ \left(15, 12, 11, 4\right) $}, 
\tiny{$ \left(16, 14, 7, 5\right) $}, 
\tiny{$ \left(15, 13, 10, 4\right) $}, 
\tiny{$ \left(16, 15, 6, 5\right) $}, 
\tiny{$ \left(15, 14, 9, 4\right) $}, 
\tiny{$ \left(14, 13, 12, 3\right) $}, 
\tiny{$ \left(16, 16, 5, 5\right) $}, 
\tiny{$ \left(15, 15, 8, 4\right) $}, 
\tiny{$ \left(14, 14, 11, 3\right) $}, 
\tiny{$ \left(19, 9, 7, 7\right) $}, 
\tiny{$ \left(18, 9, 9, 6\right) $}, 
\tiny{$ \left(17, 11, 9, 5\right) $}, 
\tiny{$ \left(17, 12, 8, 5\right) $}, 
\tiny{$ \left(16, 11, 11, 4\right) $}, 
\tiny{$ \left(17, 13, 7, 5\right) $}, 
\tiny{$ \left(16, 12, 10, 4\right) $}, 
\tiny{$ \left(16, 13, 9, 4\right) $}, 
\tiny{$ \left(15, 12, 12, 3\right) $}, 
\tiny{$ \left(17, 15, 5, 5\right) $}, 
\tiny{$ \left(15, 13, 11, 3\right) $}, 
\tiny{$ \left(16, 15, 7, 4\right) $}, 
\tiny{$ \left(15, 14, 10, 3\right) $}, 
\tiny{$ \left(14, 13, 13, 2\right) $}, 
\tiny{$ \left(16, 16, 6, 4\right) $}, 
\tiny{$ \left(15, 15, 9, 3\right) $}, 
\tiny{$ \left(14, 14, 12, 2\right) $}, 
\tiny{$ \left(18, 12, 7, 5\right) $}, 
\tiny{$ \left(17, 11, 10, 4\right) $}, 
\tiny{$ \left(18, 13, 6, 5\right) $}, 
\tiny{$ \left(17, 12, 9, 4\right) $}, 
\tiny{$ \left(18, 14, 5, 5\right) $}, 
\tiny{$ \left(17, 13, 8, 4\right) $}, 
\tiny{$ \left(16, 12, 11, 3\right) $}, 
\tiny{$ \left(17, 14, 7, 4\right) $}, 
\tiny{$ \left(16, 13, 10, 3\right) $}, 
\tiny{$ \left(17, 15, 6, 4\right) $}, 
\tiny{$ \left(16, 14, 9, 3\right) $}, 
\tiny{$ \left(15, 13, 12, 2\right) $}, 
\tiny{$ \left(17, 16, 5, 4\right) $}, 
\tiny{$ \left(16, 15, 8, 3\right) $}, 
\tiny{$ \left(15, 14, 11, 2\right) $}, 
\tiny{$ \left(17, 17, 4, 4\right) $}, 
\tiny{$ \left(16, 16, 7, 3\right) $}, 
\tiny{$ \left(15, 15, 10, 2\right) $}, 
\tiny{$ \left(14, 14, 13, 1\right) $}, 
\tiny{$ \left(21, 7, 7, 7\right) $}, 
\tiny{$ \left(19, 13, 5, 5\right) $}, 
\tiny{$ \left(17, 11, 11, 3\right) $}, 
\tiny{$ \left(18, 13, 7, 4\right) $}, 
\tiny{$ \left(17, 12, 10, 3\right) $}, 
\tiny{$ \left(17, 13, 9, 3\right) $}, 
\tiny{$ \left(16, 12, 12, 2\right) $}, 
\tiny{$ \left(18, 15, 5, 4\right) $}, 
\tiny{$ \left(17, 14, 8, 3\right) $}, 
\tiny{$ \left(16, 13, 11, 2\right) $}, 
\tiny{$ \left(17, 15, 7, 3\right) $}, 
\tiny{$ \left(16, 14, 10, 2\right) $}, 
\tiny{$ \left(15, 13, 13, 1\right) $}, 
\tiny{$ \left(17, 16, 6, 3\right) $}, 
\tiny{$ \left(16, 15, 9, 2\right) $}, 
\tiny{$ \left(15, 14, 12, 1\right) $}, 
\tiny{$ \left(17, 17, 5, 3\right) $}, 
\tiny{$ \left(16, 16, 8, 2\right) $}, 
\tiny{$ \left(15, 15, 11, 1\right) $}, 
\tiny{$ \left(14, 14, 14\right) $}, 
\tiny{$ \left(18, 12, 9, 3\right) $}, 
\tiny{$ \left(18, 13, 8, 3\right) $}, 
\tiny{$ \left(17, 12, 11, 2\right) $}, 
\tiny{$ \left(19, 15, 4, 4\right) $}, 
\tiny{$ \left(17, 13, 10, 2\right) $}, 
\tiny{$ \left(17, 14, 9, 2\right) $}, 
\tiny{$ \left(16, 13, 12, 1\right) $}, 
\tiny{$ \left(18, 16, 5, 3\right) $}, 
\tiny{$ \left(17, 15, 8, 2\right) $}, 
\tiny{$ \left(16, 14, 11, 1\right) $}, 
\tiny{$ \left(18, 17, 4, 3\right) $}, 
\tiny{$ \left(17, 16, 7, 2\right) $}, 
\tiny{$ \left(16, 15, 10, 1\right) $}, 
\tiny{$ \left(15, 14, 13\right) $}, 
\tiny{$ \left(18, 18, 3, 3\right) $}, 
\tiny{$ \left(17, 17, 6, 2\right) $}, 
\tiny{$ \left(16, 16, 9, 1\right) $}, 
\tiny{$ \left(15, 15, 12\right) $}, 
\tiny{$ \left(18, 11, 11, 2\right) $}, 
\tiny{$ \left(19, 13, 7, 3\right) $}, 
\tiny{$ \left(18, 12, 10, 2\right) $}, 
\tiny{$ \left(18, 13, 9, 2\right) $}, 
\tiny{$ \left(17, 12, 12, 1\right) $}, 
\tiny{$ \left(17, 13, 11, 1\right) $}, 
\tiny{$ \left(18, 15, 7, 2\right) $}, 
\tiny{$ \left(17, 14, 10, 1\right) $}, 
\tiny{$ \left(16, 13, 13\right) $}, 
\tiny{$ \left(19, 17, 3, 3\right) $}, 
\tiny{$ \left(17, 15, 9, 1\right) $}, 
\tiny{$ \left(16, 14, 12\right) $}, 
\tiny{$ \left(18, 17, 5, 2\right) $}, 
\tiny{$ \left(17, 16, 8, 1\right) $}, 
\tiny{$ \left(16, 15, 11\right) $}, 
\tiny{$ \left(18, 18, 4, 2\right) $}, 
\tiny{$ \left(17, 17, 7, 1\right) $}, 
\tiny{$ \left(19, 13, 8, 2\right) $}, 
\tiny{$ \left(18, 12, 11, 1\right) $}, 
\tiny{$ \left(18, 13, 10, 1\right) $}, 
\tiny{$ \left(20, 16, 3, 3\right) $}, 
\tiny{$ \left(18, 14, 9, 1\right) $}, 
\tiny{$ \left(17, 13, 12\right) $}, 
\tiny{$ \left(18, 15, 8, 1\right) $}, 
\tiny{$ \left(19, 17, 4, 2\right) $}, 
\tiny{$ \left(17, 15, 10\right) $}, 
\tiny{$ \left(19, 18, 3, 2\right) $}, 
\tiny{$ \left(18, 17, 6, 1\right) $}, 
\tiny{$ \left(17, 16, 9\right) $}, 
\tiny{$ \left(19, 19, 2, 2\right) $}, 
\tiny{$ \left(18, 18, 5, 1\right) $}, 
\tiny{$ \left(17, 17, 8\right) $}, 
\tiny{$ \left(19, 11, 11, 1\right) $}, 
\tiny{$ \left(19, 13, 9, 1\right) $}, 
\tiny{$ \left(18, 13, 11\right) $}, 
\tiny{$ \left(18, 14, 10\right) $}, 
\tiny{$ \left(19, 18, 4, 1\right) $}, 
\tiny{$ \left(18, 17, 7\right) $}, 
\tiny{$ \left(19, 19, 3, 1\right) $}, 
\tiny{$ \left(21, 17, 2, 2\right) $}, 
\tiny{$ \left(19, 15, 8\right) $}, 
\tiny{$ \left(20, 18, 3, 1\right) $}, 
\tiny{$ \left(19, 17, 6\right) $}, 
\tiny{$ \left(20, 19, 2, 1\right) $}, 
\tiny{$ \left(19, 18, 5\right) $}, 
\tiny{$ \left(19, 19, 4\right) $}, 
\tiny{$ \left(20, 11, 11\right) $}, 
\tiny{$ \left(20, 19, 3\right) $}, 
\tiny{$ \left(26, 6, 5, 5\right) $}, 
\tiny{$ \left(22, 18, 1, 1\right) $}, 
\tiny{$ \left(21, 19, 2\right) $}, 
\tiny{$ \left(21, 20, 1\right) $}, 
\tiny{$ \left(23, 19\right) $}, 
\tiny{$ \left(28, 5, 5, 4\right) $}, 
\tiny{$ \left(29, 5, 5, 3\right) $}, 
\tiny{$ \left(31, 5, 3, 3\right) $}, 
\tiny{$ \left(32, 5, 5\right) $}, 
\tiny{$ \left(13, 12, 12, 12\right) $}, 
\tiny{$ \left(13, 13, 12, 11\right) $}, 
\tiny{$ \left(14, 12, 12, 11\right) $}, 
\tiny{$ \left(14, 13, 11, 11\right) $}, 
\tiny{$ \left(13, 13, 13, 10\right) $}, 
\tiny{$ \left(15, 12, 11, 11\right) $}, 
\tiny{$ \left(14, 13, 12, 10\right) $}, 
\tiny{$ \left(14, 14, 11, 10\right) $}, 
\tiny{$ \left(16, 11, 11, 11\right) $}, 
\tiny{$ \left(15, 12, 12, 10\right) $}, 
\tiny{$ \left(15, 13, 11, 10\right) $}, 
\tiny{$ \left(15, 14, 10, 10\right) $}, 
\tiny{$ \left(14, 13, 13, 9\right) $}, 
\tiny{$ \left(14, 14, 12, 9\right) $}, 
\tiny{$ \left(16, 12, 11, 10\right) $}, 
\tiny{$ \left(16, 13, 10, 10\right) $}, 
\tiny{$ \left(15, 13, 12, 9\right) $}, 
\tiny{$ \left(15, 14, 11, 9\right) $}, 
\tiny{$ \left(15, 15, 10, 9\right) $}, 
\tiny{$ \left(14, 14, 13, 8\right) $}, 
\tiny{$ \left(17, 11, 11, 10\right) $}, 
\tiny{$ \left(17, 12, 10, 10\right) $}, 
\tiny{$ \left(16, 12, 12, 9\right) $}, 
\tiny{$ \left(16, 13, 11, 9\right) $}, 
\tiny{$ \left(16, 14, 10, 9\right) $}, 
\tiny{$ \left(15, 13, 13, 8\right) $}, 
\tiny{$ \left(16, 15, 9, 9\right) $}, 
\tiny{$ \left(15, 14, 12, 8\right) $}, 
\tiny{$ \left(15, 15, 11, 8\right) $}, 
\tiny{$ \left(14, 14, 14, 7\right) $}, 
\tiny{$ \left(18, 11, 10, 10\right) $}, 
\tiny{$ \left(17, 12, 11, 9\right) $}, 
\tiny{$ \left(17, 13, 10, 9\right) $}, 
\tiny{$ \left(17, 14, 9, 9\right) $}, 
\tiny{$ \left(16, 13, 12, 8\right) $}, 
\tiny{$ \left(16, 14, 11, 8\right) $}, 
\tiny{$ \left(16, 15, 10, 8\right) $}, 
\tiny{$ \left(15, 14, 13, 7\right) $}, 
\tiny{$ \left(16, 16, 9, 8\right) $}, 
\tiny{$ \left(15, 15, 12, 7\right) $}, 
\tiny{$ \left(18, 12, 10, 9\right) $}, 
\tiny{$ \left(18, 13, 9, 9\right) $}, 
\tiny{$ \left(17, 12, 12, 8\right) $}, 
\tiny{$ \left(17, 13, 11, 8\right) $}, 
\tiny{$ \left(17, 14, 10, 8\right) $}, 
\tiny{$ \left(16, 13, 13, 7\right) $}, 
\tiny{$ \left(17, 15, 9, 8\right) $}, 
\tiny{$ \left(16, 14, 12, 7\right) $}, 
\tiny{$ \left(17, 16, 8, 8\right) $}, 
\tiny{$ \left(16, 15, 11, 7\right) $}, 
\tiny{$ \left(15, 14, 14, 6\right) $}, 
\tiny{$ \left(16, 16, 10, 7\right) $}, 
\tiny{$ \left(15, 15, 13, 6\right) $}, 
\tiny{$ \left(18, 12, 11, 8\right) $}, 
\tiny{$ \left(18, 13, 10, 8\right) $}, 
\tiny{$ \left(18, 14, 9, 8\right) $}, 
\tiny{$ \left(17, 13, 12, 7\right) $}, 
\tiny{$ \left(18, 15, 8, 8\right) $}, 
\tiny{$ \left(17, 14, 11, 7\right) $}, 
\tiny{$ \left(17, 15, 10, 7\right) $}, 
\tiny{$ \left(16, 14, 13, 6\right) $}, 
\tiny{$ \left(17, 16, 9, 7\right) $}, 
\tiny{$ \left(16, 15, 12, 6\right) $}, 
\tiny{$ \left(17, 17, 8, 7\right) $}, 
\tiny{$ \left(16, 16, 11, 6\right) $}, 
\tiny{$ \left(15, 15, 14, 5\right) $}, 
\tiny{$ \left(19, 13, 9, 8\right) $}, 
\tiny{$ \left(18, 12, 12, 7\right) $}, 
\tiny{$ \left(18, 13, 11, 7\right) $}, 
\tiny{$ \left(18, 14, 10, 7\right) $}, 
\tiny{$ \left(17, 13, 13, 6\right) $}, 
\tiny{$ \left(18, 15, 9, 7\right) $}, 
\tiny{$ \left(17, 14, 12, 6\right) $}, 
\tiny{$ \left(17, 15, 11, 6\right) $}, 
\tiny{$ \left(16, 14, 14, 5\right) $}, 
\tiny{$ \left(18, 17, 7, 7\right) $}, 
\tiny{$ \left(17, 16, 10, 6\right) $}, 
\tiny{$ \left(16, 15, 13, 5\right) $}, 
\tiny{$ \left(17, 17, 9, 6\right) $}, 
\tiny{$ \left(16, 16, 12, 5\right) $}, 
\tiny{$ \left(15, 15, 15, 4\right) $}, 
\tiny{$ \left(19, 13, 10, 7\right) $}, 
\tiny{$ \left(19, 14, 9, 7\right) $}, 
\tiny{$ \left(18, 13, 12, 6\right) $}, 
\tiny{$ \left(18, 15, 10, 6\right) $}, 
\tiny{$ \left(17, 14, 13, 5\right) $}, 
\tiny{$ \left(17, 15, 12, 5\right) $}, 
\tiny{$ \left(18, 17, 8, 6\right) $}, 
\tiny{$ \left(16, 15, 14, 4\right) $}, 
\tiny{$ \left(18, 18, 7, 6\right) $}, 
\tiny{$ \left(17, 17, 10, 5\right) $}, 
\tiny{$ \left(19, 13, 11, 6\right) $}, 
\tiny{$ \left(20, 15, 7, 7\right) $}, 
\tiny{$ \left(19, 14, 10, 6\right) $}, 
\tiny{$ \left(18, 13, 13, 5\right) $}, 
\tiny{$ \left(19, 15, 9, 6\right) $}, 
\tiny{$ \left(18, 15, 11, 5\right) $}, 
\tiny{$ \left(17, 15, 13, 4\right) $}, 
\tiny{$ \left(19, 18, 6, 6\right) $}, 
\tiny{$ \left(18, 17, 9, 5\right) $}, 
\tiny{$ \left(17, 16, 12, 4\right) $}, 
\tiny{$ \left(16, 15, 15, 3\right) $}, 
\tiny{$ \left(18, 18, 8, 5\right) $}, 
\tiny{$ \left(17, 17, 11, 4\right) $}, 
\tiny{$ \left(16, 16, 14, 3\right) $}, 
\tiny{$ \left(19, 13, 12, 5\right) $}, 
\tiny{$ \left(19, 14, 11, 5\right) $}, 
\tiny{$ \left(18, 14, 13, 4\right) $}, 
\tiny{$ \left(17, 15, 14, 3\right) $}, 
\tiny{$ \left(18, 17, 10, 4\right) $}, 
\tiny{$ \left(17, 16, 13, 3\right) $}, 
\tiny{$ \left(19, 19, 6, 5\right) $}, 
\tiny{$ \left(17, 17, 12, 3\right) $}, 
\tiny{$ \left(16, 16, 15, 2\right) $}, 
\tiny{$ \left(19, 13, 13, 4\right) $}, 
\tiny{$ \left(18, 14, 14, 3\right) $}, 
\tiny{$ \left(18, 15, 13, 3\right) $}, 
\tiny{$ \left(17, 15, 15, 2\right) $}, 
\tiny{$ \left(20, 19, 5, 5\right) $}, 
\tiny{$ \left(19, 18, 8, 4\right) $}, 
\tiny{$ \left(17, 16, 14, 2\right) $}, 
\tiny{$ \left(19, 19, 7, 4\right) $}, 
\tiny{$ \left(18, 18, 10, 3\right) $}, 
\tiny{$ \left(17, 17, 13, 2\right) $}, 
\tiny{$ \left(16, 16, 16, 1\right) $}, 
\tiny{$ \left(18, 15, 14, 2\right) $}, 
\tiny{$ \left(20, 19, 6, 4\right) $}, 
\tiny{$ \left(18, 17, 12, 2\right) $}, 
\tiny{$ \left(17, 16, 15, 1\right) $}, 
\tiny{$ \left(19, 19, 8, 3\right) $}, 
\tiny{$ \left(17, 17, 14, 1\right) $}, 
\tiny{$ \left(19, 15, 13, 2\right) $}, 
\tiny{$ \left(18, 15, 15, 1\right) $}, 
\tiny{$ \left(18, 16, 14, 1\right) $}, 
\tiny{$ \left(21, 20, 4, 4\right) $}, 
\tiny{$ \left(18, 17, 13, 1\right) $}, 
\tiny{$ \left(17, 16, 16\right) $}, 
\tiny{$ \left(20, 20, 6, 3\right) $}, 
\tiny{$ \left(19, 19, 9, 2\right) $}, 
\tiny{$ \left(17, 17, 15\right) $}, 
\tiny{$ \left(19, 15, 14, 1\right) $}, 
\tiny{$ \left(18, 16, 15\right) $}, 
\tiny{$ \left(18, 17, 14\right) $}, 
\tiny{$ \left(21, 21, 4, 3\right) $}, 
\tiny{$ \left(19, 15, 15\right) $}, 
\tiny{$ \left(19, 17, 13\right) $}, 
\tiny{$ \left(22, 21, 3, 3\right) $}, 
\tiny{$ \left(21, 21, 5, 2\right) $}, 
\tiny{$ \left(19, 19, 11\right) $}, 
\tiny{$ \left(20, 15, 14\right) $}, 
\tiny{$ \left(23, 22, 2, 2\right) $}, 
\tiny{$ \left(22, 22, 4, 1\right) $}, 
\tiny{$ \left(21, 21, 7\right) $}, 
\tiny{$ \left(24, 23, 1, 1\right) $}, 
\tiny{$ \left(23, 23, 3\right) $}, 
\tiny{$ \left(25, 23, 1\right) $}, 
\tiny{$ \left(25, 24\right) $}, 
\tiny{$ \left(14, 14, 14, 14\right) $}, 
\tiny{$ \left(15, 14, 14, 13\right) $}, 
\tiny{$ \left(15, 15, 13, 13\right) $}, 
\tiny{$ \left(16, 14, 13, 13\right) $}, 
\tiny{$ \left(15, 15, 14, 12\right) $}, 
\tiny{$ \left(17, 13, 13, 13\right) $}, 
\tiny{$ \left(16, 14, 14, 12\right) $}, 
\tiny{$ \left(16, 15, 13, 12\right) $}, 
\tiny{$ \left(16, 16, 12, 12\right) $}, 
\tiny{$ \left(15, 15, 15, 11\right) $}, 
\tiny{$ \left(17, 14, 13, 12\right) $}, 
\tiny{$ \left(17, 15, 12, 12\right) $}, 
\tiny{$ \left(16, 15, 14, 11\right) $}, 
\tiny{$ \left(16, 16, 13, 11\right) $}, 
\tiny{$ \left(18, 13, 13, 12\right) $}, 
\tiny{$ \left(18, 14, 12, 12\right) $}, 
\tiny{$ \left(17, 14, 14, 11\right) $}, 
\tiny{$ \left(17, 15, 13, 11\right) $}, 
\tiny{$ \left(17, 16, 12, 11\right) $}, 
\tiny{$ \left(16, 15, 15, 10\right) $}, 
\tiny{$ \left(17, 17, 11, 11\right) $}, 
\tiny{$ \left(16, 16, 14, 10\right) $}, 
\tiny{$ \left(19, 13, 12, 12\right) $}, 
\tiny{$ \left(18, 14, 13, 11\right) $}, 
\tiny{$ \left(18, 15, 12, 11\right) $}, 
\tiny{$ \left(18, 16, 11, 11\right) $}, 
\tiny{$ \left(17, 15, 14, 10\right) $}, 
\tiny{$ \left(17, 16, 13, 10\right) $}, 
\tiny{$ \left(17, 17, 12, 10\right) $}, 
\tiny{$ \left(16, 16, 15, 9\right) $}, 
\tiny{$ \left(19, 13, 13, 11\right) $}, 
\tiny{$ \left(19, 15, 11, 11\right) $}, 
\tiny{$ \left(18, 15, 13, 10\right) $}, 
\tiny{$ \left(17, 15, 15, 9\right) $}, 
\tiny{$ \left(18, 17, 11, 10\right) $}, 
\tiny{$ \left(18, 18, 10, 10\right) $}, 
\tiny{$ \left(17, 17, 13, 9\right) $}, 
\tiny{$ \left(20, 14, 11, 11\right) $}, 
\tiny{$ \left(19, 14, 13, 10\right) $}, 
\tiny{$ \left(19, 15, 12, 10\right) $}, 
\tiny{$ \left(19, 16, 11, 10\right) $}, 
\tiny{$ \left(18, 15, 14, 9\right) $}, 
\tiny{$ \left(19, 17, 10, 10\right) $}, 
\tiny{$ \left(18, 16, 13, 9\right) $}, 
\tiny{$ \left(18, 17, 12, 9\right) $}, 
\tiny{$ \left(17, 16, 15, 8\right) $}, 
\tiny{$ \left(18, 18, 11, 9\right) $}, 
\tiny{$ \left(17, 17, 14, 8\right) $}, 
\tiny{$ \left(20, 15, 11, 10\right) $}, 
\tiny{$ \left(19, 14, 14, 9\right) $}, 
\tiny{$ \left(19, 15, 13, 9\right) $}, 
\tiny{$ \left(18, 15, 15, 8\right) $}, 
\tiny{$ \left(18, 17, 13, 8\right) $}, 
\tiny{$ \left(17, 16, 16, 7\right) $}, 
\tiny{$ \left(19, 19, 9, 9\right) $}, 
\tiny{$ \left(17, 17, 15, 7\right) $}, 
\tiny{$ \left(21, 15, 10, 10\right) $}, 
\tiny{$ \left(19, 18, 11, 8\right) $}, 
\tiny{$ \left(19, 19, 10, 8\right) $}, 
\tiny{$ \left(17, 17, 16, 6\right) $}, 
\tiny{$ \left(20, 15, 13, 8\right) $}, 
\tiny{$ \left(19, 15, 15, 7\right) $}, 
\tiny{$ \left(20, 19, 9, 8\right) $}, 
\tiny{$ \left(18, 17, 15, 6\right) $}, 
\tiny{$ \left(19, 19, 11, 7\right) $}, 
\tiny{$ \left(17, 17, 17, 5\right) $}, 
\tiny{$ \left(18, 17, 16, 5\right) $}, 
\tiny{$ \left(20, 15, 15, 6\right) $}, 
\tiny{$ \left(18, 17, 17, 4\right) $}, 
\tiny{$ \left(21, 21, 7, 7\right) $}, 
\tiny{$ \left(19, 17, 16, 4\right) $}, 
\tiny{$ \left(18, 18, 17, 3\right) $}, 
\tiny{$ \left(19, 17, 17, 3\right) $}, 
\tiny{$ \left(19, 18, 16, 3\right) $}, 
\tiny{$ \left(18, 18, 18, 2\right) $}, 
\tiny{$ \left(19, 18, 17, 2\right) $}, 
\tiny{$ \left(19, 19, 16, 2\right) $}, 
\tiny{$ \left(20, 17, 17, 2\right) $}, 
\tiny{$ \left(19, 18, 18, 1\right) $}, 
\tiny{$ \left(19, 19, 17, 1\right) $}, 
\tiny{$ \left(20, 18, 17, 1\right) $}, 
\tiny{$ \left(20, 19, 16, 1\right) $}, 
\tiny{$ \left(19, 19, 18\right) $}, 
\tiny{$ \left(21, 17, 17, 1\right) $}, 
\tiny{$ \left(20, 18, 18\right) $}, 
\tiny{$ \left(20, 19, 17\right) $}, 
\tiny{$ \left(21, 19, 16\right) $}, 
\tiny{$ \left(22, 17, 17\right) $}, 
\tiny{$ \left(28, 26, 1, 1\right) $}, 
\tiny{$ \left(27, 27, 2\right) $}, 
\tiny{$ \left(29, 27\right) $}, 
\tiny{$ \left(16, 16, 16, 15\right) $}, 
\tiny{$ \left(17, 16, 15, 15\right) $}, 
\tiny{$ \left(18, 15, 15, 15\right) $}, 
\tiny{$ \left(17, 16, 16, 14\right) $}, 
\tiny{$ \left(17, 17, 15, 14\right) $}, 
\tiny{$ \left(18, 16, 15, 14\right) $}, 
\tiny{$ \left(18, 17, 14, 14\right) $}, 
\tiny{$ \left(17, 17, 16, 13\right) $}, 
\tiny{$ \left(19, 15, 15, 14\right) $}, 
\tiny{$ \left(19, 16, 14, 14\right) $}, 
\tiny{$ \left(18, 16, 16, 13\right) $}, 
\tiny{$ \left(18, 17, 15, 13\right) $}, 
\tiny{$ \left(18, 18, 14, 13\right) $}, 
\tiny{$ \left(17, 17, 17, 12\right) $}, 
\tiny{$ \left(20, 15, 14, 14\right) $}, 
\tiny{$ \left(19, 16, 15, 13\right) $}, 
\tiny{$ \left(19, 17, 14, 13\right) $}, 
\tiny{$ \left(19, 18, 13, 13\right) $}, 
\tiny{$ \left(18, 17, 16, 12\right) $}, 
\tiny{$ \left(18, 18, 15, 12\right) $}, 
\tiny{$ \left(20, 15, 15, 13\right) $}, 
\tiny{$ \left(20, 17, 13, 13\right) $}, 
\tiny{$ \left(19, 17, 15, 12\right) $}, 
\tiny{$ \left(18, 17, 17, 11\right) $}, 
\tiny{$ \left(19, 19, 13, 12\right) $}, 
\tiny{$ \left(20, 19, 12, 12\right) $}, 
\tiny{$ \left(19, 19, 14, 11\right) $}, 
\tiny{$ \left(21, 15, 15, 12\right) $}, 
\tiny{$ \left(20, 17, 15, 11\right) $}, 
\tiny{$ \left(19, 17, 17, 10\right) $}, 
\tiny{$ \left(19, 19, 19, 6\right) $}, 
\tiny{$ \left(20, 19, 19, 5\right) $}, 
\tiny{$ \left(21, 19, 19, 4\right) $}, 
\tiny{$ \left(20, 20, 20, 3\right) $}, 
\tiny{$ \left(21, 20, 19, 3\right) $}, 
\tiny{$ \left(21, 20, 20, 2\right) $}, 
\tiny{$ \left(21, 21, 19, 2\right) $}, 
\tiny{$ \left(21, 21, 20, 1\right) $}, 
\tiny{$ \left(22, 20, 20, 1\right) $}, 
\tiny{$ \left(22, 21, 19, 1\right) $}, 
\tiny{$ \left(21, 21, 21\right) $}, 
\tiny{$ \left(22, 21, 20\right) $}, 
\tiny{$ \left(23, 21, 19\right) $}, 
\tiny{$ \left(31, 31, 1\right) $}, 
\tiny{$ \left(18, 18, 17, 17\right) $}, 
\tiny{$ \left(19, 17, 17, 17\right) $}, 
\tiny{$ \left(18, 18, 18, 16\right) $}, 
\tiny{$ \left(19, 18, 17, 16\right) $}, 
\tiny{$ \left(19, 19, 16, 16\right) $}, 
\tiny{$ \left(20, 17, 17, 16\right) $}, 
\tiny{$ \left(20, 18, 16, 16\right) $}, 
\tiny{$ \left(19, 18, 18, 15\right) $}, 
\tiny{$ \left(19, 19, 17, 15\right) $}, 
\tiny{$ \left(21, 17, 16, 16\right) $}, 
\tiny{$ \left(20, 18, 17, 15\right) $}, 
\tiny{$ \left(20, 19, 16, 15\right) $}, 
\tiny{$ \left(20, 20, 15, 15\right) $}, 
\tiny{$ \left(19, 19, 18, 14\right) $}, 
\tiny{$ \left(21, 19, 15, 15\right) $}, 
\tiny{$ \left(21, 21, 14, 14\right) $}, 
\tiny{$ \left(23, 22, 22, 3\right) $}, 
\tiny{$ \left(23, 23, 22, 2\right) $}, 
\tiny{$ \left(23, 23, 23, 1\right) $}, 
\tiny{$ \left(24, 23, 22, 1\right) $}, 
\tiny{$ \left(24, 23, 23\right) $}, 
\tiny{$ \left(25, 23, 22\right) $}, 
\tiny{$ \left(34, 34, 1, 1\right) $}, 
\tiny{$ \left(20, 19, 19, 19\right) $}, 
\tiny{$ \left(20, 20, 19, 18\right) $}, 
\tiny{$ \left(21, 19, 19, 18\right) $}, 
\tiny{$ \left(21, 20, 18, 18\right) $}, 
\tiny{$ \left(20, 20, 20, 17\right) $}, 
\tiny{$ \left(22, 19, 18, 18\right) $}, 
\tiny{$ \left(21, 20, 19, 17\right) $}, 
\tiny{$ \left(21, 21, 18, 17\right) $}, 
\tiny{$ \left(22, 21, 17, 17\right) $}, 
\tiny{$ \left(25, 25, 25, 2\right) $}, 
\tiny{$ \left(26, 25, 25, 1\right) $}, 
\tiny{$ \left(27, 25, 25\right) $}, 
\tiny{$ \left(21, 21, 21, 21\right) $}, 
\tiny{$ \left(22, 21, 21, 20\right) $}, 
\tiny{$ \left(22, 22, 20, 20\right) $}, 
\tiny{$ \left(23, 21, 20, 20\right) $}, 
\tiny{$ \left(22, 22, 21, 19\right) $}, 
\tiny{$ \left(23, 23, 19, 19\right) $}, 
\tiny{$ \left(23, 23, 23, 22\right) $}, 
\tiny{$ \left(24, 23, 22, 22\right) $}, 
\tiny{$ \left(25, 25, 24, 24\right) $}, 
\tiny{$ \left(49, 49\right)\}$}.

\end{minipage}

Here we truncated trailing zeros from the 4-partitions.
The set $X$ is the set of generators of the semigroup of 4-partitions $\mu$ that have \mbox{$\pl \mu d 7 >0$}.
\end{proposition}

\section{Tableau computation results for Proposition~\ref{pro:computercalc}} \label{sec:computer_resultscomputercalc}
In the following we list the actual tableaux from our computer computations described in the proof of Proposition~\ref{pro:computercalc}.

We start with the case $n=6$, $m=3$.
The only partition where we were not able to generate all the semistandard tableaux was $\la = (45,45)$, so we generated random semistandard tableaux and tested only those which guaranteed fast evaluation.

\noindent\begin{minipage}{5cm}
\begin{minipage}{1.7cm}$(6)$\end{minipage}\scalebox{.55}{\def\lr#1{\multicolumn{1}{|@{\hspace{.6ex}}c@{\hspace{.6ex}}|}{\raisebox{-.3ex}{$#1$}}}
\raisebox{-.6ex}{$\begin{array}[b]{*{6}c}\cline{1-6}
\lr{1}&\lr{1}&\lr{1}&\lr{1}&\lr{1}&\lr{1}\\\cline{1-6}
\end{array}$}
}\end{minipage} \ \begin{minipage}{5cm}
\begin{minipage}{1.7cm}$(6,6)$\end{minipage}\scalebox{.55}{\def\lr#1{\multicolumn{1}{|@{\hspace{.6ex}}c@{\hspace{.6ex}}|}{\raisebox{-.3ex}{$#1$}}}
\raisebox{-.6ex}{$\begin{array}[b]{*{6}c}\cline{1-6}
\lr{1}&\lr{1}&\lr{1}&\lr{1}&\lr{1}&\lr{1}\\\cline{1-6}
\lr{2}&\lr{2}&\lr{2}&\lr{2}&\lr{2}&\lr{2}\\\cline{1-6}
\end{array}$}
}\end{minipage} \ \begin{minipage}{5cm}
\begin{minipage}{1.7cm}$(8,4)$\end{minipage}\scalebox{.55}{\def\lr#1{\multicolumn{1}{|@{\hspace{.6ex}}c@{\hspace{.6ex}}|}{\raisebox{-.3ex}{$#1$}}}
\raisebox{-.6ex}{$\begin{array}[b]{*{8}c}\cline{1-8}
\lr{1}&\lr{1}&\lr{1}&\lr{1}&\lr{1}&\lr{1}&\lr{2}&\lr{2}\\\cline{1-8}
\lr{2}&\lr{2}&\lr{2}&\lr{2}\\\cline{1-4}
\end{array}$}
}\end{minipage} \ \begin{minipage}{5cm}
\begin{minipage}{1.7cm}$(10,2)$\end{minipage}\scalebox{.55}{\def\lr#1{\multicolumn{1}{|@{\hspace{.6ex}}c@{\hspace{.6ex}}|}{\raisebox{-.3ex}{$#1$}}}
\raisebox{-.6ex}{$\begin{array}[b]{*{10}c}\cline{1-10}
\lr{1}&\lr{1}&\lr{1}&\lr{1}&\lr{1}&\lr{1}&\lr{2}&\lr{2}&\lr{2}&\lr{2}\\\cline{1-10}
\lr{2}&\lr{2}\\\cline{1-2}
\end{array}$}
}\end{minipage} \ \begin{minipage}{5cm}
\begin{minipage}{1.7cm}$(6,6,6)$\end{minipage}\scalebox{.55}{\def\lr#1{\multicolumn{1}{|@{\hspace{.6ex}}c@{\hspace{.6ex}}|}{\raisebox{-.3ex}{$#1$}}}
\raisebox{-.6ex}{$\begin{array}[b]{*{6}c}\cline{1-6}
\lr{1}&\lr{1}&\lr{1}&\lr{1}&\lr{1}&\lr{1}\\\cline{1-6}
\lr{2}&\lr{2}&\lr{2}&\lr{2}&\lr{2}&\lr{2}\\\cline{1-6}
\lr{3}&\lr{3}&\lr{3}&\lr{3}&\lr{3}&\lr{3}\\\cline{1-6}
\end{array}$}
}\end{minipage} \ \begin{minipage}{5cm}
\begin{minipage}{1.7cm}$(8,6,4)$\end{minipage}\scalebox{.55}{\def\lr#1{\multicolumn{1}{|@{\hspace{.6ex}}c@{\hspace{.6ex}}|}{\raisebox{-.3ex}{$#1$}}}
\raisebox{-.6ex}{$\begin{array}[b]{*{8}c}\cline{1-8}
\lr{1}&\lr{1}&\lr{1}&\lr{1}&\lr{1}&\lr{1}&\lr{2}&\lr{2}\\\cline{1-8}
\lr{2}&\lr{2}&\lr{2}&\lr{2}&\lr{3}&\lr{3}\\\cline{1-6}
\lr{3}&\lr{3}&\lr{3}&\lr{3}\\\cline{1-4}
\end{array}$}
}\end{minipage} \ \begin{minipage}{5cm}
\begin{minipage}{1.7cm}$(10,4,4)$\end{minipage}\scalebox{.55}{\def\lr#1{\multicolumn{1}{|@{\hspace{.6ex}}c@{\hspace{.6ex}}|}{\raisebox{-.3ex}{$#1$}}}
\raisebox{-.6ex}{$\begin{array}[b]{*{10}c}\cline{1-10}
\lr{1}&\lr{1}&\lr{1}&\lr{1}&\lr{1}&\lr{1}&\lr{2}&\lr{2}&\lr{3}&\lr{3}\\\cline{1-10}
\lr{2}&\lr{2}&\lr{2}&\lr{2}\\\cline{1-4}
\lr{3}&\lr{3}&\lr{3}&\lr{3}\\\cline{1-4}
\end{array}$}
}\end{minipage} \ \begin{minipage}{5cm}
\begin{minipage}{1.7cm}$(9,6,3)$\end{minipage}\scalebox{.55}{\def\lr#1{\multicolumn{1}{|@{\hspace{.6ex}}c@{\hspace{.6ex}}|}{\raisebox{-.3ex}{$#1$}}}
\raisebox{-.6ex}{$\begin{array}[b]{*{9}c}\cline{1-9}
\lr{1}&\lr{1}&\lr{1}&\lr{1}&\lr{1}&\lr{1}&\lr{2}&\lr{2}&\lr{2}\\\cline{1-9}
\lr{2}&\lr{2}&\lr{2}&\lr{3}&\lr{3}&\lr{3}\\\cline{1-6}
\lr{3}&\lr{3}&\lr{3}\\\cline{1-3}
\end{array}$}
}\end{minipage} \ \begin{minipage}{5cm}
\begin{minipage}{1.7cm}$(8,8,2)$\end{minipage}\scalebox{.55}{\def\lr#1{\multicolumn{1}{|@{\hspace{.6ex}}c@{\hspace{.6ex}}|}{\raisebox{-.3ex}{$#1$}}}
\raisebox{-.6ex}{$\begin{array}[b]{*{8}c}\cline{1-8}
\lr{1}&\lr{1}&\lr{1}&\lr{1}&\lr{1}&\lr{1}&\lr{2}&\lr{2}\\\cline{1-8}
\lr{2}&\lr{2}&\lr{2}&\lr{2}&\lr{3}&\lr{3}&\lr{3}&\lr{3}\\\cline{1-8}
\lr{3}&\lr{3}\\\cline{1-2}
\end{array}$}
}\end{minipage} \ \begin{minipage}{5cm}
\begin{minipage}{1.7cm}$(10,6,2)$\end{minipage}\scalebox{.55}{\def\lr#1{\multicolumn{1}{|@{\hspace{.6ex}}c@{\hspace{.6ex}}|}{\raisebox{-.3ex}{$#1$}}}
\raisebox{-.6ex}{$\begin{array}[b]{*{10}c}\cline{1-10}
\lr{1}&\lr{1}&\lr{1}&\lr{1}&\lr{1}&\lr{1}&\lr{2}&\lr{2}&\lr{2}&\lr{2}\\\cline{1-10}
\lr{2}&\lr{2}&\lr{3}&\lr{3}&\lr{3}&\lr{3}\\\cline{1-6}
\lr{3}&\lr{3}\\\cline{1-2}
\end{array}$}
}\end{minipage} \ \begin{minipage}{5cm}
\begin{minipage}{1.7cm}$(11,5,2)$\end{minipage}\scalebox{.55}{\def\lr#1{\multicolumn{1}{|@{\hspace{.6ex}}c@{\hspace{.6ex}}|}{\raisebox{-.3ex}{$#1$}}}
\raisebox{-.6ex}{$\begin{array}[b]{*{11}c}\cline{1-11}
\lr{1}&\lr{1}&\lr{1}&\lr{1}&\lr{1}&\lr{1}&\lr{2}&\lr{2}&\lr{3}&\lr{3}&\lr{3}\\\cline{1-11}
\lr{2}&\lr{2}&\lr{2}&\lr{2}&\lr{3}\\\cline{1-5}
\lr{3}&\lr{3}\\\cline{1-2}
\end{array}$}
}\end{minipage} \ \begin{minipage}{5cm}
\begin{minipage}{1.7cm}$(10,7,1)$\end{minipage}\scalebox{.55}{\def\lr#1{\multicolumn{1}{|@{\hspace{.6ex}}c@{\hspace{.6ex}}|}{\raisebox{-.3ex}{$#1$}}}
\raisebox{-.6ex}{$\begin{array}[b]{*{10}c}\cline{1-10}
\lr{1}&\lr{1}&\lr{1}&\lr{1}&\lr{1}&\lr{1}&\lr{2}&\lr{2}&\lr{3}&\lr{3}\\\cline{1-10}
\lr{2}&\lr{2}&\lr{2}&\lr{2}&\lr{3}&\lr{3}&\lr{3}\\\cline{1-7}
\lr{3}\\\cline{1-1}
\end{array}$}
}\end{minipage} \ \begin{minipage}{5cm}
\begin{minipage}{1.7cm}$(12,4,2)$\end{minipage}\scalebox{.55}{\def\lr#1{\multicolumn{1}{|@{\hspace{.6ex}}c@{\hspace{.6ex}}|}{\raisebox{-.3ex}{$#1$}}}
\raisebox{-.6ex}{$\begin{array}[b]{*{12}c}\cline{1-12}
\lr{1}&\lr{1}&\lr{1}&\lr{1}&\lr{1}&\lr{1}&\lr{2}&\lr{2}&\lr{2}&\lr{2}&\lr{3}&\lr{3}\\\cline{1-12}
\lr{2}&\lr{2}&\lr{3}&\lr{3}\\\cline{1-4}
\lr{3}&\lr{3}\\\cline{1-2}
\end{array}$}
}\end{minipage} \ \begin{minipage}{5cm}
\begin{minipage}{1.7cm}$(11,6,1)$\end{minipage}\scalebox{.55}{\def\lr#1{\multicolumn{1}{|@{\hspace{.6ex}}c@{\hspace{.6ex}}|}{\raisebox{-.3ex}{$#1$}}}
\raisebox{-.6ex}{$\begin{array}[b]{*{11}c}\cline{1-11}
\lr{1}&\lr{1}&\lr{1}&\lr{1}&\lr{1}&\lr{1}&\lr{2}&\lr{2}&\lr{2}&\lr{2}&\lr{2}\\\cline{1-11}
\lr{2}&\lr{3}&\lr{3}&\lr{3}&\lr{3}&\lr{3}\\\cline{1-6}
\lr{3}\\\cline{1-1}
\end{array}$}
}\end{minipage} \ \begin{minipage}{5cm}
\begin{minipage}{1.7cm}$(10,8)$\end{minipage}\scalebox{.55}{\def\lr#1{\multicolumn{1}{|@{\hspace{.6ex}}c@{\hspace{.6ex}}|}{\raisebox{-.3ex}{$#1$}}}
\raisebox{-.6ex}{$\begin{array}[b]{*{10}c}\cline{1-10}
\lr{1}&\lr{1}&\lr{1}&\lr{1}&\lr{1}&\lr{1}&\lr{2}&\lr{2}&\lr{3}&\lr{3}\\\cline{1-10}
\lr{2}&\lr{2}&\lr{2}&\lr{2}&\lr{3}&\lr{3}&\lr{3}&\lr{3}\\\cline{1-8}
\end{array}$}
}\end{minipage} \ \begin{minipage}{5cm}
\begin{minipage}{1.7cm}$(14,2,2)$\end{minipage}\scalebox{.55}{\def\lr#1{\multicolumn{1}{|@{\hspace{.6ex}}c@{\hspace{.6ex}}|}{\raisebox{-.3ex}{$#1$}}}
\raisebox{-.6ex}{$\begin{array}[b]{*{14}c}\cline{1-14}
\lr{1}&\lr{1}&\lr{1}&\lr{1}&\lr{1}&\lr{1}&\lr{2}&\lr{2}&\lr{2}&\lr{2}&\lr{3}&\lr{3}&\lr{3}&\lr{3}\\\cline{1-14}
\lr{2}&\lr{2}\\\cline{1-2}
\lr{3}&\lr{3}\\\cline{1-2}
\end{array}$}
}\end{minipage} \ \begin{minipage}{5cm}
\begin{minipage}{1.7cm}$(13,4,1)$\end{minipage}\scalebox{.55}{\def\lr#1{\multicolumn{1}{|@{\hspace{.6ex}}c@{\hspace{.6ex}}|}{\raisebox{-.3ex}{$#1$}}}
\raisebox{-.6ex}{$\begin{array}[b]{*{13}c}\cline{1-13}
\lr{1}&\lr{1}&\lr{1}&\lr{1}&\lr{1}&\lr{1}&\lr{2}&\lr{2}&\lr{2}&\lr{2}&\lr{2}&\lr{3}&\lr{3}\\\cline{1-13}
\lr{2}&\lr{3}&\lr{3}&\lr{3}\\\cline{1-4}
\lr{3}\\\cline{1-1}
\end{array}$}
}\end{minipage} \ \begin{minipage}{5cm}
\begin{minipage}{1.7cm}$(13,5)$\end{minipage}\scalebox{.55}{\def\lr#1{\multicolumn{1}{|@{\hspace{.6ex}}c@{\hspace{.6ex}}|}{\raisebox{-.3ex}{$#1$}}}
\raisebox{-.6ex}{$\begin{array}[b]{*{13}c}\cline{1-13}
\lr{1}&\lr{1}&\lr{1}&\lr{1}&\lr{1}&\lr{1}&\lr{2}&\lr{2}&\lr{3}&\lr{3}&\lr{3}&\lr{3}&\lr{3}\\\cline{1-13}
\lr{2}&\lr{2}&\lr{2}&\lr{2}&\lr{3}\\\cline{1-5}
\end{array}$}
}\end{minipage} \ \begin{minipage}{5cm}
\begin{minipage}{1.7cm}$(15,3)$\end{minipage}\scalebox{.55}{\def\lr#1{\multicolumn{1}{|@{\hspace{.6ex}}c@{\hspace{.6ex}}|}{\raisebox{-.3ex}{$#1$}}}
\raisebox{-.6ex}{$\begin{array}[b]{*{15}c}\cline{1-15}
\lr{1}&\lr{1}&\lr{1}&\lr{1}&\lr{1}&\lr{1}&\lr{2}&\lr{2}&\lr{2}&\lr{2}&\lr{3}&\lr{3}&\lr{3}&\lr{3}&\lr{3}\\\cline{1-15}
\lr{2}&\lr{2}&\lr{3}\\\cline{1-3}
\end{array}$}
}\end{minipage} \ \begin{minipage}{5cm}
\begin{minipage}{1.7cm}$(8,8,8)$\end{minipage}\scalebox{.55}{\def\lr#1{\multicolumn{1}{|@{\hspace{.6ex}}c@{\hspace{.6ex}}|}{\raisebox{-.3ex}{$#1$}}}
\raisebox{-.6ex}{$\begin{array}[b]{*{8}c}\cline{1-8}
\lr{1}&\lr{1}&\lr{1}&\lr{1}&\lr{1}&\lr{1}&\lr{2}&\lr{2}\\\cline{1-8}
\lr{2}&\lr{2}&\lr{2}&\lr{2}&\lr{3}&\lr{3}&\lr{3}&\lr{3}\\\cline{1-8}
\lr{3}&\lr{3}&\lr{4}&\lr{4}&\lr{4}&\lr{4}&\lr{4}&\lr{4}\\\cline{1-8}
\end{array}$}
}\end{minipage} \ \begin{minipage}{5cm}
\begin{minipage}{1.7cm}$(10,8,6)$\end{minipage}\scalebox{.55}{\def\lr#1{\multicolumn{1}{|@{\hspace{.6ex}}c@{\hspace{.6ex}}|}{\raisebox{-.3ex}{$#1$}}}
\raisebox{-.6ex}{$\begin{array}[b]{*{10}c}\cline{1-10}
\lr{1}&\lr{1}&\lr{1}&\lr{1}&\lr{1}&\lr{1}&\lr{2}&\lr{2}&\lr{3}&\lr{3}\\\cline{1-10}
\lr{2}&\lr{2}&\lr{2}&\lr{2}&\lr{3}&\lr{3}&\lr{3}&\lr{3}\\\cline{1-8}
\lr{4}&\lr{4}&\lr{4}&\lr{4}&\lr{4}&\lr{4}\\\cline{1-6}
\end{array}$}
}\end{minipage} \ \begin{minipage}{5cm}
\begin{minipage}{1.7cm}$(11,7,6)$\end{minipage}\scalebox{.55}{\def\lr#1{\multicolumn{1}{|@{\hspace{.6ex}}c@{\hspace{.6ex}}|}{\raisebox{-.3ex}{$#1$}}}
\raisebox{-.6ex}{$\begin{array}[b]{*{11}c}\cline{1-11}
\lr{1}&\lr{1}&\lr{1}&\lr{1}&\lr{1}&\lr{1}&\lr{2}&\lr{2}&\lr{3}&\lr{3}&\lr{3}\\\cline{1-11}
\lr{2}&\lr{2}&\lr{2}&\lr{2}&\lr{3}&\lr{3}&\lr{3}\\\cline{1-7}
\lr{4}&\lr{4}&\lr{4}&\lr{4}&\lr{4}&\lr{4}\\\cline{1-6}
\end{array}$}
}\end{minipage} \ \begin{minipage}{5cm}
\begin{minipage}{1.7cm}$(10,9,5)$\end{minipage}\scalebox{.55}{\def\lr#1{\multicolumn{1}{|@{\hspace{.6ex}}c@{\hspace{.6ex}}|}{\raisebox{-.3ex}{$#1$}}}
\raisebox{-.6ex}{$\begin{array}[b]{*{10}c}\cline{1-10}
\lr{1}&\lr{1}&\lr{1}&\lr{1}&\lr{1}&\lr{1}&\lr{2}&\lr{2}&\lr{2}&\lr{3}\\\cline{1-10}
\lr{2}&\lr{2}&\lr{2}&\lr{3}&\lr{3}&\lr{4}&\lr{4}&\lr{4}&\lr{4}\\\cline{1-9}
\lr{3}&\lr{3}&\lr{3}&\lr{4}&\lr{4}\\\cline{1-5}
\end{array}$}
}\end{minipage} \ \begin{minipage}{5cm}
\begin{minipage}{1.7cm}$(11,8,5)$\end{minipage}\scalebox{.55}{\def\lr#1{\multicolumn{1}{|@{\hspace{.6ex}}c@{\hspace{.6ex}}|}{\raisebox{-.3ex}{$#1$}}}
\raisebox{-.6ex}{$\begin{array}[b]{*{11}c}\cline{1-11}
\lr{1}&\lr{1}&\lr{1}&\lr{1}&\lr{1}&\lr{1}&\lr{2}&\lr{2}&\lr{4}&\lr{4}&\lr{4}\\\cline{1-11}
\lr{2}&\lr{2}&\lr{2}&\lr{2}&\lr{3}&\lr{3}&\lr{4}&\lr{4}\\\cline{1-8}
\lr{3}&\lr{3}&\lr{3}&\lr{3}&\lr{4}\\\cline{1-5}
\end{array}$}
}\end{minipage} \ \begin{minipage}{5cm}
\begin{minipage}{1.7cm}$(10,10,4)$\end{minipage}\scalebox{.55}{\def\lr#1{\multicolumn{1}{|@{\hspace{.6ex}}c@{\hspace{.6ex}}|}{\raisebox{-.3ex}{$#1$}}}
\raisebox{-.6ex}{$\begin{array}[b]{*{10}c}\cline{1-10}
\lr{1}&\lr{1}&\lr{1}&\lr{1}&\lr{1}&\lr{1}&\lr{2}&\lr{2}&\lr{3}&\lr{3}\\\cline{1-10}
\lr{2}&\lr{2}&\lr{2}&\lr{2}&\lr{3}&\lr{3}&\lr{3}&\lr{3}&\lr{4}&\lr{4}\\\cline{1-10}
\lr{4}&\lr{4}&\lr{4}&\lr{4}\\\cline{1-4}
\end{array}$}
}\end{minipage} \ \begin{minipage}{5cm}
\begin{minipage}{1.7cm}$(12,7,5)$\end{minipage}\scalebox{.55}{\def\lr#1{\multicolumn{1}{|@{\hspace{.6ex}}c@{\hspace{.6ex}}|}{\raisebox{-.3ex}{$#1$}}}
\raisebox{-.6ex}{$\begin{array}[b]{*{12}c}\cline{1-12}
\lr{1}&\lr{1}&\lr{1}&\lr{1}&\lr{1}&\lr{1}&\lr{2}&\lr{2}&\lr{3}&\lr{3}&\lr{3}&\lr{3}\\\cline{1-12}
\lr{2}&\lr{2}&\lr{2}&\lr{2}&\lr{3}&\lr{3}&\lr{4}\\\cline{1-7}
\lr{4}&\lr{4}&\lr{4}&\lr{4}&\lr{4}\\\cline{1-5}
\end{array}$}
}\end{minipage} \ \begin{minipage}{5cm}
\begin{minipage}{1.7cm}$(11,9,4)$\end{minipage}\scalebox{.55}{\def\lr#1{\multicolumn{1}{|@{\hspace{.6ex}}c@{\hspace{.6ex}}|}{\raisebox{-.3ex}{$#1$}}}
\raisebox{-.6ex}{$\begin{array}[b]{*{11}c}\cline{1-11}
\lr{1}&\lr{1}&\lr{1}&\lr{1}&\lr{1}&\lr{1}&\lr{2}&\lr{2}&\lr{3}&\lr{3}&\lr{4}\\\cline{1-11}
\lr{2}&\lr{2}&\lr{2}&\lr{2}&\lr{3}&\lr{3}&\lr{3}&\lr{3}&\lr{4}\\\cline{1-9}
\lr{4}&\lr{4}&\lr{4}&\lr{4}\\\cline{1-4}
\end{array}$}
}\end{minipage} \ \begin{minipage}{5cm}
\begin{minipage}{1.7cm}$(13,6,5)$\end{minipage}\scalebox{.55}{\def\lr#1{\multicolumn{1}{|@{\hspace{.6ex}}c@{\hspace{.6ex}}|}{\raisebox{-.3ex}{$#1$}}}
\raisebox{-.6ex}{$\begin{array}[b]{*{13}c}\cline{1-13}
\lr{1}&\lr{1}&\lr{1}&\lr{1}&\lr{1}&\lr{1}&\lr{2}&\lr{2}&\lr{3}&\lr{3}&\lr{3}&\lr{3}&\lr{3}\\\cline{1-13}
\lr{2}&\lr{2}&\lr{2}&\lr{2}&\lr{3}&\lr{4}\\\cline{1-6}
\lr{4}&\lr{4}&\lr{4}&\lr{4}&\lr{4}\\\cline{1-5}
\end{array}$}
}\end{minipage} \ \begin{minipage}{5cm}
\begin{minipage}{1.7cm}$(12,8,4)$\end{minipage}\scalebox{.55}{\def\lr#1{\multicolumn{1}{|@{\hspace{.6ex}}c@{\hspace{.6ex}}|}{\raisebox{-.3ex}{$#1$}}}
\raisebox{-.6ex}{$\begin{array}[b]{*{12}c}\cline{1-12}
\lr{1}&\lr{1}&\lr{1}&\lr{1}&\lr{1}&\lr{1}&\lr{2}&\lr{2}&\lr{2}&\lr{2}&\lr{2}&\lr{2}\\\cline{1-12}
\lr{3}&\lr{3}&\lr{3}&\lr{3}&\lr{3}&\lr{3}&\lr{4}&\lr{4}\\\cline{1-8}
\lr{4}&\lr{4}&\lr{4}&\lr{4}\\\cline{1-4}
\end{array}$}
}\end{minipage} \ \begin{minipage}{5cm}
\begin{minipage}{1.7cm}$(11,10,3)$\end{minipage}\scalebox{.55}{\def\lr#1{\multicolumn{1}{|@{\hspace{.6ex}}c@{\hspace{.6ex}}|}{\raisebox{-.3ex}{$#1$}}}
\raisebox{-.6ex}{$\begin{array}[b]{*{11}c}\cline{1-11}
\lr{1}&\lr{1}&\lr{1}&\lr{1}&\lr{1}&\lr{1}&\lr{2}&\lr{2}&\lr{2}&\lr{2}&\lr{4}\\\cline{1-11}
\lr{2}&\lr{2}&\lr{3}&\lr{3}&\lr{3}&\lr{3}&\lr{4}&\lr{4}&\lr{4}&\lr{4}\\\cline{1-10}
\lr{3}&\lr{3}&\lr{4}\\\cline{1-3}
\end{array}$}
}\end{minipage} \ \begin{minipage}{5cm}
\begin{minipage}{1.7cm}$(13,7,4)$\end{minipage}\scalebox{.55}{\def\lr#1{\multicolumn{1}{|@{\hspace{.6ex}}c@{\hspace{.6ex}}|}{\raisebox{-.3ex}{$#1$}}}
\raisebox{-.6ex}{$\begin{array}[b]{*{13}c}\cline{1-13}
\lr{1}&\lr{1}&\lr{1}&\lr{1}&\lr{1}&\lr{1}&\lr{2}&\lr{2}&\lr{2}&\lr{2}&\lr{2}&\lr{2}&\lr{4}\\\cline{1-13}
\lr{3}&\lr{3}&\lr{3}&\lr{3}&\lr{3}&\lr{3}&\lr{4}\\\cline{1-7}
\lr{4}&\lr{4}&\lr{4}&\lr{4}\\\cline{1-4}
\end{array}$}
}\end{minipage} \ \begin{minipage}{5cm}
\begin{minipage}{1.7cm}$(12,9,3)$\end{minipage}\scalebox{.55}{\def\lr#1{\multicolumn{1}{|@{\hspace{.6ex}}c@{\hspace{.6ex}}|}{\raisebox{-.3ex}{$#1$}}}
\raisebox{-.6ex}{$\begin{array}[b]{*{12}c}\cline{1-12}
\lr{1}&\lr{1}&\lr{1}&\lr{1}&\lr{1}&\lr{1}&\lr{2}&\lr{2}&\lr{2}&\lr{2}&\lr{2}&\lr{2}\\\cline{1-12}
\lr{3}&\lr{3}&\lr{3}&\lr{3}&\lr{3}&\lr{3}&\lr{4}&\lr{4}&\lr{4}\\\cline{1-9}
\lr{4}&\lr{4}&\lr{4}\\\cline{1-3}
\end{array}$}
}\end{minipage} \ \begin{minipage}{5cm}
\begin{minipage}{1.7cm}$(13,8,3)$\end{minipage}\scalebox{.55}{\def\lr#1{\multicolumn{1}{|@{\hspace{.6ex}}c@{\hspace{.6ex}}|}{\raisebox{-.3ex}{$#1$}}}
\raisebox{-.6ex}{$\begin{array}[b]{*{13}c}\cline{1-13}
\lr{1}&\lr{1}&\lr{1}&\lr{1}&\lr{1}&\lr{1}&\lr{2}&\lr{2}&\lr{2}&\lr{2}&\lr{2}&\lr{2}&\lr{4}\\\cline{1-13}
\lr{3}&\lr{3}&\lr{3}&\lr{3}&\lr{3}&\lr{3}&\lr{4}&\lr{4}\\\cline{1-8}
\lr{4}&\lr{4}&\lr{4}\\\cline{1-3}
\end{array}$}
}\end{minipage} \ \begin{minipage}{5cm}
\begin{minipage}{1.7cm}$(12,10,2)$\end{minipage}\scalebox{.55}{\def\lr#1{\multicolumn{1}{|@{\hspace{.6ex}}c@{\hspace{.6ex}}|}{\raisebox{-.3ex}{$#1$}}}
\raisebox{-.6ex}{$\begin{array}[b]{*{12}c}\cline{1-12}
\lr{1}&\lr{1}&\lr{1}&\lr{1}&\lr{1}&\lr{1}&\lr{2}&\lr{2}&\lr{2}&\lr{2}&\lr{2}&\lr{2}\\\cline{1-12}
\lr{3}&\lr{3}&\lr{3}&\lr{3}&\lr{3}&\lr{3}&\lr{4}&\lr{4}&\lr{4}&\lr{4}\\\cline{1-10}
\lr{4}&\lr{4}\\\cline{1-2}
\end{array}$}
}\end{minipage} \ \begin{minipage}{5cm}
\begin{minipage}{1.7cm}$(15,5,4)$\end{minipage}\scalebox{.55}{\def\lr#1{\multicolumn{1}{|@{\hspace{.6ex}}c@{\hspace{.6ex}}|}{\raisebox{-.3ex}{$#1$}}}
\raisebox{-.6ex}{$\begin{array}[b]{*{15}c}\cline{1-15}
\lr{1}&\lr{1}&\lr{1}&\lr{1}&\lr{1}&\lr{1}&\lr{2}&\lr{2}&\lr{3}&\lr{3}&\lr{3}&\lr{3}&\lr{3}&\lr{4}&\lr{4}\\\cline{1-15}
\lr{2}&\lr{2}&\lr{2}&\lr{2}&\lr{3}\\\cline{1-5}
\lr{4}&\lr{4}&\lr{4}&\lr{4}\\\cline{1-4}
\end{array}$}
}\end{minipage} \ \begin{minipage}{5cm}
\begin{minipage}{1.7cm}$(14,7,3)$\end{minipage}\scalebox{.55}{\def\lr#1{\multicolumn{1}{|@{\hspace{.6ex}}c@{\hspace{.6ex}}|}{\raisebox{-.3ex}{$#1$}}}
\raisebox{-.6ex}{$\begin{array}[b]{*{14}c}\cline{1-14}
\lr{1}&\lr{1}&\lr{1}&\lr{1}&\lr{1}&\lr{1}&\lr{2}&\lr{2}&\lr{2}&\lr{2}&\lr{2}&\lr{2}&\lr{4}&\lr{4}\\\cline{1-14}
\lr{3}&\lr{3}&\lr{3}&\lr{3}&\lr{3}&\lr{3}&\lr{4}\\\cline{1-7}
\lr{4}&\lr{4}&\lr{4}\\\cline{1-3}
\end{array}$}
}\end{minipage} \ \begin{minipage}{5cm}
\begin{minipage}{1.7cm}$(13,9,2)$\end{minipage}\scalebox{.55}{\def\lr#1{\multicolumn{1}{|@{\hspace{.6ex}}c@{\hspace{.6ex}}|}{\raisebox{-.3ex}{$#1$}}}
\raisebox{-.6ex}{$\begin{array}[b]{*{13}c}\cline{1-13}
\lr{1}&\lr{1}&\lr{1}&\lr{1}&\lr{1}&\lr{1}&\lr{2}&\lr{2}&\lr{2}&\lr{2}&\lr{2}&\lr{2}&\lr{4}\\\cline{1-13}
\lr{3}&\lr{3}&\lr{3}&\lr{3}&\lr{3}&\lr{3}&\lr{4}&\lr{4}&\lr{4}\\\cline{1-9}
\lr{4}&\lr{4}\\\cline{1-2}
\end{array}$}
}\end{minipage} \ \begin{minipage}{5cm}
\begin{minipage}{1.7cm}$(13,10,1)$\end{minipage}\scalebox{.55}{\def\lr#1{\multicolumn{1}{|@{\hspace{.6ex}}c@{\hspace{.6ex}}|}{\raisebox{-.3ex}{$#1$}}}
\raisebox{-.6ex}{$\begin{array}[b]{*{13}c}\cline{1-13}
\lr{1}&\lr{1}&\lr{1}&\lr{1}&\lr{1}&\lr{1}&\lr{2}&\lr{2}&\lr{2}&\lr{2}&\lr{2}&\lr{2}&\lr{4}\\\cline{1-13}
\lr{3}&\lr{3}&\lr{3}&\lr{3}&\lr{3}&\lr{3}&\lr{4}&\lr{4}&\lr{4}&\lr{4}\\\cline{1-10}
\lr{4}\\\cline{1-1}
\end{array}$}
}\end{minipage} \ \begin{minipage}{5cm}
\begin{minipage}{1.7cm}$(16,5,3)$\end{minipage}\scalebox{.55}{\def\lr#1{\multicolumn{1}{|@{\hspace{.6ex}}c@{\hspace{.6ex}}|}{\raisebox{-.3ex}{$#1$}}}
\raisebox{-.6ex}{$\begin{array}[b]{*{16}c}\cline{1-16}
\lr{1}&\lr{1}&\lr{1}&\lr{1}&\lr{1}&\lr{1}&\lr{2}&\lr{2}&\lr{2}&\lr{2}&\lr{3}&\lr{3}&\lr{3}&\lr{4}&\lr{4}&\lr{4}\\\cline{1-16}
\lr{2}&\lr{2}&\lr{3}&\lr{3}&\lr{3}\\\cline{1-5}
\lr{4}&\lr{4}&\lr{4}\\\cline{1-3}
\end{array}$}
}\end{minipage} \ \begin{minipage}{5cm}
\begin{minipage}{1.7cm}$(15,7,2)$\end{minipage}\scalebox{.55}{\def\lr#1{\multicolumn{1}{|@{\hspace{.6ex}}c@{\hspace{.6ex}}|}{\raisebox{-.3ex}{$#1$}}}
\raisebox{-.6ex}{$\begin{array}[b]{*{15}c}\cline{1-15}
\lr{1}&\lr{1}&\lr{1}&\lr{1}&\lr{1}&\lr{1}&\lr{2}&\lr{2}&\lr{2}&\lr{2}&\lr{2}&\lr{2}&\lr{4}&\lr{4}&\lr{4}\\\cline{1-15}
\lr{3}&\lr{3}&\lr{3}&\lr{3}&\lr{3}&\lr{3}&\lr{4}\\\cline{1-7}
\lr{4}&\lr{4}\\\cline{1-2}
\end{array}$}
}\end{minipage} \ \begin{minipage}{5cm}
\begin{minipage}{1.7cm}$(14,9,1)$\end{minipage}\scalebox{.55}{\def\lr#1{\multicolumn{1}{|@{\hspace{.6ex}}c@{\hspace{.6ex}}|}{\raisebox{-.3ex}{$#1$}}}
\raisebox{-.6ex}{$\begin{array}[b]{*{14}c}\cline{1-14}
\lr{1}&\lr{1}&\lr{1}&\lr{1}&\lr{1}&\lr{1}&\lr{2}&\lr{2}&\lr{2}&\lr{2}&\lr{2}&\lr{2}&\lr{4}&\lr{4}\\\cline{1-14}
\lr{3}&\lr{3}&\lr{3}&\lr{3}&\lr{3}&\lr{3}&\lr{4}&\lr{4}&\lr{4}\\\cline{1-9}
\lr{4}\\\cline{1-1}
\end{array}$}
}\end{minipage} \ \begin{minipage}{5cm}
\begin{minipage}{1.7cm}$(17,4,3)$\end{minipage}\scalebox{.55}{\def\lr#1{\multicolumn{1}{|@{\hspace{.6ex}}c@{\hspace{.6ex}}|}{\raisebox{-.3ex}{$#1$}}}
\raisebox{-.6ex}{$\begin{array}[b]{*{17}c}\cline{1-17}
\lr{1}&\lr{1}&\lr{1}&\lr{1}&\lr{1}&\lr{1}&\lr{2}&\lr{2}&\lr{2}&\lr{2}&\lr{3}&\lr{3}&\lr{3}&\lr{3}&\lr{3}&\lr{4}&\lr{4}\\\cline{1-17}
\lr{2}&\lr{2}&\lr{3}&\lr{4}\\\cline{1-4}
\lr{4}&\lr{4}&\lr{4}\\\cline{1-3}
\end{array}$}
}\end{minipage} \ \begin{minipage}{5cm}
\begin{minipage}{1.7cm}$(15,8,1)$\end{minipage}\scalebox{.55}{\def\lr#1{\multicolumn{1}{|@{\hspace{.6ex}}c@{\hspace{.6ex}}|}{\raisebox{-.3ex}{$#1$}}}
\raisebox{-.6ex}{$\begin{array}[b]{*{15}c}\cline{1-15}
\lr{1}&\lr{1}&\lr{1}&\lr{1}&\lr{1}&\lr{1}&\lr{2}&\lr{2}&\lr{2}&\lr{2}&\lr{2}&\lr{2}&\lr{4}&\lr{4}&\lr{4}\\\cline{1-15}
\lr{3}&\lr{3}&\lr{3}&\lr{3}&\lr{3}&\lr{3}&\lr{4}&\lr{4}\\\cline{1-8}
\lr{4}\\\cline{1-1}
\end{array}$}
}\end{minipage} \ \begin{minipage}{5cm}
\begin{minipage}{1.7cm}$(15,9)$\end{minipage}\scalebox{.55}{\def\lr#1{\multicolumn{1}{|@{\hspace{.6ex}}c@{\hspace{.6ex}}|}{\raisebox{-.3ex}{$#1$}}}
\raisebox{-.6ex}{$\begin{array}[b]{*{15}c}\cline{1-15}
\lr{1}&\lr{1}&\lr{1}&\lr{1}&\lr{1}&\lr{1}&\lr{2}&\lr{2}&\lr{3}&\lr{3}&\lr{3}&\lr{3}&\lr{3}&\lr{3}&\lr{4}\\\cline{1-15}
\lr{2}&\lr{2}&\lr{2}&\lr{2}&\lr{4}&\lr{4}&\lr{4}&\lr{4}&\lr{4}\\\cline{1-9}
\end{array}$}
}\end{minipage} \ \begin{minipage}{5cm}
\begin{minipage}{1.7cm}$(19,3,2)$\end{minipage}\scalebox{.55}{\def\lr#1{\multicolumn{1}{|@{\hspace{.6ex}}c@{\hspace{.6ex}}|}{\raisebox{-.3ex}{$#1$}}}
\raisebox{-.6ex}{$\begin{array}[b]{*{19}c}\cline{1-19}
\lr{1}&\lr{1}&\lr{1}&\lr{1}&\lr{1}&\lr{1}&\lr{2}&\lr{2}&\lr{2}&\lr{2}&\lr{3}&\lr{3}&\lr{3}&\lr{3}&\lr{3}&\lr{4}&\lr{4}&\lr{4}&\lr{4}\\\cline{1-19}
\lr{2}&\lr{2}&\lr{3}\\\cline{1-3}
\lr{4}&\lr{4}\\\cline{1-2}
\end{array}$}
}\end{minipage} \ \begin{minipage}{5cm}
\begin{minipage}{1.7cm}$(18,5,1)$\end{minipage}\scalebox{.55}{\def\lr#1{\multicolumn{1}{|@{\hspace{.6ex}}c@{\hspace{.6ex}}|}{\raisebox{-.3ex}{$#1$}}}
\raisebox{-.6ex}{$\begin{array}[b]{*{18}c}\cline{1-18}
\lr{1}&\lr{1}&\lr{1}&\lr{1}&\lr{1}&\lr{1}&\lr{2}&\lr{2}&\lr{2}&\lr{2}&\lr{2}&\lr{3}&\lr{3}&\lr{3}&\lr{4}&\lr{4}&\lr{4}&\lr{4}\\\cline{1-18}
\lr{2}&\lr{3}&\lr{3}&\lr{3}&\lr{4}\\\cline{1-5}
\lr{4}\\\cline{1-1}
\end{array}$}
}\end{minipage} \ \begin{minipage}{5cm}
\begin{minipage}{1.7cm}$(17,7)$\end{minipage}\scalebox{.55}{\def\lr#1{\multicolumn{1}{|@{\hspace{.6ex}}c@{\hspace{.6ex}}|}{\raisebox{-.3ex}{$#1$}}}
\raisebox{-.6ex}{$\begin{array}[b]{*{17}c}\cline{1-17}
\lr{1}&\lr{1}&\lr{1}&\lr{1}&\lr{1}&\lr{1}&\lr{2}&\lr{2}&\lr{2}&\lr{2}&\lr{2}&\lr{2}&\lr{3}&\lr{3}&\lr{4}&\lr{4}&\lr{4}\\\cline{1-17}
\lr{3}&\lr{3}&\lr{3}&\lr{3}&\lr{4}&\lr{4}&\lr{4}\\\cline{1-7}
\end{array}$}
}\end{minipage} \ \begin{minipage}{5cm}
\begin{minipage}{1.7cm}$(10,10,10)$\end{minipage}\scalebox{.55}{\def\lr#1{\multicolumn{1}{|@{\hspace{.6ex}}c@{\hspace{.6ex}}|}{\raisebox{-.3ex}{$#1$}}}
\raisebox{-.6ex}{$\begin{array}[b]{*{10}c}\cline{1-10}
\lr{1}&\lr{1}&\lr{1}&\lr{1}&\lr{1}&\lr{1}&\lr{2}&\lr{2}&\lr{3}&\lr{3}\\\cline{1-10}
\lr{2}&\lr{2}&\lr{2}&\lr{2}&\lr{3}&\lr{3}&\lr{3}&\lr{3}&\lr{4}&\lr{4}\\\cline{1-10}
\lr{4}&\lr{4}&\lr{4}&\lr{4}&\lr{5}&\lr{5}&\lr{5}&\lr{5}&\lr{5}&\lr{5}\\\cline{1-10}
\end{array}$}
}\end{minipage} \ \begin{minipage}{5cm}
\begin{minipage}{1.7cm}$(11,10,9)$\end{minipage}\scalebox{.55}{\def\lr#1{\multicolumn{1}{|@{\hspace{.6ex}}c@{\hspace{.6ex}}|}{\raisebox{-.3ex}{$#1$}}}
\raisebox{-.6ex}{$\begin{array}[b]{*{11}c}\cline{1-11}
\lr{1}&\lr{1}&\lr{1}&\lr{1}&\lr{1}&\lr{1}&\lr{2}&\lr{2}&\lr{3}&\lr{3}&\lr{4}\\\cline{1-11}
\lr{2}&\lr{2}&\lr{2}&\lr{2}&\lr{3}&\lr{3}&\lr{3}&\lr{3}&\lr{4}&\lr{5}\\\cline{1-10}
\lr{4}&\lr{4}&\lr{4}&\lr{4}&\lr{5}&\lr{5}&\lr{5}&\lr{5}&\lr{5}\\\cline{1-9}
\end{array}$}
}\end{minipage} \ \begin{minipage}{5cm}
\begin{minipage}{1.7cm}$(12,10,8)$\end{minipage}\scalebox{.55}{\def\lr#1{\multicolumn{1}{|@{\hspace{.6ex}}c@{\hspace{.6ex}}|}{\raisebox{-.3ex}{$#1$}}}
\raisebox{-.6ex}{$\begin{array}[b]{*{12}c}\cline{1-12}
\lr{1}&\lr{1}&\lr{1}&\lr{1}&\lr{1}&\lr{1}&\lr{2}&\lr{2}&\lr{2}&\lr{2}&\lr{2}&\lr{2}\\\cline{1-12}
\lr{3}&\lr{3}&\lr{3}&\lr{3}&\lr{3}&\lr{3}&\lr{4}&\lr{4}&\lr{5}&\lr{5}\\\cline{1-10}
\lr{4}&\lr{4}&\lr{4}&\lr{4}&\lr{5}&\lr{5}&\lr{5}&\lr{5}\\\cline{1-8}
\end{array}$}
}\end{minipage} \ \begin{minipage}{5cm}
\begin{minipage}{1.7cm}$(13,9,8)$\end{minipage}\scalebox{.55}{\def\lr#1{\multicolumn{1}{|@{\hspace{.6ex}}c@{\hspace{.6ex}}|}{\raisebox{-.3ex}{$#1$}}}
\raisebox{-.6ex}{$\begin{array}[b]{*{13}c}\cline{1-13}
\lr{1}&\lr{1}&\lr{1}&\lr{1}&\lr{1}&\lr{1}&\lr{2}&\lr{2}&\lr{2}&\lr{2}&\lr{2}&\lr{2}&\lr{5}\\\cline{1-13}
\lr{3}&\lr{3}&\lr{3}&\lr{3}&\lr{3}&\lr{3}&\lr{4}&\lr{4}&\lr{5}\\\cline{1-9}
\lr{4}&\lr{4}&\lr{4}&\lr{4}&\lr{5}&\lr{5}&\lr{5}&\lr{5}\\\cline{1-8}
\end{array}$}
}\end{minipage} \ \begin{minipage}{5cm}
\begin{minipage}{1.7cm}$(12,11,7)$\end{minipage}\scalebox{.55}{\def\lr#1{\multicolumn{1}{|@{\hspace{.6ex}}c@{\hspace{.6ex}}|}{\raisebox{-.3ex}{$#1$}}}
\raisebox{-.6ex}{$\begin{array}[b]{*{12}c}\cline{1-12}
\lr{1}&\lr{1}&\lr{1}&\lr{1}&\lr{1}&\lr{1}&\lr{2}&\lr{2}&\lr{2}&\lr{2}&\lr{2}&\lr{2}\\\cline{1-12}
\lr{3}&\lr{3}&\lr{3}&\lr{3}&\lr{3}&\lr{3}&\lr{4}&\lr{4}&\lr{5}&\lr{5}&\lr{5}\\\cline{1-11}
\lr{4}&\lr{4}&\lr{4}&\lr{4}&\lr{5}&\lr{5}&\lr{5}\\\cline{1-7}
\end{array}$}
}\end{minipage} \ \begin{minipage}{5cm}
\begin{minipage}{1.7cm}$(13,10,7)$\end{minipage}\scalebox{.55}{\def\lr#1{\multicolumn{1}{|@{\hspace{.6ex}}c@{\hspace{.6ex}}|}{\raisebox{-.3ex}{$#1$}}}
\raisebox{-.6ex}{$\begin{array}[b]{*{13}c}\cline{1-13}
\lr{1}&\lr{1}&\lr{1}&\lr{1}&\lr{1}&\lr{1}&\lr{2}&\lr{2}&\lr{2}&\lr{2}&\lr{2}&\lr{2}&\lr{4}\\\cline{1-13}
\lr{3}&\lr{3}&\lr{3}&\lr{3}&\lr{3}&\lr{3}&\lr{4}&\lr{5}&\lr{5}&\lr{5}\\\cline{1-10}
\lr{4}&\lr{4}&\lr{4}&\lr{4}&\lr{5}&\lr{5}&\lr{5}\\\cline{1-7}
\end{array}$}
}\end{minipage} \ \begin{minipage}{5cm}
\begin{minipage}{1.7cm}$(14,9,7)$\end{minipage}\scalebox{.55}{\def\lr#1{\multicolumn{1}{|@{\hspace{.6ex}}c@{\hspace{.6ex}}|}{\raisebox{-.3ex}{$#1$}}}
\raisebox{-.6ex}{$\begin{array}[b]{*{14}c}\cline{1-14}
\lr{1}&\lr{1}&\lr{1}&\lr{1}&\lr{1}&\lr{1}&\lr{2}&\lr{2}&\lr{2}&\lr{2}&\lr{2}&\lr{2}&\lr{4}&\lr{5}\\\cline{1-14}
\lr{3}&\lr{3}&\lr{3}&\lr{3}&\lr{3}&\lr{3}&\lr{4}&\lr{5}&\lr{5}\\\cline{1-9}
\lr{4}&\lr{4}&\lr{4}&\lr{4}&\lr{5}&\lr{5}&\lr{5}\\\cline{1-7}
\end{array}$}
}\end{minipage} \ \begin{minipage}{5cm}
\begin{minipage}{1.7cm}$(13,11,6)$\end{minipage}\scalebox{.55}{\def\lr#1{\multicolumn{1}{|@{\hspace{.6ex}}c@{\hspace{.6ex}}|}{\raisebox{-.3ex}{$#1$}}}
\raisebox{-.6ex}{$\begin{array}[b]{*{13}c}\cline{1-13}
\lr{1}&\lr{1}&\lr{1}&\lr{1}&\lr{1}&\lr{1}&\lr{2}&\lr{2}&\lr{2}&\lr{2}&\lr{2}&\lr{2}&\lr{5}\\\cline{1-13}
\lr{3}&\lr{3}&\lr{3}&\lr{3}&\lr{3}&\lr{3}&\lr{4}&\lr{4}&\lr{5}&\lr{5}&\lr{5}\\\cline{1-11}
\lr{4}&\lr{4}&\lr{4}&\lr{4}&\lr{5}&\lr{5}\\\cline{1-6}
\end{array}$}
}\end{minipage} \ \begin{minipage}{5cm}
\begin{minipage}{1.7cm}$(15,8,7)$\end{minipage}\scalebox{.55}{\def\lr#1{\multicolumn{1}{|@{\hspace{.6ex}}c@{\hspace{.6ex}}|}{\raisebox{-.3ex}{$#1$}}}
\raisebox{-.6ex}{$\begin{array}[b]{*{15}c}\cline{1-15}
\lr{1}&\lr{1}&\lr{1}&\lr{1}&\lr{1}&\lr{1}&\lr{2}&\lr{2}&\lr{4}&\lr{4}&\lr{4}&\lr{4}&\lr{4}&\lr{5}&\lr{5}\\\cline{1-15}
\lr{2}&\lr{2}&\lr{2}&\lr{2}&\lr{3}&\lr{3}&\lr{4}&\lr{5}\\\cline{1-8}
\lr{3}&\lr{3}&\lr{3}&\lr{3}&\lr{5}&\lr{5}&\lr{5}\\\cline{1-7}
\end{array}$}
}\end{minipage} \ \begin{minipage}{5cm}
\begin{minipage}{1.7cm}$(13,12,5)$\end{minipage}\scalebox{.55}{\def\lr#1{\multicolumn{1}{|@{\hspace{.6ex}}c@{\hspace{.6ex}}|}{\raisebox{-.3ex}{$#1$}}}
\raisebox{-.6ex}{$\begin{array}[b]{*{13}c}\cline{1-13}
\lr{1}&\lr{1}&\lr{1}&\lr{1}&\lr{1}&\lr{1}&\lr{2}&\lr{2}&\lr{2}&\lr{2}&\lr{2}&\lr{2}&\lr{5}\\\cline{1-13}
\lr{3}&\lr{3}&\lr{3}&\lr{3}&\lr{3}&\lr{3}&\lr{4}&\lr{4}&\lr{5}&\lr{5}&\lr{5}&\lr{5}\\\cline{1-12}
\lr{4}&\lr{4}&\lr{4}&\lr{4}&\lr{5}\\\cline{1-5}
\end{array}$}
}\end{minipage} \ \begin{minipage}{5cm}
\begin{minipage}{1.7cm}$(16,7,7)$\end{minipage}\scalebox{.55}{\def\lr#1{\multicolumn{1}{|@{\hspace{.6ex}}c@{\hspace{.6ex}}|}{\raisebox{-.3ex}{$#1$}}}
\raisebox{-.6ex}{$\begin{array}[b]{*{16}c}\cline{1-16}
\lr{1}&\lr{1}&\lr{1}&\lr{1}&\lr{1}&\lr{1}&\lr{2}&\lr{2}&\lr{2}&\lr{2}&\lr{2}&\lr{2}&\lr{4}&\lr{5}&\lr{5}&\lr{5}\\\cline{1-16}
\lr{3}&\lr{3}&\lr{3}&\lr{3}&\lr{3}&\lr{3}&\lr{4}\\\cline{1-7}
\lr{4}&\lr{4}&\lr{4}&\lr{4}&\lr{5}&\lr{5}&\lr{5}\\\cline{1-7}
\end{array}$}
}\end{minipage} \ \begin{minipage}{5cm}
\begin{minipage}{1.7cm}$(15,9,6)$\end{minipage}\scalebox{.55}{\def\lr#1{\multicolumn{1}{|@{\hspace{.6ex}}c@{\hspace{.6ex}}|}{\raisebox{-.3ex}{$#1$}}}
\raisebox{-.6ex}{$\begin{array}[b]{*{15}c}\cline{1-15}
\lr{1}&\lr{1}&\lr{1}&\lr{1}&\lr{1}&\lr{1}&\lr{2}&\lr{2}&\lr{2}&\lr{2}&\lr{2}&\lr{2}&\lr{3}&\lr{3}&\lr{4}\\\cline{1-15}
\lr{3}&\lr{3}&\lr{3}&\lr{3}&\lr{4}&\lr{4}&\lr{4}&\lr{4}&\lr{4}\\\cline{1-9}
\lr{5}&\lr{5}&\lr{5}&\lr{5}&\lr{5}&\lr{5}\\\cline{1-6}
\end{array}$}
}\end{minipage} \ \begin{minipage}{5cm}
\begin{minipage}{1.7cm}$(14,11,5)$\end{minipage}\scalebox{.55}{\def\lr#1{\multicolumn{1}{|@{\hspace{.6ex}}c@{\hspace{.6ex}}|}{\raisebox{-.3ex}{$#1$}}}
\raisebox{-.6ex}{$\begin{array}[b]{*{14}c}\cline{1-14}
\lr{1}&\lr{1}&\lr{1}&\lr{1}&\lr{1}&\lr{1}&\lr{2}&\lr{2}&\lr{2}&\lr{2}&\lr{2}&\lr{2}&\lr{3}&\lr{3}\\\cline{1-14}
\lr{3}&\lr{3}&\lr{3}&\lr{3}&\lr{4}&\lr{4}&\lr{5}&\lr{5}&\lr{5}&\lr{5}&\lr{5}\\\cline{1-11}
\lr{4}&\lr{4}&\lr{4}&\lr{4}&\lr{5}\\\cline{1-5}
\end{array}$}
}\end{minipage} \ \begin{minipage}{5cm}
\begin{minipage}{1.7cm}$(13,13,4)$\end{minipage}\scalebox{.55}{\def\lr#1{\multicolumn{1}{|@{\hspace{.6ex}}c@{\hspace{.6ex}}|}{\raisebox{-.3ex}{$#1$}}}
\raisebox{-.6ex}{$\begin{array}[b]{*{13}c}\cline{1-13}
\lr{1}&\lr{1}&\lr{1}&\lr{1}&\lr{1}&\lr{1}&\lr{2}&\lr{2}&\lr{2}&\lr{3}&\lr{3}&\lr{3}&\lr{4}\\\cline{1-13}
\lr{2}&\lr{2}&\lr{2}&\lr{3}&\lr{3}&\lr{3}&\lr{4}&\lr{4}&\lr{4}&\lr{4}&\lr{4}&\lr{5}&\lr{5}\\\cline{1-13}
\lr{5}&\lr{5}&\lr{5}&\lr{5}\\\cline{1-4}
\end{array}$}
}\end{minipage} \ \begin{minipage}{5cm}
\begin{minipage}{1.7cm}$(15,10,5)$\end{minipage}\scalebox{.55}{\def\lr#1{\multicolumn{1}{|@{\hspace{.6ex}}c@{\hspace{.6ex}}|}{\raisebox{-.3ex}{$#1$}}}
\raisebox{-.6ex}{$\begin{array}[b]{*{15}c}\cline{1-15}
\lr{1}&\lr{1}&\lr{1}&\lr{1}&\lr{1}&\lr{1}&\lr{2}&\lr{2}&\lr{2}&\lr{2}&\lr{2}&\lr{2}&\lr{3}&\lr{3}&\lr{4}\\\cline{1-15}
\lr{3}&\lr{3}&\lr{3}&\lr{3}&\lr{4}&\lr{5}&\lr{5}&\lr{5}&\lr{5}&\lr{5}\\\cline{1-10}
\lr{4}&\lr{4}&\lr{4}&\lr{4}&\lr{5}\\\cline{1-5}
\end{array}$}
}\end{minipage} \ \begin{minipage}{5cm}
\begin{minipage}{1.7cm}$(15,11,4)$\end{minipage}\scalebox{.55}{\def\lr#1{\multicolumn{1}{|@{\hspace{.6ex}}c@{\hspace{.6ex}}|}{\raisebox{-.3ex}{$#1$}}}
\raisebox{-.6ex}{$\begin{array}[b]{*{15}c}\cline{1-15}
\lr{1}&\lr{1}&\lr{1}&\lr{1}&\lr{1}&\lr{1}&\lr{2}&\lr{2}&\lr{2}&\lr{2}&\lr{2}&\lr{2}&\lr{3}&\lr{3}&\lr{5}\\\cline{1-15}
\lr{3}&\lr{3}&\lr{3}&\lr{3}&\lr{4}&\lr{4}&\lr{4}&\lr{4}&\lr{4}&\lr{4}&\lr{5}\\\cline{1-11}
\lr{5}&\lr{5}&\lr{5}&\lr{5}\\\cline{1-4}
\end{array}$}
}\end{minipage} \ \begin{minipage}{5cm}
\begin{minipage}{1.7cm}$(14,13,3)$\end{minipage}\scalebox{.55}{\def\lr#1{\multicolumn{1}{|@{\hspace{.6ex}}c@{\hspace{.6ex}}|}{\raisebox{-.3ex}{$#1$}}}
\raisebox{-.6ex}{$\begin{array}[b]{*{14}c}\cline{1-14}
\lr{1}&\lr{1}&\lr{1}&\lr{1}&\lr{1}&\lr{1}&\lr{2}&\lr{2}&\lr{2}&\lr{2}&\lr{2}&\lr{2}&\lr{4}&\lr{4}\\\cline{1-14}
\lr{3}&\lr{3}&\lr{3}&\lr{3}&\lr{3}&\lr{3}&\lr{4}&\lr{4}&\lr{4}&\lr{4}&\lr{5}&\lr{5}&\lr{5}\\\cline{1-13}
\lr{5}&\lr{5}&\lr{5}\\\cline{1-3}
\end{array}$}
}\end{minipage} \ \begin{minipage}{5cm}
\begin{minipage}{1.7cm}$(16,11,3)$\end{minipage}\scalebox{.55}{\def\lr#1{\multicolumn{1}{|@{\hspace{.6ex}}c@{\hspace{.6ex}}|}{\raisebox{-.3ex}{$#1$}}}
\raisebox{-.6ex}{$\begin{array}[b]{*{16}c}\cline{1-16}
\lr{1}&\lr{1}&\lr{1}&\lr{1}&\lr{1}&\lr{1}&\lr{2}&\lr{2}&\lr{2}&\lr{2}&\lr{2}&\lr{2}&\lr{4}&\lr{4}&\lr{5}&\lr{5}\\\cline{1-16}
\lr{3}&\lr{3}&\lr{3}&\lr{3}&\lr{3}&\lr{3}&\lr{4}&\lr{4}&\lr{4}&\lr{4}&\lr{5}\\\cline{1-11}
\lr{5}&\lr{5}&\lr{5}\\\cline{1-3}
\end{array}$}
}\end{minipage} \ \begin{minipage}{5cm}
\begin{minipage}{1.7cm}$(15,13,2)$\end{minipage}\scalebox{.55}{\def\lr#1{\multicolumn{1}{|@{\hspace{.6ex}}c@{\hspace{.6ex}}|}{\raisebox{-.3ex}{$#1$}}}
\raisebox{-.6ex}{$\begin{array}[b]{*{15}c}\cline{1-15}
\lr{1}&\lr{1}&\lr{1}&\lr{1}&\lr{1}&\lr{1}&\lr{2}&\lr{2}&\lr{2}&\lr{2}&\lr{2}&\lr{2}&\lr{4}&\lr{4}&\lr{5}\\\cline{1-15}
\lr{3}&\lr{3}&\lr{3}&\lr{3}&\lr{3}&\lr{3}&\lr{4}&\lr{4}&\lr{4}&\lr{4}&\lr{5}&\lr{5}&\lr{5}\\\cline{1-13}
\lr{5}&\lr{5}\\\cline{1-2}
\end{array}$}
}\end{minipage} \ \begin{minipage}{5cm}
\begin{minipage}{1.7cm}$(15,14,1)$\end{minipage}\scalebox{.55}{\def\lr#1{\multicolumn{1}{|@{\hspace{.6ex}}c@{\hspace{.6ex}}|}{\raisebox{-.3ex}{$#1$}}}
\raisebox{-.6ex}{$\begin{array}[b]{*{15}c}\cline{1-15}
\lr{1}&\lr{1}&\lr{1}&\lr{1}&\lr{1}&\lr{1}&\lr{2}&\lr{2}&\lr{2}&\lr{2}&\lr{2}&\lr{2}&\lr{4}&\lr{4}&\lr{5}\\\cline{1-15}
\lr{3}&\lr{3}&\lr{3}&\lr{3}&\lr{3}&\lr{3}&\lr{4}&\lr{4}&\lr{4}&\lr{4}&\lr{5}&\lr{5}&\lr{5}&\lr{5}\\\cline{1-14}
\lr{5}\\\cline{1-1}
\end{array}$}
}\end{minipage} \ \begin{minipage}{5cm}
\begin{minipage}{1.7cm}$(17,13)$\end{minipage}\scalebox{.55}{\def\lr#1{\multicolumn{1}{|@{\hspace{.6ex}}c@{\hspace{.6ex}}|}{\raisebox{-.3ex}{$#1$}}}
\raisebox{-.6ex}{$\begin{array}[b]{*{17}c}\cline{1-17}
\lr{1}&\lr{1}&\lr{1}&\lr{1}&\lr{1}&\lr{1}&\lr{2}&\lr{2}&\lr{2}&\lr{2}&\lr{2}&\lr{2}&\lr{3}&\lr{3}&\lr{5}&\lr{5}&\lr{5}\\\cline{1-17}
\lr{3}&\lr{3}&\lr{3}&\lr{3}&\lr{4}&\lr{4}&\lr{4}&\lr{4}&\lr{4}&\lr{4}&\lr{5}&\lr{5}&\lr{5}\\\cline{1-13}
\end{array}$}
}\end{minipage} \ \begin{minipage}{5cm}
\begin{minipage}{1.7cm}$(13,12,11)$\end{minipage}\scalebox{.55}{\def\lr#1{\multicolumn{1}{|@{\hspace{.6ex}}c@{\hspace{.6ex}}|}{\raisebox{-.3ex}{$#1$}}}
\raisebox{-.6ex}{$\begin{array}[b]{*{13}c}\cline{1-13}
\lr{1}&\lr{1}&\lr{1}&\lr{1}&\lr{1}&\lr{1}&\lr{2}&\lr{2}&\lr{2}&\lr{2}&\lr{2}&\lr{2}&\lr{5}\\\cline{1-13}
\lr{3}&\lr{3}&\lr{3}&\lr{3}&\lr{3}&\lr{3}&\lr{4}&\lr{4}&\lr{5}&\lr{5}&\lr{5}&\lr{6}\\\cline{1-12}
\lr{4}&\lr{4}&\lr{4}&\lr{4}&\lr{5}&\lr{5}&\lr{6}&\lr{6}&\lr{6}&\lr{6}&\lr{6}\\\cline{1-11}
\end{array}$}
}\end{minipage} \ \begin{minipage}{5cm}
\begin{minipage}{1.7cm}$(14,11,11)$\end{minipage}\scalebox{.55}{\def\lr#1{\multicolumn{1}{|@{\hspace{.6ex}}c@{\hspace{.6ex}}|}{\raisebox{-.3ex}{$#1$}}}
\raisebox{-.6ex}{$\begin{array}[b]{*{14}c}\cline{1-14}
\lr{1}&\lr{1}&\lr{1}&\lr{1}&\lr{1}&\lr{1}&\lr{2}&\lr{2}&\lr{3}&\lr{3}&\lr{4}&\lr{5}&\lr{5}&\lr{6}\\\cline{1-14}
\lr{2}&\lr{2}&\lr{2}&\lr{2}&\lr{3}&\lr{3}&\lr{3}&\lr{3}&\lr{4}&\lr{5}&\lr{5}\\\cline{1-11}
\lr{4}&\lr{4}&\lr{4}&\lr{4}&\lr{5}&\lr{5}&\lr{6}&\lr{6}&\lr{6}&\lr{6}&\lr{6}\\\cline{1-11}
\end{array}$}
}\end{minipage} \ \begin{minipage}{5cm}
\begin{minipage}{1.7cm}$(13,13,10)$\end{minipage}\scalebox{.55}{\def\lr#1{\multicolumn{1}{|@{\hspace{.6ex}}c@{\hspace{.6ex}}|}{\raisebox{-.3ex}{$#1$}}}
\raisebox{-.6ex}{$\begin{array}[b]{*{13}c}\cline{1-13}
\lr{1}&\lr{1}&\lr{1}&\lr{1}&\lr{1}&\lr{1}&\lr{2}&\lr{2}&\lr{3}&\lr{3}&\lr{4}&\lr{5}&\lr{5}\\\cline{1-13}
\lr{2}&\lr{2}&\lr{2}&\lr{2}&\lr{3}&\lr{3}&\lr{3}&\lr{3}&\lr{4}&\lr{5}&\lr{5}&\lr{6}&\lr{6}\\\cline{1-13}
\lr{4}&\lr{4}&\lr{4}&\lr{4}&\lr{5}&\lr{5}&\lr{6}&\lr{6}&\lr{6}&\lr{6}\\\cline{1-10}
\end{array}$}
}\end{minipage} \ \begin{minipage}{5cm}
\begin{minipage}{1.7cm}$(15,11,10)$\end{minipage}\scalebox{.55}{\def\lr#1{\multicolumn{1}{|@{\hspace{.6ex}}c@{\hspace{.6ex}}|}{\raisebox{-.3ex}{$#1$}}}
\raisebox{-.6ex}{$\begin{array}[b]{*{15}c}\cline{1-15}
\lr{1}&\lr{1}&\lr{1}&\lr{1}&\lr{1}&\lr{1}&\lr{2}&\lr{2}&\lr{2}&\lr{2}&\lr{2}&\lr{2}&\lr{3}&\lr{3}&\lr{5}\\\cline{1-15}
\lr{3}&\lr{3}&\lr{3}&\lr{3}&\lr{4}&\lr{4}&\lr{4}&\lr{4}&\lr{4}&\lr{4}&\lr{5}\\\cline{1-11}
\lr{5}&\lr{5}&\lr{5}&\lr{5}&\lr{6}&\lr{6}&\lr{6}&\lr{6}&\lr{6}&\lr{6}\\\cline{1-10}
\end{array}$}
}\end{minipage} \ \begin{minipage}{5cm}
\begin{minipage}{1.7cm}$(14,13,9)$\end{minipage}\scalebox{.55}{\def\lr#1{\multicolumn{1}{|@{\hspace{.6ex}}c@{\hspace{.6ex}}|}{\raisebox{-.3ex}{$#1$}}}
\raisebox{-.6ex}{$\begin{array}[b]{*{14}c}\cline{1-14}
\lr{1}&\lr{1}&\lr{1}&\lr{1}&\lr{1}&\lr{1}&\lr{2}&\lr{2}&\lr{2}&\lr{2}&\lr{2}&\lr{2}&\lr{3}&\lr{3}\\\cline{1-14}
\lr{3}&\lr{3}&\lr{3}&\lr{3}&\lr{4}&\lr{4}&\lr{5}&\lr{5}&\lr{5}&\lr{5}&\lr{5}&\lr{5}&\lr{6}\\\cline{1-13}
\lr{4}&\lr{4}&\lr{4}&\lr{4}&\lr{6}&\lr{6}&\lr{6}&\lr{6}&\lr{6}\\\cline{1-9}
\end{array}$}
}\end{minipage} \ \begin{minipage}{5cm}
\begin{minipage}{1.7cm}$(16,11,9)$\end{minipage}\scalebox{.55}{\def\lr#1{\multicolumn{1}{|@{\hspace{.6ex}}c@{\hspace{.6ex}}|}{\raisebox{-.3ex}{$#1$}}}
\raisebox{-.6ex}{$\begin{array}[b]{*{16}c}\cline{1-16}
\lr{1}&\lr{1}&\lr{1}&\lr{1}&\lr{1}&\lr{1}&\lr{2}&\lr{2}&\lr{2}&\lr{2}&\lr{2}&\lr{2}&\lr{4}&\lr{4}&\lr{4}&\lr{6}\\\cline{1-16}
\lr{3}&\lr{3}&\lr{3}&\lr{3}&\lr{3}&\lr{3}&\lr{4}&\lr{4}&\lr{4}&\lr{5}&\lr{5}\\\cline{1-11}
\lr{5}&\lr{5}&\lr{5}&\lr{5}&\lr{6}&\lr{6}&\lr{6}&\lr{6}&\lr{6}\\\cline{1-9}
\end{array}$}
}\end{minipage} \ \begin{minipage}{5cm}
\begin{minipage}{1.7cm}$(15,13,8)$\end{minipage}\scalebox{.55}{\def\lr#1{\multicolumn{1}{|@{\hspace{.6ex}}c@{\hspace{.6ex}}|}{\raisebox{-.3ex}{$#1$}}}
\raisebox{-.6ex}{$\begin{array}[b]{*{15}c}\cline{1-15}
\lr{1}&\lr{1}&\lr{1}&\lr{1}&\lr{1}&\lr{1}&\lr{2}&\lr{2}&\lr{2}&\lr{2}&\lr{2}&\lr{2}&\lr{3}&\lr{3}&\lr{6}\\\cline{1-15}
\lr{3}&\lr{3}&\lr{3}&\lr{3}&\lr{4}&\lr{4}&\lr{5}&\lr{5}&\lr{5}&\lr{5}&\lr{5}&\lr{5}&\lr{6}\\\cline{1-13}
\lr{4}&\lr{4}&\lr{4}&\lr{4}&\lr{6}&\lr{6}&\lr{6}&\lr{6}\\\cline{1-8}
\end{array}$}
}\end{minipage} \ \begin{minipage}{5cm}
\begin{minipage}{1.7cm}$(15,14,7)$\end{minipage}\scalebox{.55}{\def\lr#1{\multicolumn{1}{|@{\hspace{.6ex}}c@{\hspace{.6ex}}|}{\raisebox{-.3ex}{$#1$}}}
\raisebox{-.6ex}{$\begin{array}[b]{*{15}c}\cline{1-15}
\lr{1}&\lr{1}&\lr{1}&\lr{1}&\lr{1}&\lr{1}&\lr{2}&\lr{2}&\lr{2}&\lr{2}&\lr{2}&\lr{2}&\lr{3}&\lr{3}&\lr{6}\\\cline{1-15}
\lr{3}&\lr{3}&\lr{3}&\lr{3}&\lr{4}&\lr{4}&\lr{5}&\lr{5}&\lr{5}&\lr{5}&\lr{5}&\lr{5}&\lr{6}&\lr{6}\\\cline{1-14}
\lr{4}&\lr{4}&\lr{4}&\lr{4}&\lr{6}&\lr{6}&\lr{6}\\\cline{1-7}
\end{array}$}
}\end{minipage} \ \begin{minipage}{5cm}
\begin{minipage}{1.7cm}$(18,9,9)$\end{minipage}\scalebox{.55}{\def\lr#1{\multicolumn{1}{|@{\hspace{.6ex}}c@{\hspace{.6ex}}|}{\raisebox{-.3ex}{$#1$}}}
\raisebox{-.6ex}{$\begin{array}[b]{*{18}c}\cline{1-18}
\lr{1}&\lr{1}&\lr{1}&\lr{1}&\lr{1}&\lr{1}&\lr{2}&\lr{2}&\lr{3}&\lr{3}&\lr{3}&\lr{3}&\lr{3}&\lr{3}&\lr{5}&\lr{5}&\lr{5}&\lr{6}\\\cline{1-18}
\lr{2}&\lr{2}&\lr{2}&\lr{2}&\lr{4}&\lr{4}&\lr{5}&\lr{5}&\lr{5}\\\cline{1-9}
\lr{4}&\lr{4}&\lr{4}&\lr{4}&\lr{6}&\lr{6}&\lr{6}&\lr{6}&\lr{6}\\\cline{1-9}
\end{array}$}
}\end{minipage} \ \begin{minipage}{5cm}
\begin{minipage}{1.7cm}$(15,15,6)$\end{minipage}\scalebox{.55}{\def\lr#1{\multicolumn{1}{|@{\hspace{.6ex}}c@{\hspace{.6ex}}|}{\raisebox{-.3ex}{$#1$}}}
\raisebox{-.6ex}{$\begin{array}[b]{*{15}c}\cline{1-15}
\lr{1}&\lr{1}&\lr{1}&\lr{1}&\lr{1}&\lr{1}&\lr{2}&\lr{2}&\lr{2}&\lr{2}&\lr{2}&\lr{2}&\lr{3}&\lr{3}&\lr{5}\\\cline{1-15}
\lr{3}&\lr{3}&\lr{3}&\lr{3}&\lr{4}&\lr{4}&\lr{4}&\lr{4}&\lr{4}&\lr{4}&\lr{5}&\lr{6}&\lr{6}&\lr{6}&\lr{6}\\\cline{1-15}
\lr{5}&\lr{5}&\lr{5}&\lr{5}&\lr{6}&\lr{6}\\\cline{1-6}
\end{array}$}
}\end{minipage} \ \begin{minipage}{5cm}
\begin{minipage}{1.7cm}$(17,17,2)$\end{minipage}\scalebox{.55}{\def\lr#1{\multicolumn{1}{|@{\hspace{.6ex}}c@{\hspace{.6ex}}|}{\raisebox{-.3ex}{$#1$}}}
\raisebox{-.6ex}{$\begin{array}[b]{*{17}c}\cline{1-17}
\lr{1}&\lr{1}&\lr{1}&\lr{1}&\lr{1}&\lr{1}&\lr{2}&\lr{2}&\lr{2}&\lr{2}&\lr{3}&\lr{3}&\lr{4}&\lr{4}&\lr{5}&\lr{5}&\lr{5}\\\cline{1-17}
\lr{2}&\lr{2}&\lr{3}&\lr{3}&\lr{3}&\lr{3}&\lr{4}&\lr{4}&\lr{4}&\lr{4}&\lr{5}&\lr{5}&\lr{5}&\lr{6}&\lr{6}&\lr{6}&\lr{6}\\\cline{1-17}
\lr{6}&\lr{6}\\\cline{1-2}
\end{array}$}
}\end{minipage} \ \begin{minipage}{5cm}
\begin{minipage}{1.7cm}$(18,17,1)$\end{minipage}\scalebox{.55}{\def\lr#1{\multicolumn{1}{|@{\hspace{.6ex}}c@{\hspace{.6ex}}|}{\raisebox{-.3ex}{$#1$}}}
\raisebox{-.6ex}{$\begin{array}[b]{*{18}c}\cline{1-18}
\lr{1}&\lr{1}&\lr{1}&\lr{1}&\lr{1}&\lr{1}&\lr{2}&\lr{2}&\lr{2}&\lr{2}&\lr{2}&\lr{2}&\lr{4}&\lr{4}&\lr{5}&\lr{5}&\lr{5}&\lr{5}\\\cline{1-18}
\lr{3}&\lr{3}&\lr{3}&\lr{3}&\lr{3}&\lr{3}&\lr{4}&\lr{4}&\lr{4}&\lr{4}&\lr{5}&\lr{5}&\lr{6}&\lr{6}&\lr{6}&\lr{6}&\lr{6}\\\cline{1-17}
\lr{6}\\\cline{1-1}
\end{array}$}
}\end{minipage} \ \begin{minipage}{5cm}
\begin{minipage}{1.7cm}$(26,5,5)$\end{minipage}\scalebox{.55}{\def\lr#1{\multicolumn{1}{|@{\hspace{.6ex}}c@{\hspace{.6ex}}|}{\raisebox{-.3ex}{$#1$}}}
\raisebox{-.6ex}{$\begin{array}[b]{*{26}c}\cline{1-26}
\lr{1}&\lr{1}&\lr{1}&\lr{1}&\lr{1}&\lr{1}&\lr{2}&\lr{2}&\lr{2}&\lr{2}&\lr{3}&\lr{3}&\lr{3}&\lr{3}&\lr{4}&\lr{4}&\lr{4}&\lr{4}&\lr{4}&\lr{5}&\lr{5}&\lr{5}&\lr{5}&\lr{6}&\lr{6}&\lr{6}\\\cline{1-26}
\lr{2}&\lr{2}&\lr{3}&\lr{3}&\lr{4}\\\cline{1-5}
\lr{5}&\lr{5}&\lr{6}&\lr{6}&\lr{6}\\\cline{1-5}
\end{array}$}
}\end{minipage} \ \begin{minipage}{5cm}
\begin{minipage}{1.7cm}$(15,14,13)$\end{minipage}\scalebox{.55}{\def\lr#1{\multicolumn{1}{|@{\hspace{.6ex}}c@{\hspace{.6ex}}|}{\raisebox{-.3ex}{$#1$}}}
\raisebox{-.6ex}{$\begin{array}[b]{*{15}c}\cline{1-15}
\lr{1}&\lr{1}&\lr{1}&\lr{1}&\lr{1}&\lr{1}&\lr{2}&\lr{2}&\lr{2}&\lr{2}&\lr{2}&\lr{2}&\lr{3}&\lr{3}&\lr{7}\\\cline{1-15}
\lr{3}&\lr{3}&\lr{3}&\lr{3}&\lr{4}&\lr{4}&\lr{5}&\lr{5}&\lr{5}&\lr{5}&\lr{5}&\lr{5}&\lr{6}&\lr{7}\\\cline{1-14}
\lr{4}&\lr{4}&\lr{4}&\lr{4}&\lr{6}&\lr{6}&\lr{6}&\lr{6}&\lr{6}&\lr{7}&\lr{7}&\lr{7}&\lr{7}\\\cline{1-13}
\end{array}$}
}\end{minipage} \ \begin{minipage}{5cm}
\begin{minipage}{1.7cm}$(16,13,13)$\end{minipage}\scalebox{.55}{\def\lr#1{\multicolumn{1}{|@{\hspace{.6ex}}c@{\hspace{.6ex}}|}{\raisebox{-.3ex}{$#1$}}}
\raisebox{-.6ex}{$\begin{array}[b]{*{16}c}\cline{1-16}
\lr{1}&\lr{1}&\lr{1}&\lr{1}&\lr{1}&\lr{1}&\lr{2}&\lr{2}&\lr{2}&\lr{2}&\lr{2}&\lr{2}&\lr{3}&\lr{3}&\lr{7}&\lr{7}\\\cline{1-16}
\lr{3}&\lr{3}&\lr{3}&\lr{3}&\lr{4}&\lr{4}&\lr{5}&\lr{5}&\lr{5}&\lr{5}&\lr{5}&\lr{5}&\lr{6}\\\cline{1-13}
\lr{4}&\lr{4}&\lr{4}&\lr{4}&\lr{6}&\lr{6}&\lr{6}&\lr{6}&\lr{6}&\lr{7}&\lr{7}&\lr{7}&\lr{7}\\\cline{1-13}
\end{array}$}
}\end{minipage} \ \begin{minipage}{5cm}
\begin{minipage}{1.7cm}$(15,15,12)$\end{minipage}\scalebox{.55}{\def\lr#1{\multicolumn{1}{|@{\hspace{.6ex}}c@{\hspace{.6ex}}|}{\raisebox{-.3ex}{$#1$}}}
\raisebox{-.6ex}{$\begin{array}[b]{*{15}c}\cline{1-15}
\lr{1}&\lr{1}&\lr{1}&\lr{1}&\lr{1}&\lr{1}&\lr{2}&\lr{2}&\lr{2}&\lr{2}&\lr{2}&\lr{2}&\lr{3}&\lr{3}&\lr{5}\\\cline{1-15}
\lr{3}&\lr{3}&\lr{3}&\lr{3}&\lr{4}&\lr{4}&\lr{4}&\lr{4}&\lr{4}&\lr{4}&\lr{5}&\lr{6}&\lr{6}&\lr{6}&\lr{6}\\\cline{1-15}
\lr{5}&\lr{5}&\lr{5}&\lr{5}&\lr{6}&\lr{6}&\lr{7}&\lr{7}&\lr{7}&\lr{7}&\lr{7}&\lr{7}\\\cline{1-12}
\end{array}$}
}\end{minipage} \ \begin{minipage}{5cm}
\begin{minipage}{1.7cm}$(17,17,8)$\end{minipage}\scalebox{.55}{\def\lr#1{\multicolumn{1}{|@{\hspace{.6ex}}c@{\hspace{.6ex}}|}{\raisebox{-.3ex}{$#1$}}}
\raisebox{-.6ex}{$\begin{array}[b]{*{17}c}\cline{1-17}
\lr{1}&\lr{1}&\lr{1}&\lr{1}&\lr{1}&\lr{1}&\lr{2}&\lr{2}&\lr{2}&\lr{2}&\lr{2}&\lr{2}&\lr{4}&\lr{4}&\lr{4}&\lr{4}&\lr{6}\\\cline{1-17}
\lr{3}&\lr{3}&\lr{3}&\lr{3}&\lr{3}&\lr{3}&\lr{4}&\lr{4}&\lr{5}&\lr{5}&\lr{5}&\lr{5}&\lr{5}&\lr{5}&\lr{6}&\lr{7}&\lr{7}\\\cline{1-17}
\lr{6}&\lr{6}&\lr{6}&\lr{6}&\lr{7}&\lr{7}&\lr{7}&\lr{7}\\\cline{1-8}
\end{array}$}
}\end{minipage} \ \begin{minipage}{5cm}
\begin{minipage}{1.7cm}$(18,15,15)$\end{minipage}\scalebox{.55}{\def\lr#1{\multicolumn{1}{|@{\hspace{.6ex}}c@{\hspace{.6ex}}|}{\raisebox{-.3ex}{$#1$}}}
\raisebox{-.6ex}{$\begin{array}[b]{*{18}c}\cline{1-18}
\lr{1}&\lr{1}&\lr{1}&\lr{1}&\lr{1}&\lr{1}&\lr{2}&\lr{2}&\lr{3}&\lr{4}&\lr{4}&\lr{4}&\lr{4}&\lr{5}&\lr{6}&\lr{6}&\lr{7}&\lr{8}\\\cline{1-18}
\lr{2}&\lr{2}&\lr{2}&\lr{2}&\lr{3}&\lr{4}&\lr{5}&\lr{5}&\lr{5}&\lr{5}&\lr{6}&\lr{7}&\lr{7}&\lr{7}&\lr{7}\\\cline{1-15}
\lr{3}&\lr{3}&\lr{3}&\lr{3}&\lr{4}&\lr{5}&\lr{6}&\lr{6}&\lr{6}&\lr{7}&\lr{8}&\lr{8}&\lr{8}&\lr{8}&\lr{8}\\\cline{1-15}
\end{array}$}
}\end{minipage} \ \begin{minipage}{5cm}
\begin{minipage}{1.7cm}$(17,17,14)$\end{minipage}\scalebox{.55}{\def\lr#1{\multicolumn{1}{|@{\hspace{.6ex}}c@{\hspace{.6ex}}|}{\raisebox{-.3ex}{$#1$}}}
\raisebox{-.6ex}{$\begin{array}[b]{*{17}c}\cline{1-17}
\lr{1}&\lr{1}&\lr{1}&\lr{1}&\lr{1}&\lr{1}&\lr{2}&\lr{2}&\lr{2}&\lr{2}&\lr{2}&\lr{2}&\lr{3}&\lr{3}&\lr{6}&\lr{6}&\lr{7}\\\cline{1-17}
\lr{3}&\lr{3}&\lr{3}&\lr{3}&\lr{4}&\lr{4}&\lr{5}&\lr{5}&\lr{5}&\lr{5}&\lr{5}&\lr{5}&\lr{6}&\lr{6}&\lr{7}&\lr{8}&\lr{8}\\\cline{1-17}
\lr{4}&\lr{4}&\lr{4}&\lr{4}&\lr{6}&\lr{6}&\lr{7}&\lr{7}&\lr{7}&\lr{7}&\lr{8}&\lr{8}&\lr{8}&\lr{8}\\\cline{1-14}
\end{array}$}
}\end{minipage} \ \begin{minipage}{5cm}
\begin{minipage}{1.7cm}$(25,23)$\end{minipage}\scalebox{.55}{\def\lr#1{\multicolumn{1}{|@{\hspace{.6ex}}c@{\hspace{.6ex}}|}{\raisebox{-.3ex}{$#1$}}}
\raisebox{-.6ex}{$\begin{array}[b]{*{25}c}\cline{1-25}
\lr{1}&\lr{1}&\lr{1}&\lr{1}&\lr{1}&\lr{1}&\lr{2}&\lr{2}&\lr{2}&\lr{2}&\lr{2}&\lr{3}&\lr{3}&\lr{3}&\lr{4}&\lr{5}&\lr{6}&\lr{6}&\lr{6}&\lr{6}&\lr{6}&\lr{7}&\lr{7}&\lr{7}&\lr{8}\\\cline{1-25}
\lr{2}&\lr{3}&\lr{3}&\lr{3}&\lr{4}&\lr{4}&\lr{4}&\lr{4}&\lr{4}&\lr{5}&\lr{5}&\lr{5}&\lr{5}&\lr{5}&\lr{6}&\lr{7}&\lr{7}&\lr{7}&\lr{8}&\lr{8}&\lr{8}&\lr{8}&\lr{8}\\\cline{1-23}
\end{array}$}
}\end{minipage}

\noindent\begin{minipage}{10cm}
\begin{minipage}{1.7cm}$(45,45)$\end{minipage}\scalebox{.55}{\def\lr#1{\multicolumn{1}{|@{\hspace{.6ex}}c@{\hspace{.6ex}}|}{\raisebox{-.3ex}{$#1$}}}
\raisebox{-.6ex}{$\begin{array}[b]{*{45}c}\cline{1-45}
\lr{1}&\lr{1}&\lr{1}&\lr{1}&\lr{1}&\lr{1}&\lr{2}&\lr{2}&\lr{3}&\lr{3}&\lr{3}&\lr{4}&\lr{4}&\lr{4}&\lr{4}&\lr{5}&\lr{5}&\lr{6}&\lr{6}&\lr{6}&\lr{6}&\lr{6}&\lr{7}&\lr{7}&\lr{7}&\lr{8}&\lr{8}&\lr{8}&\lr{9}&\lr{9}&\lr{10}&\lr{10}&\lr{11}&\lr{11}&\lr{11}&\lr{12}&\lr{12}&\lr{12}&\lr{13}&\lr{13}&\lr{13}&\lr{13}&\lr{14}&\lr{14}&\lr{14}\\\cline{1-45}
\lr{2}&\lr{2}&\lr{2}&\lr{2}&\lr{3}&\lr{3}&\lr{3}&\lr{4}&\lr{4}&\lr{5}&\lr{5}&\lr{5}&\lr{5}&\lr{6}&\lr{7}&\lr{7}&\lr{7}&\lr{8}&\lr{8}&\lr{8}&\lr{9}&\lr{9}&\lr{9}&\lr{9}&\lr{10}&\lr{10}&\lr{10}&\lr{10}&\lr{11}&\lr{11}&\lr{11}&\lr{12}&\lr{12}&\lr{12}&\lr{13}&\lr{13}&\lr{14}&\lr{14}&\lr{14}&\lr{15}&\lr{15}&\lr{15}&\lr{15}&\lr{15}&\lr{15}\\\cline{1-45}
\end{array}$}
}\end{minipage}

\medskip

We now give the tableaux for the case $n=7$, $m=4$.
We list them in short notation, using row-wise encoding and a dot to separate rows. Commas separate the tableaux.
Furthermore, we write $A$ instead of 10, $B$ instead of 11, etc.
Repeated letters are written in short form, e.g., instead of 111 we write $1^3$.

{\begin{flushleft}
\noindent\scalebox{0.75}{$1^{7}$},
\scalebox{0.75}{$1^{7}2.2^{6}$},
\scalebox{0.75}{$1^{7}2^{3}.2^{4}$},
\scalebox{0.75}{$1^{7}2^{5}.2^{2}$},
\scalebox{0.75}{$1^{7}2^{6}3.23^{6}$},
\scalebox{0.75}{$1^{7}23.2^{6}.3^{6}$},
\scalebox{0.75}{$1^{7}23^{6}.2^{6}.3$},
\scalebox{0.75}{$1^{7}2^{5}3^{6}.2^{2}3$},
\scalebox{0.75}{$1^{7}23^{5}.2^{6}3^{2}$},
\scalebox{0.75}{$1^{7}2^{6}3.23^{6}4^{7}$},
\scalebox{0.75}{$1^{7}2^{3}.2^{4}3^{6}.3$},
\scalebox{0.75}{$1^{7}2.2^{6}3^{2}.3^{5}$},
\scalebox{0.75}{$1^{7}23^{4}.2^{6}3.3^{2}$},
\scalebox{0.75}{$1^{7}2^{5}3^{4}.2^{2}3^{3}$},
\scalebox{0.75}{$1^{7}2^{3}3^{2}.2^{4}3^{5}$},
\scalebox{0.75}{$1^{7}234.2^{6}.3^{6}.4^{6}$},
\scalebox{0.75}{$1^{7}2^{3}.2^{4}3^{4}.3^{3}$},
\scalebox{0.75}{$1^{7}2^{3}3^{3}.2^{4}.3^{4}$},
\scalebox{0.75}{$1^{7}23^{6}4.2^{6}4^{5}.3.4$},
\scalebox{0.75}{$1^{7}23^{5}.2^{6}34^{6}.3.4$},
\scalebox{0.75}{$1^{7}2^{3}3^{5}.2^{4}.3^{2}$},
\scalebox{0.75}{$1^{7}2^{5}3^{5}.2^{2}.3^{2}$},
\scalebox{0.75}{$1^{7}2^{4}3^{2}.2^{3}3^{4}.3$},
\scalebox{0.75}{$1^{7}2^{3}3.2^{4}3^{2}.3^{4}$},
\scalebox{0.75}{$1^{7}2^{5}3^{4}.2^{2}3^{2}.3$},
\scalebox{0.75}{$1^{7}2^{3}3^{4}.2^{4}3.3^{2}$},
\scalebox{0.75}{$1^{7}2^{3}3^{2}.2^{4}3^{4}.3$},
\scalebox{0.75}{$1^{7}2^{6}34.23^{5}.34^{5}.4$},
\scalebox{0.75}{$1^{7}2^{3}3.2^{4}3^{4}.3^{2}$},
\scalebox{0.75}{$1^{7}2^{2}3.2^{5}3^{2}.3^{4}$},
\scalebox{0.75}{$1^{7}2^{7}3.3^{6}4^{3}.4^{4}$},
\scalebox{0.75}{$1^{7}2^{2}3^{3}.2^{5}3.3^{3}$},
\scalebox{0.75}{$1^{7}2^{7}34.3^{6}4^{2}.4^{4}$},
\scalebox{0.75}{$1^{7}23^{3}4^{6}.2^{6}.3^{4}4$},
\scalebox{0.75}{$1^{7}2^{3}3^{6}4^{6}.2^{4}3.4$},
\scalebox{0.75}{$1^{7}2^{6}34^{4}.23^{5}4^{3}.3$},
\scalebox{0.75}{$1^{7}23^{2}4^{6}.2^{6}3.3^{4}4$},
\scalebox{0.75}{$1^{7}23.2^{6}4^{2}.3^{6}.4^{5}$},
\scalebox{0.75}{$1^{7}23^{5}4.2^{6}3^{2}4^{5}.4$},
\scalebox{0.75}{$1^{7}2^{5}3^{6}4^{5}.2^{2}4.34$},
\scalebox{0.75}{$1^{7}234^{3}.2^{6}4^{3}.3^{6}.4$},
\scalebox{0.75}{$1^{7}234^{4}.2^{6}4^{2}.3^{6}.4$},
\scalebox{0.75}{$1^{7}23^{4}4.2^{6}34^{5}.3^{2}4$},
\scalebox{0.75}{$1^{7}23^{6}4^{3}.2^{6}.34^{3}.4$},
\scalebox{0.75}{$1^{7}2^{3}3^{6}4^{5}.2^{4}.34.4$},
\scalebox{0.75}{$1^{7}23^{6}4^{2}.2^{6}4^{4}.3.4$},
\scalebox{0.75}{$1^{7}2^{7}3^{2}.3^{5}4^{5}.4^{2}$},
\scalebox{0.75}{$1^{7}23^{6}45.2^{6}3.4^{6}.5^{6}$},
\scalebox{0.75}{$1^{7}2^{5}34^{4}.2^{2}3^{6}4^{3}$},
\scalebox{0.75}{$1^{7}2^{6}3.23^{5}4^{2}.34^{4}.4$},
\scalebox{0.75}{$1^{7}2^{3}3^{3}.2^{4}3^{2}.3^{2}$},
\scalebox{0.75}{$1^{7}23^{4}4^{5}.2^{6}3^{2}4.3.4$},
\scalebox{0.75}{$1^{7}2^{6}3^{4}.23^{2}4^{5}.34.4$},
\scalebox{0.75}{$1^{7}2^{6}3^{3}.23^{3}4^{5}.34.4$},
\scalebox{0.75}{$1^{7}2^{3}3^{6}4^{3}.2^{4}3.4^{4}$},
\scalebox{0.75}{$1^{7}2^{2}34^{4}.2^{5}34.3^{5}4.4$},
\scalebox{0.75}{$1^{7}23^{5}.2^{6}3^{2}4^{3}.4^{4}$},
\scalebox{0.75}{$1^{7}234^{3}.2^{6}34.3^{5}4^{2}.4$},
\scalebox{0.75}{$1^{7}2^{7}34^{2}.3^{6}4^{3}.4^{2}$},
\scalebox{0.75}{$1^{7}2^{3}34^{2}.2^{4}3^{6}.4^{5}$},
\scalebox{0.75}{$1^{7}2^{5}3^{7}.2^{2}4^{6}5^{7}.4$},
\scalebox{0.75}{$1^{7}2^{2}34.2^{5}34.3^{5}4.4^{4}$},
\scalebox{0.75}{$1^{7}24^{4}.2^{6}3^{2}.3^{5}4^{3}$},
\scalebox{0.75}{$1^{7}234.2^{6}34^{2}.3^{5}4.4^{3}$},
\scalebox{0.75}{$1^{7}23^{7}4^{7}5.2^{6}5^{6}6^{7}$},
\scalebox{0.75}{$1^{7}23^{4}4^{4}.2^{6}3^{2}4^{3}.3$},
\scalebox{0.75}{$1^{7}23^{3}4^{4}.2^{6}4.3^{4}4^{2}$},
\scalebox{0.75}{$1^{7}2^{4}3.2^{3}3^{3}4^{6}.3^{3}4$},
\scalebox{0.75}{$1^{7}2^{7}34^{2}5.3^{6}4^{5}.5^{6}$},
\scalebox{0.75}{$1^{7}23^{3}4.2^{6}4^{4}.3^{4}4^{2}$},
\scalebox{0.75}{$1^{7}2^{2}3^{4}4.2^{5}34^{6}.3^{2}$},
\scalebox{0.75}{$1^{7}2.2^{6}3^{2}.3^{5}4^{3}.4^{4}$},
\scalebox{0.75}{$1^{7}23^{3}4^{4}.2^{6}34^{3}.3^{3}$},
\scalebox{0.75}{$1^{7}234^{5}.2^{6}35^{6}.3^{5}45.4$},
\scalebox{0.75}{$1^{7}23^{4}.2^{6}4^{5}.3^{3}.4^{2}$},
\scalebox{0.75}{$1^{7}2^{7}3.3^{6}4^{3}5^{6}.4^{4}5$},
\scalebox{0.75}{$1^{7}23^{2}4^{4}.2^{6}4^{2}.3^{5}4$},
\scalebox{0.75}{$1^{7}23^{5}4^{2}.2^{6}3^{2}4^{4}.4$},
\scalebox{0.75}{$1^{7}23^{2}.2^{6}4^{2}.3^{5}4.4^{4}$},
\scalebox{0.75}{$1^{7}24^{3}.2^{6}3^{2}.3^{5}4^{3}.4$},
\scalebox{0.75}{$1^{7}23^{3}4^{2}.2^{6}4^{4}.3^{4}.4$},
\scalebox{0.75}{$1^{7}23^{3}4^{3}.2^{6}4.3^{4}.4^{3}$},
\scalebox{0.75}{$1^{7}23^{2}4^{4}.2^{6}34^{2}.3^{4}4$},
\scalebox{0.75}{$1^{7}23^{6}4^{2}.2^{6}345^{7}.4^{4}$},
\scalebox{0.75}{$1^{7}23^{5}4^{4}.2^{6}.3^{2}4.4^{2}$},
\scalebox{0.75}{$1^{7}23^{3}4^{3}.2^{6}34^{3}.3^{3}4$},
\scalebox{0.75}{$1^{7}2^{7}35.3^{6}4^{2}5^{5}.4^{5}5$},
\scalebox{0.75}{$1^{7}2^{3}34.2^{4}3^{2}4^{5}.3^{4}4$},
\scalebox{0.75}{$1^{7}23^{3}4^{4}.2^{6}.3^{4}4.4^{2}$},
\scalebox{0.75}{$1^{7}2^{7}3.3^{6}4^{3}5^{6}.4^{4}.5$},
\scalebox{0.75}{$1^{7}2^{7}3.3^{6}4^{2}5^{6}.4^{5}.5$},
\scalebox{0.75}{$1^{7}23^{2}.2^{6}34^{3}.3^{4}.4^{4}$},
\scalebox{0.75}{$1^{7}24^{6}.2^{6}3^{2}5^{6}.3^{5}45$},
\scalebox{0.75}{$1^{7}23^{4}4^{4}.2^{6}.3^{3}4.4^{2}$},
\scalebox{0.75}{$1^{7}2^{3}3^{2}.2^{4}34^{6}.3^{4}.4$},
\scalebox{0.75}{$1^{7}2^{3}3^{6}4^{3}.2^{4}.34^{3}.4$},
\scalebox{0.75}{$1^{7}23^{3}4.2^{6}4^{3}.3^{4}.4^{3}$},
\scalebox{0.75}{$1^{7}2^{2}3^{6}4^{4}.2^{5}4^{2}.3.4$},
\scalebox{0.75}{$1^{7}2^{5}3.2^{2}3^{4}4^{6}.3^{2}.4$},
\scalebox{0.75}{$1^{7}23^{4}4^{4}.2^{6}34.3^{2}.4^{2}$},
\scalebox{0.75}{$1^{7}23^{3}4^{3}.2^{6}4.3^{4}4.4^{2}$},
\scalebox{0.75}{$1^{7}2^{2}3^{3}4^{5}.2^{5}34.3^{3}.4$},
\scalebox{0.75}{$1^{7}234^{6}.2^{6}3^{2}5^{6}.3^{4}45$},
\scalebox{0.75}{$1^{7}2^{2}3^{6}4^{3}.2^{5}4.34^{2}.4$},
\scalebox{0.75}{$1^{7}23^{3}4.2^{6}4^{5}5.3^{4}45^{6}$},
\scalebox{0.75}{$1^{7}2^{2}34.2^{5}34^{2}.3^{5}.4^{4}$},
\scalebox{0.75}{$1^{7}2^{3}4.2^{4}3^{4}4.3^{3}4^{4}.4$},
\scalebox{0.75}{$1^{7}23^{6}4^{4}.2^{6}4^{3}5^{6}.3.5$},
\scalebox{0.75}{$1^{7}2^{2}3^{4}4^{5}.2^{5}34.3^{2}.4$},
\scalebox{0.75}{$1^{7}23^{4}4^{4}.2^{6}3^{2}.34^{2}.4$},
\scalebox{0.75}{$1^{7}2^{2}3^{2}.2^{5}34^{5}.3^{4}4.4$},
\scalebox{0.75}{$1^{7}234^{4}.2^{6}4^{3}5.3^{6}.5^{6}$},
\scalebox{0.75}{$1^{7}23^{2}4^{2}.2^{6}34.3^{4}.4^{4}$},
\scalebox{0.75}{$1^{7}2^{5}34^{2}.2^{2}3^{5}.34^{4}.4$},
\scalebox{0.75}{$1^{7}23^{3}4^{4}.2^{6}3.3^{3}4.4^{2}$},
\scalebox{0.75}{$1^{7}234^{4}5.2^{6}4^{3}.3^{6}.5^{6}$},
\scalebox{0.75}{$1^{7}24^{6}.2^{6}3^{2}5^{6}.3^{5}4.5$},
\scalebox{0.75}{$1^{7}2^{6}4.23^{7}4^{5}5.45^{6}6^{7}$},
\scalebox{0.75}{$1^{7}2^{3}3^{7}4^{2}.2^{4}4^{5}5^{7}$},
\scalebox{0.75}{$1^{7}2^{2}34^{3}.2^{5}3.3^{5}4.4^{3}$},
\scalebox{0.75}{$1^{7}2^{7}345.3^{6}4^{2}5^{5}.4^{4}.5$},
\scalebox{0.75}{$1^{7}2^{7}3.3^{6}4^{7}56.5^{6}6.6^{5}$},
\scalebox{0.75}{$1^{7}23^{4}4.2^{6}3^{2}4^{3}.34^{2}.4$},
\scalebox{0.75}{$1^{7}234^{4}.2^{6}45^{6}.3^{6}5.4^{2}$},
\scalebox{0.75}{$1^{7}234^{6}5^{2}.2^{6}45^{4}.3^{6}.5$},
\scalebox{0.75}{$1^{7}2^{3}3^{2}4^{6}.2^{4}3^{5}45^{7}$},
\scalebox{0.75}{$1^{7}2^{5}34^{6}5^{6}.2^{2}3^{5}4.3.5$},
\scalebox{0.75}{$1^{7}23^{3}4.2^{6}34.3^{3}4^{2}.4^{3}$},
\scalebox{0.75}{$1^{7}23^{5}4^{5}.2^{6}45^{6}.3^{2}.45$},
\scalebox{0.75}{$1^{7}23^{2}4.2^{6}34^{2}.3^{4}4.4^{3}$},
\scalebox{0.75}{$1^{7}2^{3}34.2^{4}3^{6}5.4^{6}5.5^{5}$},
\scalebox{0.75}{$1^{7}23^{6}4^{2}5.2^{6}5^{6}.34^{4}.4$},
\scalebox{0.75}{$1^{7}23^{4}4.2^{6}34^{2}.3^{2}4^{3}.4$},
\scalebox{0.75}{$1^{7}234^{4}5^{5}.2^{6}345.3^{5}4^{2}5$},
\scalebox{0.75}{$1^{7}234^{6}6.2^{6}35^{7}6.3^{5}46^{5}$},
\scalebox{0.75}{$1^{7}23^{2}.2^{6}3^{2}4^{2}.3^{3}4^{5}$},
\scalebox{0.75}{$1^{7}2^{5}3^{5}4^{5}.2^{2}.3^{2}.4^{2}$},
\scalebox{0.75}{$1^{7}2^{3}3^{3}4^{5}.2^{4}.3^{4}.4^{2}$},
\scalebox{0.75}{$1^{7}2^{3}.2^{4}3^{4}.3^{3}4^{5}.4^{2}$},
\scalebox{0.75}{$1^{7}2^{3}3^{3}.2^{4}4^{4}.3^{4}.4^{3}$},
\scalebox{0.75}{$1^{7}2^{6}3.23^{5}4^{2}5^{6}.34^{4}.45$},
\scalebox{0.75}{$1^{7}2^{5}3^{3}4^{6}.2^{2}3^{2}.3^{2}4$},
\scalebox{0.75}{$1^{7}2^{2}4^{4}.2^{5}3^{3}4.3^{4}4^{2}$},
\scalebox{0.75}{$1^{7}2^{2}3^{5}4^{6}5.2^{5}345^{5}.3.5$},
\scalebox{0.75}{$1^{7}2^{7}3^{2}5.3^{5}4^{4}5^{6}.4^{3}$},
\scalebox{0.75}{$1^{7}2^{3}3^{3}4^{3}.2^{4}.3^{4}.4^{4}$},
\scalebox{0.75}{$1^{7}234^{6}5^{2}.2^{6}3.3^{5}45.5^{4}$},
\scalebox{0.75}{$1^{7}234^{6}5.2^{6}345^{5}.3^{5}56^{7}$},
\scalebox{0.75}{$1^{7}2^{3}4.2^{4}3^{6}5.34^{5}5.45^{5}$},
\scalebox{0.75}{$1^{7}2^{2}3^{6}45.2^{5}34^{5}5^{5}.4.5$},
\scalebox{0.75}{$1^{7}2^{4}3^{3}4^{6}.2^{3}3^{2}.3^{2}4$},
\scalebox{0.75}{$1^{7}2^{3}3^{2}4^{6}.2^{4}3^{2}4.3^{3}$},
\scalebox{0.75}{$1^{7}2^{7}345^{3}.3^{6}4^{5}5^{2}.45.5$},
\scalebox{0.75}{$1^{7}2345^{6}.2^{6}56^{6}.3^{6}6.4^{6}$},
\scalebox{0.75}{$1^{7}2^{2}3^{6}4.2^{5}34^{6}5^{4}.5^{3}$},
\scalebox{0.75}{$1^{7}234^{5}.2^{6}345^{6}.3^{5}46^{7}.5$},
\scalebox{0.75}{$1^{7}23^{7}.2^{6}4^{2}5^{4}.4^{5}.5^{3}$},
\scalebox{0.75}{$1^{7}2^{3}34^{3}.2^{4}3^{2}.3^{4}.4^{4}$},
\scalebox{0.75}{$1^{7}234^{5}.2^{6}345^{4}.3^{5}45^{2}.5$},
\scalebox{0.75}{$1^{7}2^{4}3^{2}4^{6}.2^{3}3^{3}.3^{2}.4$},
\scalebox{0.75}{$1^{7}2^{7}3^{3}.3^{4}4^{5}5^{6}.4^{2}.5$},
\scalebox{0.75}{$1^{7}2^{4}.2^{3}3^{5}4^{2}.3^{2}4^{4}.4$},
\scalebox{0.75}{$1^{7}2^{2}34^{3}.2^{5}3^{2}4.3^{4}4^{3}$},
\scalebox{0.75}{$1^{7}2^{7}35^{3}.3^{6}4^{3}.4^{4}.5^{4}$},
\scalebox{0.75}{$1^{7}23^{4}4^{2}.2^{6}3^{2}4^{3}.34^{2}$},
\scalebox{0.75}{$1^{7}23^{6}45.2^{6}34^{5}5^{2}.45^{3}.5$},
\scalebox{0.75}{$1^{7}2^{2}34.2^{5}3^{2}4^{3}.3^{4}4^{3}$},
\scalebox{0.75}{$1^{7}23^{5}4^{2}5.2^{6}35^{5}.34^{4}5.4$},
\scalebox{0.75}{$1^{7}23^{2}4^{2}.2^{6}34^{2}.3^{4}4^{3}$},
\scalebox{0.75}{$1^{7}23^{3}.2^{6}4^{3}.3^{4}4^{2}.4^{2}$},
\scalebox{0.75}{$1^{7}2^{2}34^{3}.2^{5}3^{2}.3^{4}.4^{4}$},
\scalebox{0.75}{$1^{7}2^{2}3^{2}.2^{5}34^{4}.3^{4}.4^{3}$},
\scalebox{0.75}{$1^{7}2^{5}.2^{2}3^{7}45^{2}.4^{6}5^{4}.5$},
\scalebox{0.75}{$1^{7}23^{7}.2^{6}4^{2}5^{3}.4^{5}5^{3}.5$},
\scalebox{0.75}{$1^{7}2^{2}34^{3}.2^{5}3^{2}.3^{4}4^{3}.4$},
\scalebox{0.75}{$1^{7}2^{3}3^{2}.2^{4}34^{4}.3^{4}4.4^{2}$},
\scalebox{0.75}{$1^{7}2^{7}3.3^{6}4^{7}5.5^{6}6^{3}.6^{4}$},
\scalebox{0.75}{$1^{7}23^{4}4^{5}5.2^{6}34^{2}5^{6}.3^{2}$},
\scalebox{0.75}{$1^{7}2^{4}4^{5}.2^{3}3^{5}45^{7}.3^{2}.4$},
\scalebox{0.75}{$1^{7}234^{4}5.2^{6}4^{3}5^{4}.3^{6}5^{2}$},
\scalebox{0.75}{$1^{7}234^{2}.2^{6}3^{2}.3^{4}4^{2}.4^{3}$},
\scalebox{0.75}{$1^{7}23^{7}5.2^{6}4^{7}.5^{6}6^{2}.6^{5}$},
\scalebox{0.75}{$1^{7}2^{7}34.3^{6}5^{7}6^{2}.4^{6}.6^{5}$},
\scalebox{0.75}{$1^{7}2^{2}3^{3}4^{3}.2^{5}3^{3}.34^{3}.4$},
\scalebox{0.75}{$1^{7}23^{7}.2^{6}4^{2}5^{6}6.4^{5}56^{6}$},
\scalebox{0.75}{$1^{7}2^{6}345.23^{5}5^{6}.34^{5}6^{6}.46$},
\scalebox{0.75}{$1^{7}2^{5}3^{3}4^{5}.2^{2}3^{2}4.3^{2}.4$},
\scalebox{0.75}{$1^{7}2^{2}3.2^{5}3^{2}4^{3}.3^{4}4^{3}.4$},
\scalebox{0.75}{$1^{7}2^{3}34.2^{4}3^{6}.4^{6}5^{3}.5^{4}$},
\scalebox{0.75}{$1^{7}23^{4}4^{4}5.2^{6}34^{3}5^{6}.3^{2}$},
\scalebox{0.75}{$1^{7}234^{3}.2^{6}5^{5}.3^{6}.4^{4}5^{2}$},
\scalebox{0.75}{$1^{7}2^{3}34.2^{4}3^{3}.3^{3}4^{3}.4^{3}$},
\scalebox{0.75}{$1^{7}23^{6}5.2^{6}34^{6}5.45^{4}6^{6}.56$},
\scalebox{0.75}{$1^{7}2^{7}3.3^{6}4^{3}5^{4}.4^{4}5.5^{2}$},
\scalebox{0.75}{$1^{7}2^{2}3^{4}4^{3}.2^{5}3^{2}.34^{3}.4$},
\scalebox{0.75}{$1^{7}234^{2}.2^{6}3^{2}.3^{4}4^{3}.4^{2}$},
\scalebox{0.75}{$1^{7}23^{2}4^{4}.2^{6}3^{2}4^{2}.3^{3}.4$},
\scalebox{0.75}{$1^{7}2^{4}34^{4}.2^{3}3^{3}4.3^{3}.4^{2}$},
\scalebox{0.75}{$1^{7}2^{3}34^{3}.2^{4}3^{2}4^{3}.3^{4}.4$},
\scalebox{0.75}{$1^{7}2^{2}3.2^{5}3^{2}4.3^{4}4^{3}.4^{3}$},
\scalebox{0.75}{$1^{7}2^{7}3.3^{6}4^{3}5.4^{4}5^{2}.5^{4}$},
\scalebox{0.75}{$1^{7}234^{5}5^{2}.2^{6}4^{2}5^{4}.3^{6}5$},
\scalebox{0.75}{$1^{7}234^{6}56.2^{6}45^{5}.3^{6}56^{5}.6$},
\scalebox{0.75}{$1^{7}2^{3}3.2^{4}3^{6}4.4^{6}5^{4}.5^{3}$},
\scalebox{0.75}{$1^{7}2^{3}34.2^{4}3^{2}.3^{4}4^{2}.4^{4}$},
\scalebox{0.75}{$1^{7}23^{3}4^{3}.2^{6}3^{2}4^{3}.3^{2}.4$},
\scalebox{0.75}{$1^{7}2^{3}3^{3}5.2^{4}3^{4}4.4^{6}5.5^{5}$},
\scalebox{0.75}{$1^{7}2^{7}34.3^{6}4^{3}5.4^{3}5^{4}.5^{2}$},
\scalebox{0.75}{$1^{7}23^{7}4^{2}5^{2}.2^{6}4^{4}5^{4}.4.5$},
\scalebox{0.75}{$1^{7}2^{7}3^{3}5.3^{4}4^{7}6.5^{6}6.6^{5}$},
\scalebox{0.75}{$1^{7}2^{2}3^{6}.2^{5}34^{7}56.5^{6}.6^{6}$},
\scalebox{0.75}{$1^{7}234^{7}6^{2}.2^{6}5^{6}.3^{6}6^{5}.5$},
\scalebox{0.75}{$1^{7}2^{3}34.2^{4}3^{3}4.3^{3}4^{3}.4^{2}$},
\scalebox{0.75}{$1^{7}23^{3}4^{5}.2^{6}34^{2}5^{6}.3^{3}.5$},
\scalebox{0.75}{$1^{7}234^{6}5^{3}.2^{6}45^{3}6^{7}.3^{6}5$},
\scalebox{0.75}{$1^{7}23^{3}4^{2}5.2^{6}45^{6}.3^{4}.4^{4}$},
\scalebox{0.75}{$1^{7}2^{7}34.3^{6}4^{3}5^{3}.4^{3}5.5^{3}$},
\scalebox{0.75}{$1^{7}2^{5}3^{4}4^{3}.2^{2}34^{2}.3^{2}4.4$},
\scalebox{0.75}{$1^{7}2^{7}35.3^{6}4^{2}5^{3}.4^{5}5.5^{2}$},
\scalebox{0.75}{$1^{7}2^{4}34.2^{3}3^{3}4^{3}.3^{3}4.4^{2}$},
\scalebox{0.75}{$1^{7}234^{7}.2^{6}5^{6}6^{4}.3^{6}6^{3}.5$},
\scalebox{0.75}{$1^{7}24^{5}.2^{6}3^{2}5^{5}.3^{5}45^{2}.4$},
\scalebox{0.75}{$1^{7}2^{3}34.2^{4}3^{6}4.4^{5}5^{5}.5^{2}$},
\scalebox{0.75}{$1^{7}2^{3}3^{3}4^{2}.2^{4}34.3^{3}4.4^{3}$},
\scalebox{0.75}{$1^{7}2^{7}34.3^{6}4^{6}5.5^{6}6^{3}.6^{4}$},
\scalebox{0.75}{$1^{7}234^{6}5.2^{6}3^{2}45^{2}.3^{4}5^{4}$},
\scalebox{0.75}{$1^{7}2^{5}3^{6}4^{5}5^{5}.2^{2}4.34.5^{2}$},
\scalebox{0.75}{$1^{7}2^{3}34^{3}.2^{4}3^{4}4.3^{2}4.4^{2}$},
\scalebox{0.75}{$1^{7}23^{5}5.2^{6}3^{2}4.4^{6}5^{3}.5^{3}$},
\scalebox{0.75}{$1^{7}2^{6}3^{7}4^{2}.24^{4}5^{6}6^{7}.4.5$},
\scalebox{0.75}{$1^{7}23^{3}4^{4}5^{5}.2^{6}.3^{4}5.4^{3}5$},
\scalebox{0.75}{$1^{7}2^{4}3^{3}4^{3}.2^{3}34.3^{3}4.4^{2}$},
\scalebox{0.75}{$1^{7}2^{4}34.2^{3}3^{3}4.3^{3}4^{3}.4^{2}$},
\scalebox{0.75}{$1^{7}23^{7}4.2^{6}5^{6}6^{3}.4^{6}6^{4}.5$},
\scalebox{0.75}{$1^{7}2^{4}3^{4}4^{3}.2^{3}34.3^{2}4^{2}.4$},
\scalebox{0.75}{$1^{7}234^{4}5.2^{6}4^{3}5^{4}.3^{6}.5^{2}$},
\scalebox{0.75}{$1^{7}23^{5}45^{2}.2^{6}3^{2}4.4^{5}5^{4}.5$},
\scalebox{0.75}{$1^{7}2^{5}3^{4}4^{6}5^{5}.2^{2}35.3^{2}4.5$},
\scalebox{0.75}{$1^{7}2^{2}345^{6}.2^{5}34^{2}.3^{5}.4^{4}5$},
\scalebox{0.75}{$1^{7}2^{2}3^{3}45.2^{5}4^{5}.3^{4}45.5^{5}$},
\scalebox{0.75}{$1^{7}2345^{5}.2^{6}5^{2}6^{6}.3^{6}6.4^{6}$},
\scalebox{0.75}{$1^{7}23^{2}4^{7}6.2^{6}5^{6}6^{5}.3^{5}6.5$},
\scalebox{0.75}{$1^{7}2^{3}3^{3}4^{5}5^{6}.2^{4}3.3^{3}45.4$},
\scalebox{0.75}{$1^{7}2^{2}3^{2}4^{6}5.2^{5}3.3^{4}45.5^{5}$},
\scalebox{0.75}{$1^{7}2^{3}3^{2}4^{4}.2^{4}3^{2}4^{3}.3^{3}$},
\scalebox{0.75}{$1^{7}2^{2}3^{2}4^{4}.2^{5}3^{3}4^{3}.3^{2}$},
\scalebox{0.75}{$1^{7}234^{6}5^{2}.2^{6}5^{5}6^{6}.3^{6}6.4$},
\scalebox{0.75}{$1^{7}2^{2}3^{3}4^{5}5.2^{5}3^{3}5^{6}.34.4$},
\scalebox{0.75}{$1^{7}23^{6}5.2^{6}34^{4}5^{3}.4^{3}5.5^{2}$},
\scalebox{0.75}{$1^{7}234^{2}5^{4}.2^{6}45.3^{6}.4^{4}5^{2}$},
\scalebox{0.75}{$1^{7}2^{7}356.3^{6}4^{7}6^{2}.5^{6}6.6^{3}$},
\scalebox{0.75}{$1^{7}2^{6}35^{3}.23^{5}4^{4}5^{4}.34^{2}.4$},
\scalebox{0.75}{$1^{7}2^{7}35^{2}.3^{6}4^{2}5^{3}.4^{5}5^{2}$},
\scalebox{0.75}{$1^{7}234^{2}5^{2}.2^{6}35^{4}.3^{5}4.4^{4}5$},
\scalebox{0.75}{$1^{7}2^{2}34^{5}5.2^{5}3^{2}45.3^{4}4.5^{5}$},
\scalebox{0.75}{$1^{7}2^{2}3^{3}.2^{5}3^{2}4^{5}.3^{2}.4^{2}$},
\scalebox{0.75}{$1^{7}2^{4}4^{2}.2^{3}3^{7}5^{2}.4^{5}.5^{5}$},
\scalebox{0.75}{$1^{7}2^{4}3^{5}4^{3}.2^{3}4^{2}.3^{2}.4^{2}$},
\scalebox{0.75}{$1^{7}2^{2}3^{6}4^{4}.2^{5}45^{7}6^{6}.34.46$},
\scalebox{0.75}{$1^{7}2^{2}3^{3}4^{4}5.2^{5}4^{3}5^{6}.3^{4}$},
\scalebox{0.75}{$1^{7}2^{3}.2^{4}3^{4}4^{2}.3^{3}4^{3}.4^{2}$},
\scalebox{0.75}{$1^{7}23^{2}.2^{6}3^{2}45.3^{3}4^{5}5.45^{5}$},
\scalebox{0.75}{$1^{7}2^{5}3^{5}4^{3}.2^{2}4^{2}.3^{2}.4^{2}$},
\scalebox{0.75}{$1^{7}23^{2}4^{7}.2^{6}5^{7}6^{3}.3^{5}6^{4}$},
\scalebox{0.75}{$1^{7}2^{6}3^{3}5.23^{3}4^{6}5.345^{3}.5^{2}$},
\scalebox{0.75}{$1^{7}2^{3}3^{2}5.2^{4}3^{5}45^{2}.4^{6}5^{4}$},
\scalebox{0.75}{$1^{7}2^{3}3^{4}45^{2}.2^{4}3^{3}.4^{6}.5^{5}$},
\scalebox{0.75}{$1^{7}23^{2}45^{4}.2^{6}34.3^{4}45^{2}.4^{4}5$},
\scalebox{0.75}{$1^{7}2^{3}3^{3}4^{4}.2^{4}3^{2}4.3^{2}.4^{2}$},
\scalebox{0.75}{$1^{7}2^{3}45^{6}6.2^{4}3^{6}6^{5}.34^{5}5.46$},
\scalebox{0.75}{$1^{7}2^{2}3^{3}4^{3}.2^{5}34^{2}.3^{3}.4^{2}$},
\scalebox{0.75}{$1^{7}2^{7}3.3^{6}4^{7}5^{2}.5^{5}6^{3}.6^{4}$},
\scalebox{0.75}{$1^{7}2^{3}34.2^{4}3^{2}45^{5}.3^{4}45.4^{4}5$},
\scalebox{0.75}{$1^{7}23^{2}4^{3}.2^{6}3^{2}4^{2}.3^{3}.4^{2}$},
\scalebox{0.75}{$1^{7}2^{3}3.2^{4}3^{4}4^{2}.3^{2}4^{3}.4^{2}$},
\scalebox{0.75}{$1^{7}2^{6}34^{2}5^{3}.23^{5}4.34^{3}5^{3}.45$},
\scalebox{0.75}{$1^{7}2^{4}34^{3}.2^{3}3^{4}4^{2}.3^{2}.4^{2}$},
\scalebox{0.75}{$1^{7}2^{7}3^{2}.3^{5}4^{5}5^{4}.4^{2}5.5^{2}$},
\scalebox{0.75}{$1^{7}2^{7}3^{4}4.3^{3}4^{3}5^{4}.4^{3}.5^{3}$},
\scalebox{0.75}{$1^{7}2^{3}3^{4}.2^{4}4^{7}5.3^{3}5^{3}.5^{3}$},
\scalebox{0.75}{$1^{7}2^{4}3^{2}4^{2}.2^{3}3^{2}4^{4}.3^{3}.4$},
\scalebox{0.75}{$1^{7}23^{2}4.2^{6}345.3^{4}4^{3}5.4^{2}5^{5}$},
\scalebox{0.75}{$1^{7}234^{6}.2^{6}35^{6}6^{2}.3^{5}456^{4}.6$},
\scalebox{0.75}{$1^{7}2^{3}34^{2}.2^{4}3^{3}4^{2}.3^{3}.4^{3}$},
\scalebox{0.75}{$1^{7}2^{2}3^{3}4^{3}.2^{5}3^{2}4^{3}.3^{2}.4$},
\scalebox{0.75}{$1^{7}2^{2}34^{5}.2^{5}34^{2}5^{4}.3^{5}5^{3}$},
\scalebox{0.75}{$1^{7}234^{5}5^{3}.2^{6}3456^{7}.3^{5}5^{3}.4$},
\scalebox{0.75}{$1^{7}2^{3}345^{2}.2^{4}3^{6}5^{2}.4^{6}5^{3}$},
\scalebox{0.75}{$1^{7}23^{3}4^{2}.2^{6}3^{2}.3^{2}4^{3}.4^{2}$},
\scalebox{0.75}{$1^{7}25^{3}.2^{6}3^{2}.3^{5}4^{3}.4^{4}5^{4}$},
\scalebox{0.75}{$1^{7}2^{5}3^{4}4^{2}.2^{2}34^{3}.3^{2}.4^{2}$},
\scalebox{0.75}{$1^{7}2^{3}34^{2}.2^{4}3^{6}.4^{5}5^{5}.5^{2}$},
\scalebox{0.75}{$1^{7}2^{3}3^{2}.2^{4}3^{2}4^{4}.3^{3}4^{2}.4$},
\scalebox{0.75}{$1^{7}234^{6}5.2^{6}35^{5}.3^{5}456^{4}.6^{3}$},
\scalebox{0.75}{$1^{7}23^{4}4^{3}.2^{6}4^{2}5^{7}.3^{3}.4^{2}$},
\scalebox{0.75}{$1^{7}2^{6}.23^{7}4^{2}5^{3}.4^{5}6^{7}.5^{4}$},
\scalebox{0.75}{$1^{7}2^{3}3^{2}4.2^{4}3^{5}5^{4}.4^{6}.5^{3}$},
\scalebox{0.75}{$1^{7}25^{5}.2^{6}3^{2}.3^{5}4^{3}.4^{4}5^{2}$},
\scalebox{0.75}{$1^{7}2^{4}3^{3}4^{2}.2^{3}3^{2}4.3^{2}4^{3}.4$},
\scalebox{0.75}{$1^{7}2^{2}3^{3}4.2^{5}3^{4}5.4^{6}5^{3}.5^{3}$},
\scalebox{0.75}{$1^{7}23^{3}4^{2}5^{4}.2^{6}4^{4}.3^{4}4.5^{3}$},
\scalebox{0.75}{$1^{7}23^{2}4^{7}5^{2}.2^{6}35^{4}6^{7}.3^{4}5$},
\scalebox{0.75}{$1^{7}2^{3}3^{2}4.2^{4}3^{2}4.3^{3}4^{2}.4^{3}$},
\scalebox{0.75}{$1^{7}2^{3}3^{4}4^{4}5^{6}.2^{4}.3^{3}4.4^{2}5$},
\scalebox{0.75}{$1^{7}23^{7}45^{2}.2^{6}4^{6}.5^{5}6^{4}.6^{3}$},
\scalebox{0.75}{$1^{7}2^{2}3^{2}4^{6}5^{4}.2^{5}3.3^{4}4.5^{3}$},
\scalebox{0.75}{$1^{7}2^{4}3^{4}45^{3}.2^{3}3^{3}4.4^{5}.5^{4}$},
\scalebox{0.75}{$1^{7}235^{3}.2^{6}4^{3}.3^{6}5^{2}.4^{4}5^{2}$},
\scalebox{0.75}{$1^{7}2^{7}3^{2}4.3^{5}4^{6}5.5^{6}6^{2}.6^{5}$},
\scalebox{0.75}{$1^{7}23^{2}4^{2}.2^{6}3^{2}4^{2}.3^{3}4.4^{2}$},
\scalebox{0.75}{$1^{7}2^{3}34.2^{4}3^{6}5^{2}.4^{6}5^{3}.5^{2}$},
\scalebox{0.75}{$1^{7}2^{7}34.3^{6}4^{2}5^{2}.4^{4}5^{2}.5^{3}$},
\scalebox{0.75}{$1^{7}2^{3}34^{4}5^{6}.2^{4}3^{2}4.3^{4}4^{2}5$},
\scalebox{0.75}{$1^{7}23^{5}5^{2}.2^{6}3^{2}4.4^{6}5^{3}.5^{2}$},
\scalebox{0.75}{$1^{7}2^{2}34^{6}5^{2}.2^{5}3^{2}5^{5}.3^{4}.4$},
\scalebox{0.75}{$1^{7}23^{4}4^{4}5^{2}.2^{6}4^{2}5^{5}.3^{3}.4$},
\scalebox{0.75}{$1^{7}2^{2}3^{3}.2^{5}4^{6}5.3^{4}45^{2}.5^{4}$},
\scalebox{0.75}{$1^{7}23^{4}4^{2}5^{3}.2^{6}4^{5}5.3^{3}.5^{3}$},
\scalebox{0.75}{$1^{7}23^{6}456.2^{6}34^{6}56^{2}.5^{5}6.6^{3}$},
\scalebox{0.75}{$1^{7}2^{3}3^{4}5.2^{4}3^{3}4^{2}5.4^{5}.5^{5}$},
\scalebox{0.75}{$1^{7}2^{2}35.2^{5}3^{2}4^{3}.3^{4}5^{6}.4^{4}$},
\scalebox{0.75}{$1^{7}2^{3}3^{3}4^{4}5^{6}.2^{4}4.3^{4}.4^{2}5$},
\scalebox{0.75}{$1^{7}23^{2}45^{4}.2^{6}34^{3}.3^{4}4^{3}5^{3}$},
\scalebox{0.75}{$1^{7}23^{7}4.2^{6}45^{5}6^{2}.4^{5}5^{2}6^{5}$},
\scalebox{0.75}{$1^{7}2^{2}3^{3}45^{4}.2^{5}3^{4}.4^{6}5^{2}.5$},
\scalebox{0.75}{$1^{7}24^{6}5^{4}.2^{6}3^{2}5^{3}6^{7}.3^{5}.4$},
\scalebox{0.75}{$1^{7}2^{7}3^{2}4.3^{5}4^{2}5^{2}.4^{4}5.5^{4}$},
\scalebox{0.75}{$1^{7}2^{6}35^{3}.23^{6}4^{5}5^{2}.4^{2}.5^{2}$},
\scalebox{0.75}{$1^{7}2^{3}345^{3}.2^{4}3^{6}5^{2}.4^{6}.5^{2}$},
\scalebox{0.75}{$1^{7}23^{5}4^{3}5^{3}.2^{6}34^{2}5^{4}.34^{2}$},
\scalebox{0.75}{$1^{7}23^{2}4^{4}5^{3}.2^{6}34^{2}5^{4}.3^{4}4$},
\scalebox{0.75}{$1^{7}234^{6}56.2^{6}3^{2}5^{6}6.3^{4}46.6^{4}$},
\scalebox{0.75}{$1^{7}2^{3}3^{4}.2^{4}3^{3}45^{4}.4^{6}5.5^{2}$},
\scalebox{0.75}{$1^{7}24^{5}5.2^{6}3^{2}5^{2}.3^{5}4^{2}.5^{4}$},
\scalebox{0.75}{$1^{7}2^{3}3^{2}4.2^{4}3^{5}5^{2}.4^{6}5.5^{4}$},
\scalebox{0.75}{$1^{7}2^{3}.2^{4}3^{5}4.3^{2}4^{5}5^{3}.45^{4}$},
\scalebox{0.75}{$1^{7}2^{3}34^{2}5.2^{4}3^{6}.4^{5}5^{4}.5^{2}$},
\scalebox{0.75}{$1^{7}23^{3}4^{3}5^{2}.2^{6}4^{3}5.3^{4}45^{4}$},
\scalebox{0.75}{$1^{7}2^{2}3^{4}4^{2}.2^{5}3^{2}4^{5}5.35^{5}.5$},
\scalebox{0.75}{$1^{7}2^{7}34.3^{6}4^{6}5^{3}6.5^{4}6^{2}.6^{4}$},
\scalebox{0.75}{$1^{7}2^{3}34^{5}5^{5}.2^{4}3^{3}4.3^{3}.45^{2}$},
\scalebox{0.75}{$1^{7}2^{2}3^{4}4.2^{5}3^{3}4^{2}5.4^{4}5^{5}.5$},
\scalebox{0.75}{$1^{7}23^{7}.2^{6}4^{3}5^{2}6.4^{4}5^{5}6.6^{5}$},
\scalebox{0.75}{$1^{7}2^{3}3^{3}45.2^{4}3^{4}4.4^{5}5^{2}.5^{4}$},
\scalebox{0.75}{$1^{7}245^{3}.2^{6}3^{2}.3^{5}4^{2}5.4^{4}5^{3}$},
\scalebox{0.75}{$1^{7}23^{3}4^{2}5.2^{6}4^{4}5^{3}.3^{4}4.5^{3}$},
\scalebox{0.75}{$1^{7}234^{5}5^{3}.2^{6}4^{2}5^{3}6^{6}.3^{6}56$},
\scalebox{0.75}{$1^{7}23^{2}4^{2}.2^{6}34^{2}5.3^{4}4^{3}.5^{6}$},
\scalebox{0.75}{$1^{7}2^{3}4.2^{4}3^{5}4.3^{2}4^{4}5^{3}.45^{4}$},
\scalebox{0.75}{$1^{7}23^{2}4^{2}5^{6}.2^{6}3.3^{4}4^{2}5.4^{3}$},
\scalebox{0.75}{$1^{7}23^{2}5^{4}.2^{6}4^{2}.3^{5}45.4^{4}5^{2}$},
\scalebox{0.75}{$1^{7}23^{2}4^{3}.2^{6}34^{3}.3^{4}45^{3}.5^{4}$},
\scalebox{0.75}{$1^{7}2^{2}3^{4}.2^{5}3^{2}4^{3}5.34^{4}5^{5}.5$},
\scalebox{0.75}{$1^{7}2^{7}345.3^{6}4^{6}5^{2}.5^{4}6^{5}.6^{2}$},
\scalebox{0.75}{$1^{7}23^{2}4^{3}.2^{6}345^{4}.3^{4}4^{3}.5^{3}$},
\scalebox{0.75}{$1^{7}2^{6}3^{4}4^{4}5^{5}6^{6}.23^{3}5.4^{3}56$},
\scalebox{0.75}{$1^{7}2^{2}34^{6}5.2^{5}3^{3}45^{4}.3^{3}.5^{2}$},
\scalebox{0.75}{$1^{7}2^{7}35.3^{6}4^{3}5^{2}6.4^{4}5^{4}.6^{6}$},
\scalebox{0.75}{$1^{7}234^{5}6.2^{6}4^{2}5^{5}6.3^{6}5^{2}6^{5}$},
\scalebox{0.75}{$1^{7}2^{3}34^{2}.2^{4}3^{2}45^{6}.3^{4}5.4^{4}$},
\scalebox{0.75}{$1^{7}2^{2}3^{3}45^{3}.2^{5}34^{6}5.3^{3}.5^{3}$},
\scalebox{0.75}{$1^{7}23^{6}4^{2}.2^{6}356^{7}.4^{5}5^{2}.5^{4}$},
\scalebox{0.75}{$1^{7}23^{7}4^{2}5^{2}.2^{6}4^{4}5^{2}.45^{2}.5$},
\scalebox{0.75}{$1^{7}23^{5}45^{2}.2^{6}3^{2}4^{2}.4^{4}5.5^{4}$},
\scalebox{0.75}{$1^{7}2^{3}3^{2}5.2^{4}3^{3}4.3^{2}4^{5}.45^{6}$},
\scalebox{0.75}{$1^{7}2^{2}3^{5}5.2^{5}3^{2}45.4^{6}5^{2}.5^{3}$},
\scalebox{0.75}{$1^{7}23^{3}5^{6}.2^{6}4^{4}56^{6}.3^{4}6.4^{3}$},
\scalebox{0.75}{$1^{7}23^{3}4^{6}5^{5}.2^{6}456^{7}7^{7}.3^{4}.5$},
\scalebox{0.75}{$1^{7}2^{7}346.3^{6}5^{3}6^{3}.4^{6}6.5^{4}6^{2}$},
\scalebox{0.75}{$1^{7}2^{2}34^{3}.2^{5}3^{2}4^{3}5.3^{4}5^{5}.45$},
\scalebox{0.75}{$1^{7}23^{4}5.2^{6}3^{2}4^{2}5^{2}.34^{4}5^{4}.4$},
\scalebox{0.75}{$1^{7}23^{2}4^{2}5.2^{6}5^{6}.3^{5}46^{6}.4^{4}6$},
\scalebox{0.75}{$1^{7}2^{2}34^{3}5^{3}.2^{5}34^{3}5.3^{5}5^{3}.4$},
\scalebox{0.75}{$1^{7}23^{7}56.2^{6}4^{2}5^{3}6.4^{5}6^{5}.5^{3}$},
\scalebox{0.75}{$1^{7}234^{7}56.2^{6}5^{3}6^{4}.3^{6}6^{2}.5^{3}$},
\scalebox{0.75}{$1^{7}234^{5}56^{2}.2^{6}5^{6}6.3^{6}6^{4}.4^{2}$},
\scalebox{0.75}{$1^{7}234^{4}5^{3}6.2^{6}4^{3}56^{5}.3^{6}5^{3}6$},
\scalebox{0.75}{$1^{7}2345^{3}6.2^{6}5^{4}6^{2}.3^{6}6^{4}.4^{6}$},
\scalebox{0.75}{$1^{7}2^{7}34.3^{6}4^{2}5^{2}6^{6}.4^{4}5.5^{4}6$},
\scalebox{0.75}{$1^{7}23^{3}45^{3}.2^{6}34^{3}.3^{3}45^{4}.4^{2}$},
\scalebox{0.75}{$1^{7}234^{6}56^{5}.2^{6}45^{6}6^{2}7^{6}.3^{6}7$},
\scalebox{0.75}{$1^{7}23^{2}45.2^{6}3^{2}4^{2}.3^{3}45^{6}.4^{3}$},
\scalebox{0.75}{$1^{7}24^{2}5^{3}.2^{6}3^{2}5.3^{5}45.4^{4}5^{2}$},
\scalebox{0.75}{$1^{7}2^{7}3^{2}45.3^{5}45^{5}.4^{5}56^{5}.6^{2}$},
\scalebox{0.75}{$1^{7}24^{5}6^{6}.2^{6}3^{2}5^{5}.3^{5}45^{2}6.4$},
\scalebox{0.75}{$1^{7}23^{2}4^{2}.2^{6}5^{6}.3^{5}46^{6}.4^{4}56$},
\scalebox{0.75}{$1^{7}2^{2}4^{2}5^{4}.2^{5}3^{3}5.3^{4}45.4^{4}5$},
\scalebox{0.75}{$1^{7}23^{5}45^{5}.2^{6}3^{2}4^{3}56^{7}.4^{3}.5$},
\scalebox{0.75}{$1^{7}23^{5}4^{5}6^{7}.2^{6}3^{2}4^{2}5^{7}7^{7}$},
\scalebox{0.75}{$1^{7}2^{2}3^{2}45^{4}.2^{5}3^{4}5^{2}.34^{5}5.4$},
\scalebox{0.75}{$1^{7}2^{7}35.3^{6}4^{2}5^{5}6^{3}.4^{5}6^{3}.56$},
\scalebox{0.75}{$1^{7}2^{6}3^{2}4^{4}.23^{4}5^{6}6^{7}.34^{2}.45$},
\scalebox{0.75}{$1^{7}2345.2^{6}4^{2}5^{2}.3^{6}5^{2}.4^{4}5^{2}$},
\scalebox{0.75}{$1^{7}2^{2}3^{7}4^{7}5^{6}.2^{5}56^{7}7^{7}8^{7}$},
\scalebox{0.75}{$1^{7}2^{2}3^{3}4^{2}.2^{5}3^{2}4^{3}.3^{2}.4^{2}$},
\scalebox{0.75}{$1^{7}2^{7}345.3^{6}4^{6}7.5^{6}6^{3}7.6^{4}7^{5}$},
\scalebox{0.75}{$1^{7}2^{2}3^{2}5^{5}.2^{5}4^{4}.3^{5}5^{2}.4^{3}$},
\scalebox{0.75}{$1^{7}2345^{4}.2^{6}45^{2}6^{5}.3^{6}56^{2}.4^{5}$},
\scalebox{0.75}{$1^{7}2^{3}3^{2}4^{2}.2^{4}3^{2}.3^{3}4^{3}.4^{2}$},
\scalebox{0.75}{$1^{7}2^{2}3^{5}.2^{5}3^{2}4^{3}5^{4}.4^{4}.5^{3}$},
\scalebox{0.75}{$1^{7}2^{2}4^{3}5^{4}.2^{5}3^{3}.3^{4}5^{3}.4^{4}$},
\scalebox{0.75}{$1^{7}2^{7}34.3^{6}4^{6}57.5^{6}6^{3}7.6^{4}7^{5}$},
\scalebox{0.75}{$1^{7}2^{7}3^{2}46.3^{5}45^{5}6^{4}.4^{5}56^{2}.5$},
\scalebox{0.75}{$1^{7}23^{5}5.2^{6}3^{2}456^{4}.4^{6}5.5^{4}6^{3}$},
\scalebox{0.75}{$1^{7}2^{2}3^{3}4^{7}5^{2}.2^{5}3^{3}5^{5}6^{7}.3$},
\scalebox{0.75}{$1^{7}23^{2}4^{3}5^{2}.2^{6}4^{4}5^{3}.3^{5}5^{2}$},
\scalebox{0.75}{$1^{7}23^{3}4^{3}5^{5}.2^{6}4^{4}5^{2}6^{7}.3^{4}$},
\scalebox{0.75}{$1^{7}2^{3}3^{4}.2^{4}3^{3}4^{2}.4^{5}5^{3}.5^{4}$},
\scalebox{0.75}{$1^{7}2^{2}34^{3}5^{2}.2^{5}34^{3}5.3^{5}45^{3}.5$},
\scalebox{0.75}{$1^{7}23^{3}45.2^{6}5^{6}6.3^{4}4^{2}6^{5}.4^{4}6$},
\scalebox{0.75}{$1^{7}2^{3}3^{5}4^{4}5^{6}6^{6}.2^{4}4.3^{2}4.456$},
\scalebox{0.75}{$1^{7}234^{2}5^{2}.2^{6}3^{2}5.3^{4}45^{3}.4^{4}5$},
\scalebox{0.75}{$1^{7}23^{6}46^{4}.2^{6}345.4^{5}5^{2}.5^{4}6^{3}$},
\scalebox{0.75}{$1^{7}2^{7}34^{2}6.3^{6}4^{4}5^{2}.45^{4}6^{5}.56$},
\scalebox{0.75}{$1^{7}23^{2}5^{2}.2^{6}34^{2}5.3^{4}45^{3}.4^{4}5$},
\scalebox{0.75}{$1^{7}2^{7}3.3^{6}4^{3}5^{2}.4^{4}5^{5}6^{2}.6^{5}$},
\scalebox{0.75}{$1^{7}2^{3}3^{2}4^{2}.2^{4}3^{2}4^{2}.3^{3}4.4^{2}$},
\scalebox{0.75}{$1^{7}2^{2}3^{4}4^{4}.2^{5}3^{3}5^{3}.4^{3}5.5^{3}$},
\scalebox{0.75}{$1^{7}23^{2}4^{2}5^{5}.2^{6}34^{2}.3^{4}4^{3}5^{2}$},
\scalebox{0.75}{$1^{7}2^{2}3^{4}.2^{5}3^{3}45^{3}.4^{6}5^{2}.5^{2}$},
\scalebox{0.75}{$1^{7}234^{7}5.2^{6}56^{6}7.3^{6}67^{6}.5^{5}8^{7}$},
\scalebox{0.75}{$1^{7}2^{2}3^{4}4^{3}.2^{5}34^{2}5^{7}.3^{2}.4^{2}$},
\scalebox{0.75}{$1^{7}2^{7}3.3^{6}4^{4}5^{2}.4^{3}5^{2}6^{7}.5^{3}$},
\scalebox{0.75}{$1^{7}2^{5}3^{4}4^{3}5^{4}.2^{2}3^{3}.4^{4}5.5^{2}$},
\scalebox{0.75}{$1^{7}2^{2}3^{4}4.2^{5}3^{3}5^{2}.4^{6}5^{2}.5^{3}$},
\scalebox{0.75}{$1^{7}2^{6}35^{3}.23^{5}4^{2}56^{7}.34^{4}5^{2}.45$},
\scalebox{0.75}{$1^{7}2^{5}3^{3}4^{6}5^{5}.2^{2}3^{2}.3^{2}4.5^{2}$},
\scalebox{0.75}{$1^{7}2^{3}3^{2}4^{4}.2^{4}34^{3}5^{3}.3^{4}.5^{4}$},
\scalebox{0.75}{$1^{7}2^{2}3^{4}5.2^{5}3^{3}4^{3}.4^{4}5^{2}.5^{4}$},
\scalebox{0.75}{$1^{7}2^{3}4^{6}.2^{4}3^{4}45^{2}.3^{3}5^{2}.5^{3}$},
\scalebox{0.75}{$1^{7}2^{7}3^{2}4.3^{5}4^{2}5^{3}6^{7}.4^{4}.5^{4}$},
\scalebox{0.75}{$1^{7}2^{3}3^{2}.2^{4}34^{3}5^{3}.3^{4}5^{4}.4^{4}$},
\scalebox{0.75}{$1^{7}234^{2}5^{5}.2^{6}56^{6}7.3^{6}67^{6}.4^{5}5$},
\scalebox{0.75}{$1^{7}234^{7}7.2^{6}5^{7}67.3^{6}6^{2}7.6^{4}7^{4}$},
\scalebox{0.75}{$1^{7}23^{3}45^{3}6.2^{6}45^{4}6.3^{4}46^{5}.4^{4}$},
\scalebox{0.75}{$1^{7}2^{2}3^{2}4.2^{5}4^{3}5^{4}.3^{5}5^{3}.4^{3}$},
\scalebox{0.75}{$1^{7}2^{7}3.3^{6}4^{2}6^{3}.4^{5}5^{3}.5^{4}6^{4}$},
\scalebox{0.75}{$1^{7}234^{5}6^{3}.2^{6}4^{2}5^{5}.3^{6}5^{2}6^{4}$},
\scalebox{0.75}{$1^{7}2^{2}3^{2}4^{2}.2^{5}4^{5}5^{2}.3^{5}5.5^{4}$},
\scalebox{0.75}{$1^{7}23^{7}4^{4}.2^{6}4^{3}5^{4}6^{4}.5^{3}.6^{3}$},
\scalebox{0.75}{$1^{7}2^{3}3^{2}5^{3}.2^{4}34^{3}5^{4}.3^{4}.4^{4}$},
\scalebox{0.75}{$1^{7}2^{2}34^{3}.2^{5}3^{2}4^{3}5^{3}.3^{4}45^{4}$},
\scalebox{0.75}{$1^{7}2^{2}3^{6}.2^{5}34^{2}5^{2}.4^{5}5^{2}.5^{3}$},
\scalebox{0.75}{$1^{7}23^{3}4^{5}.2^{6}5^{7}6^{2}.3^{4}6^{5}.4^{2}$},
\scalebox{0.75}{$1^{7}23^{2}4^{3}5^{5}6.2^{6}4^{3}6^{6}.3^{5}45^{2}$},
\scalebox{0.75}{$1^{7}2^{2}4.2^{5}3^{3}5^{2}.3^{4}4^{3}5.4^{3}5^{4}$},
\scalebox{0.75}{$1^{7}23^{7}4^{3}.2^{6}4^{4}5^{2}6.5^{5}6^{2}.6^{4}$},
\scalebox{0.75}{$1^{7}2^{3}3^{2}4^{2}.2^{4}34^{4}5^{2}.3^{4}4.5^{5}$},
\scalebox{0.75}{$1^{7}2^{7}36^{2}.3^{6}4^{3}5^{4}.4^{4}5^{3}6.6^{4}$},
\scalebox{0.75}{$1^{7}2^{3}3^{4}45.2^{4}3^{3}4^{2}.4^{4}5^{2}.5^{4}$},
\scalebox{0.75}{$1^{7}23^{2}45^{5}6.2^{6}4^{2}56^{4}.3^{5}56.4^{4}6$},
\scalebox{0.75}{$1^{7}2^{5}3^{6}4^{3}5^{4}.2^{2}34^{2}.4^{2}5.5^{2}$},
\scalebox{0.75}{$1^{7}23^{3}4^{3}.2^{6}4^{3}5^{2}.3^{4}45^{2}.5^{3}$},
\scalebox{0.75}{$1^{7}24^{2}5.2^{6}3^{2}5^{2}.3^{5}4^{2}.4^{3}5^{4}$},
\scalebox{0.75}{$1^{7}2^{2}4^{3}6^{6}.2^{5}3^{3}5.3^{4}4^{4}6.5^{6}$},
\scalebox{0.75}{$1^{7}2^{2}3^{3}46.2^{5}3^{4}5^{3}.4^{6}6^{6}.5^{4}$},
\scalebox{0.75}{$1^{7}23^{2}5^{3}.2^{6}4^{4}.3^{5}5^{4}6.4^{3}6^{6}$},
\scalebox{0.75}{$1^{7}2^{7}3.3^{6}4^{2}5^{4}6.4^{5}5^{3}6^{3}.6^{3}$},
\scalebox{0.75}{$1^{7}2^{4}35^{2}.2^{3}3^{4}4^{3}.3^{2}4^{3}5^{5}.4$},
\scalebox{0.75}{$1^{7}2^{7}3^{3}4.3^{4}4^{2}5^{3}6^{6}.4^{4}6.5^{4}$},
\scalebox{0.75}{$1^{7}2^{7}3.3^{6}4^{3}5^{3}6.4^{4}5^{4}6^{2}.6^{4}$},
\scalebox{0.75}{$1^{7}2^{2}4^{4}5.2^{5}3^{3}5^{3}.3^{4}45^{3}.4^{2}$},
\scalebox{0.75}{$1^{7}2^{7}37.3^{6}4^{7}5^{2}.5^{5}6^{4}7^{6}.6^{3}$},
\scalebox{0.75}{$1^{7}23^{2}4^{3}56.2^{6}3456^{5}.3^{4}4^{3}6.5^{5}$},
\scalebox{0.75}{$1^{7}23^{6}456.2^{6}35^{5}7^{3}.4^{6}567^{4}.6^{5}$},
\scalebox{0.75}{$1^{7}234^{3}5^{2}.2^{6}5^{5}6^{3}.3^{6}6^{4}.4^{4}$},
\scalebox{0.75}{$1^{7}2^{2}34^{5}.2^{5}345^{5}6^{3}.3^{5}456^{3}.56$},
\scalebox{0.75}{$1^{7}2^{2}3^{4}45.2^{5}3^{3}4^{2}5^{3}.4^{4}.5^{3}$},
\scalebox{0.75}{$1^{7}23^{2}4^{3}5^{5}.2^{6}45^{2}6^{7}.3^{5}.4^{3}$},
\scalebox{0.75}{$1^{7}23^{7}4^{3}6.2^{6}4^{4}5^{3}6^{2}.5^{4}.6^{4}$},
\scalebox{0.75}{$1^{7}2^{3}345^{2}.2^{4}3^{2}4^{4}5.3^{4}4^{2}5^{4}$},
\scalebox{0.75}{$1^{7}2^{4}3.2^{3}3^{3}4^{3}.3^{3}4^{3}5^{3}.45^{4}$},
\scalebox{0.75}{$1^{7}23^{2}4^{5}5^{3}6.2^{6}35^{4}6^{6}.3^{4}4^{2}$},
\scalebox{0.75}{$1^{7}2^{2}3^{2}4^{3}.2^{5}34^{3}5^{3}.3^{4}4.5^{4}$},
\scalebox{0.75}{$1^{7}2^{2}34^{4}.2^{5}3^{2}4.3^{4}4^{2}5^{2}.5^{5}$},
\scalebox{0.75}{$1^{7}2^{4}3^{3}4^{5}5^{4}.2^{3}35^{2}.3^{3}.4^{2}5$},
\scalebox{0.75}{$1^{7}23^{2}4^{3}5^{2}.2^{6}4^{4}5^{2}.3^{5}5^{2}.5$},
\scalebox{0.75}{$1^{7}2^{2}3^{2}4.2^{5}34^{3}5^{2}.3^{4}5^{5}.4^{3}$},
\scalebox{0.75}{$1^{7}2^{7}34.3^{6}4^{2}5^{3}.4^{4}5^{4}6^{2}.6^{5}$},
\scalebox{0.75}{$1^{7}235^{4}.2^{6}4^{3}6^{4}.3^{6}5^{3}.4^{4}6^{3}$},
\scalebox{0.75}{$1^{7}2^{3}345^{3}.2^{4}3^{2}4^{4}5.3^{4}4^{2}5^{3}$},
\scalebox{0.75}{$1^{7}2^{2}3^{6}4.2^{5}34^{3}5^{3}.4^{3}5^{2}.5^{2}$},
\scalebox{0.75}{$1^{7}2^{5}3^{3}4^{3}5^{4}.2^{2}3^{3}.34^{3}.45^{3}$},
\scalebox{0.75}{$1^{7}2^{2}3^{4}4^{2}5.2^{5}3^{3}5.4^{5}5^{3}.5^{2}$},
\scalebox{0.75}{$1^{7}2^{7}36^{2}.3^{6}4^{3}5^{2}6^{4}.4^{4}5.5^{4}6$},
\scalebox{0.75}{$1^{7}2^{2}356.2^{5}3^{2}4^{3}.3^{4}5^{6}.4^{4}6^{6}$},
\scalebox{0.75}{$1^{7}2^{2}3^{2}45.2^{5}4^{2}5^{3}.3^{5}5^{2}.4^{4}5$},
\scalebox{0.75}{$1^{7}2345^{3}6^{2}.2^{6}3456.3^{5}5^{3}6.4^{5}6^{3}$},
\scalebox{0.75}{$1^{7}23^{6}4.2^{6}35^{3}6^{2}.4^{6}6^{3}.5^{4}6^{2}$},
\scalebox{0.75}{$1^{7}2^{5}34^{2}.2^{2}3^{4}4^{2}5^{6}.3^{2}4.4^{2}5$},
\scalebox{0.75}{$1^{7}234^{5}5.2^{6}356^{6}7.3^{5}4^{2}7^{5}.5^{5}67$},
\scalebox{0.75}{$1^{7}2^{5}7.2^{2}3^{7}4^{2}5.4^{5}5^{6}6.6^{6}7^{6}$},
\scalebox{0.75}{$1^{7}2^{2}3^{6}4.2^{5}35^{2}6^{6}.4^{6}67^{7}.5^{5}$},
\scalebox{0.75}{$1^{7}234^{5}6^{4}.2^{6}4^{2}5^{5}.3^{6}5^{2}6^{2}.6$},
\scalebox{0.75}{$1^{7}23^{5}5^{5}6^{2}.2^{6}3^{2}4^{4}56^{5}.4^{3}.5$},
\scalebox{0.75}{$1^{7}2^{7}35.3^{6}4^{2}6^{7}.4^{5}57^{7}.5^{5}8^{7}$},
\scalebox{0.75}{$1^{7}2^{3}3^{6}4.2^{4}345^{7}6^{2}.4^{5}6^{3}.6^{2}$},
\scalebox{0.75}{$1^{7}2^{2}3^{6}4^{7}5^{2}.2^{5}35^{4}6^{6}7^{7}.5.6$},
\scalebox{0.75}{$1^{7}23^{4}4^{5}5.2^{6}4^{2}5^{4}6^{6}.3^{3}.5^{2}6$},
\scalebox{0.75}{$1^{7}234^{3}5^{4}.2^{6}45^{3}6^{5}.3^{6}6^{2}.4^{3}$},
\scalebox{0.75}{$1^{7}2^{3}345^{3}.2^{4}3^{2}4^{2}5^{3}.3^{4}5.4^{4}$},
\scalebox{0.75}{$1^{7}23^{5}6.2^{6}3^{2}46^{2}.4^{6}5^{3}.5^{4}6^{4}$},
\scalebox{0.75}{$1^{7}2345^{3}6^{3}.2^{6}3456.3^{5}5^{3}6.4^{5}6^{2}$},
\scalebox{0.75}{$1^{7}2345^{5}7.2^{6}4^{2}6^{6}7.3^{6}57^{5}.4^{4}56$},
\scalebox{0.75}{$1^{7}234^{3}5^{4}.2^{6}45^{3}6^{3}.3^{6}6^{4}.4^{3}$},
\scalebox{0.75}{$1^{7}2^{4}3^{5}4^{3}5^{3}.2^{3}34^{2}5.34^{2}.5^{3}$},
\scalebox{0.75}{$1^{7}23^{2}45.2^{6}45^{4}6.3^{5}456^{4}.4^{4}56^{2}$},
\scalebox{0.75}{$1^{7}23^{7}.2^{6}4^{7}56.5^{6}6^{2}7^{3}.6^{4}7^{4}$},
\scalebox{0.75}{$1^{7}2^{3}3.2^{4}3^{2}4^{2}5^{3}.3^{4}45^{3}.4^{4}5$},
\scalebox{0.75}{$1^{7}2^{3}4.2^{4}3^{4}45^{2}.3^{3}4^{2}5^{4}.4^{3}5$},
\scalebox{0.75}{$1^{7}2^{7}34.3^{6}4^{6}5.5^{6}6^{2}7^{5}.6^{5}7^{2}$},
\scalebox{0.75}{$1^{7}23^{3}4^{4}5.2^{6}3^{2}4^{2}5^{3}.3^{2}5^{3}.4$},
\scalebox{0.75}{$1^{7}23^{5}.2^{6}3^{2}456^{2}.4^{6}5^{2}6.5^{4}6^{4}$},
\scalebox{0.75}{$1^{7}2^{7}3^{3}4.3^{4}4^{2}5^{2}6^{5}.4^{4}56.5^{4}6$},
\scalebox{0.75}{$1^{7}2^{2}345^{2}.2^{5}3^{2}4^{3}5.3^{4}4^{3}5.5^{3}$},
\scalebox{0.75}{$1^{7}2^{3}34^{2}5^{6}.2^{4}3^{3}4^{4}56^{6}.3^{3}4.6$},
\scalebox{0.75}{$1^{7}23^{5}4.2^{6}3^{2}56^{5}.4^{6}6^{2}7^{6}.5^{6}7$},
\scalebox{0.75}{$1^{7}23^{6}6.2^{6}34^{3}5.4^{4}5^{2}6^{4}.5^{4}6^{2}$},
\scalebox{0.75}{$1^{7}23^{3}4^{2}56^{5}.2^{6}5^{5}6^{2}.3^{4}4.4^{4}5$},
\scalebox{0.75}{$1^{7}2^{2}3^{3}4^{3}6.2^{5}3^{4}4^{3}6.45^{7}6^{4}.6$},
\scalebox{0.75}{$1^{7}23^{3}4^{2}5^{2}.2^{6}345^{3}.3^{3}45^{2}.4^{3}$},
\scalebox{0.75}{$1^{7}234^{4}6^{4}.2^{6}34^{2}.3^{5}45^{2}.5^{5}6^{3}$},
\scalebox{0.75}{$1^{7}2^{2}3^{6}4^{6}5^{6}8.2^{5}46^{7}7^{7}8^{6}.3.5$},
\scalebox{0.75}{$1^{7}234^{4}5^{4}6^{3}.2^{6}34^{2}5^{2}.3^{5}456^{4}$},
\scalebox{0.75}{$1^{7}2^{2}3^{3}4^{4}5^{4}.2^{5}3^{3}4^{3}56^{7}.35.5$},
\scalebox{0.75}{$1^{7}2345^{5}67.2^{6}3456^{4}7.3^{5}4567^{4}.4^{4}67$},
\scalebox{0.75}{$1^{7}2^{3}3^{4}4^{3}5^{6}6^{6}.2^{4}4.3^{3}45.4^{2}6$},
\scalebox{0.75}{$1^{7}23^{4}4^{3}5^{2}.2^{6}3^{2}4^{3}5.345^{2}.5^{2}$},
\scalebox{0.75}{$1^{7}234^{5}6^{3}.2^{6}4^{2}5.3^{6}56^{2}.5^{5}6^{2}$},
\scalebox{0.75}{$1^{7}2^{2}345^{3}6.2^{5}3456^{5}.3^{5}456.4^{4}5^{2}$},
\scalebox{0.75}{$1^{7}235^{4}6^{4}.2^{6}34^{2}6^{2}.3^{5}5^{3}.4^{5}6$},
\scalebox{0.75}{$1^{7}2^{2}34^{3}5.2^{5}3^{2}5^{6}.3^{4}4^{3}6^{6}.46$},
\scalebox{0.75}{$1^{7}23^{3}45^{2}.2^{6}4^{2}5^{4}6.3^{4}456^{6}.4^{3}$},
\scalebox{0.75}{$1^{7}23^{2}4^{3}5.2^{6}4^{4}56^{2}.3^{5}56^{4}.5^{4}6$},
\scalebox{0.75}{$1^{7}2^{7}3.3^{6}4^{4}56.4^{3}5^{3}6^{5}7.5^{3}67^{6}$},
\scalebox{0.75}{$1^{7}2^{3}4^{3}5^{3}.2^{4}3^{5}5^{4}.3^{2}4^{2}.4^{2}$},
\scalebox{0.75}{$1^{7}245^{5}6^{3}.2^{6}3^{2}56^{2}.3^{5}45.4^{5}6^{2}$},
\scalebox{0.75}{$1^{7}23^{2}45^{2}6^{4}.2^{6}45^{4}.3^{5}456^{3}.4^{4}$},
\scalebox{0.75}{$1^{7}2^{3}4^{3}5^{5}.2^{4}3^{4}4^{2}5^{2}.3^{3}.4^{2}$},
\scalebox{0.75}{$1^{7}234^{5}5^{5}.2^{6}4^{2}56^{4}7^{6}.3^{6}57.6^{3}$},
\scalebox{0.75}{$1^{7}23^{2}45^{2}.2^{6}45^{4}.3^{5}456^{4}.4^{4}6^{3}$},
\scalebox{0.75}{$1^{7}235^{5}.2^{6}4^{2}6^{6}.3^{6}567^{5}.4^{5}57^{2}$},
\scalebox{0.75}{$1^{7}234^{3}5^{5}6^{2}.2^{6}35^{2}6^{4}.3^{5}46.4^{3}$},
\scalebox{0.75}{$1^{7}2^{2}34^{5}56^{2}.2^{5}34^{2}5^{5}.3^{5}6^{4}.56$},
\scalebox{0.75}{$1^{7}23^{5}4^{2}5^{2}6.2^{6}3^{2}46^{5}.4^{4}5.5^{4}6$},
\scalebox{0.75}{$1^{7}2^{2}34^{2}5^{4}.2^{5}3^{2}.3^{4}4^{3}.4^{2}5^{3}$},
\scalebox{0.75}{$1^{7}23^{2}45.2^{6}34^{2}5^{2}6.3^{4}45^{4}.4^{3}6^{6}$},
\scalebox{0.75}{$1^{7}23^{4}4^{4}5^{3}6^{2}.2^{6}34^{3}5^{4}6^{5}.3^{2}$},
\scalebox{0.75}{$1^{7}23^{3}456^{5}.2^{6}5^{5}7^{7}.3^{4}4^{2}6.4^{4}56$},
\scalebox{0.75}{$1^{7}2345^{5}68.2^{6}5^{2}6^{6}.3^{6}7^{7}8.4^{6}8^{5}$},
\scalebox{0.75}{$1^{7}23^{2}5^{2}.2^{6}4^{3}56.3^{5}45^{3}6.4^{3}56^{5}$},
\scalebox{0.75}{$1^{7}234^{3}5^{7}.2^{6}3^{2}6^{7}7^{4}.3^{4}4^{4}7^{3}$},
\scalebox{0.75}{$1^{7}23^{2}4^{5}5^{2}.2^{6}4^{2}5^{5}6^{3}.3^{5}.6^{4}$},
\scalebox{0.75}{$1^{7}2^{3}3^{2}45.2^{4}345^{5}6^{2}.3^{4}456^{5}.4^{4}$},
\scalebox{0.75}{$1^{7}2^{4}3^{3}.2^{3}3^{4}4^{2}5^{4}6.4^{5}5^{3}.6^{6}$},
\scalebox{0.75}{$1^{7}234^{4}5^{4}67.2^{6}34^{2}5^{2}6^{6}.3^{5}457^{6}$},
\scalebox{0.75}{$1^{7}2^{4}34^{7}5^{3}6.2^{3}3^{3}5^{4}6^{6}7^{7}.3^{3}$},
\scalebox{0.75}{$1^{7}2^{7}3^{2}.3^{5}4^{3}56^{5}.4^{4}5^{2}.5^{4}6^{2}$},
\scalebox{0.75}{$1^{7}23^{2}4^{6}57.2^{6}456^{7}7^{2}.3^{5}57^{4}.5^{4}$},
\scalebox{0.75}{$1^{7}2^{2}3^{4}4^{2}5^{2}.2^{5}3^{3}4.4^{4}5^{3}.5^{2}$},
\scalebox{0.75}{$1^{7}2^{2}3^{2}4^{6}5^{3}.2^{5}3^{3}45^{2}.3^{2}.5^{2}$},
\scalebox{0.75}{$1^{7}234^{2}5^{5}.2^{6}34^{2}56^{6}.3^{5}567^{7}.4^{3}$},
\scalebox{0.75}{$1^{7}2345^{6}7.2^{6}56^{7}7.3^{6}7^{5}8^{4}.4^{6}8^{3}$},
\scalebox{0.75}{$1^{7}2^{7}3^{2}4^{2}.3^{5}4^{2}5^{3}6^{7}.4^{3}5.5^{3}$},
\scalebox{0.75}{$1^{7}2345^{6}67^{2}.2^{6}4^{2}6^{5}.3^{6}57^{5}.4^{4}6$},
\scalebox{0.75}{$1^{7}2^{2}3^{2}4^{2}5^{3}.2^{5}34^{2}.3^{4}5^{4}.4^{3}$},
\scalebox{0.75}{$1^{7}24^{4}5^{5}.2^{6}3^{2}5^{2}6^{6}.3^{5}4^{3}67^{7}$},
\scalebox{0.75}{$1^{7}2^{2}3^{2}45^{3}.2^{5}3^{2}.3^{3}4^{4}.4^{2}5^{4}$},
\scalebox{0.75}{$1^{7}2^{3}34^{2}5^{5}.2^{4}3^{3}.3^{3}4^{2}.4^{3}5^{2}$},
\scalebox{0.75}{$1^{7}2^{7}3^{2}.3^{5}4^{3}5^{7}7.4^{4}6^{5}7^{6}.6^{2}$},
\scalebox{0.75}{$1^{7}234^{3}5^{2}.2^{6}5^{5}6^{2}.3^{6}6^{3}.4^{4}6^{2}$},
\scalebox{0.75}{$1^{7}25^{5}6.2^{6}3^{2}6^{4}.3^{5}4^{3}.4^{4}5^{2}6^{2}$},
\scalebox{0.75}{$1^{7}24^{4}5^{6}7^{2}.2^{6}3^{2}56^{6}.3^{5}4^{3}67^{5}$},
\scalebox{0.75}{$1^{7}2^{3}34^{6}567.2^{4}3^{2}45^{6}6^{5}7.3^{4}6.7^{5}$},
\scalebox{0.75}{$1^{7}23^{2}4^{4}6^{2}.2^{6}4^{3}.3^{5}5^{3}6.5^{4}6^{4}$},
\scalebox{0.75}{$1^{7}24^{4}5^{2}6.2^{6}3^{2}6^{4}.3^{5}4^{3}.5^{5}6^{2}$},
\scalebox{0.75}{$1^{7}2^{7}3^{2}5.3^{5}4^{4}5^{2}6^{6}.4^{3}5^{2}6.5^{2}$},
\scalebox{0.75}{$1^{7}234^{3}5^{3}6^{3}.2^{6}345^{3}6^{3}.3^{5}4^{2}56.4$},
\scalebox{0.75}{$1^{7}2^{2}3^{2}4^{6}5^{5}6^{4}.2^{5}3.3^{4}4.5^{2}6^{3}$},
\scalebox{0.75}{$1^{7}2^{2}3^{7}57^{2}.2^{5}4^{7}6.5^{6}6^{2}7^{5}.6^{4}$},
\scalebox{0.75}{$1^{7}24^{5}5^{6}6^{2}.2^{6}3^{2}56^{5}7^{7}.3^{5}.4^{2}$},
\scalebox{0.75}{$1^{7}2^{2}345^{4}6.2^{5}34^{2}5^{2}6^{4}.3^{5}56.4^{4}6$},
\scalebox{0.75}{$1^{7}2^{3}3^{3}4^{2}6^{2}.2^{4}3^{4}4^{5}5.5^{6}6.6^{4}$},
\scalebox{0.75}{$1^{7}23^{5}4.2^{6}3^{2}4^{2}5^{2}.4^{4}5^{5}6^{2}.6^{5}$},
\scalebox{0.75}{$1^{7}24^{6}5^{3}.2^{6}3^{2}5^{2}6^{5}.3^{5}5^{2}6^{2}.4$},
\scalebox{0.75}{$1^{7}24^{6}567.2^{6}3^{2}6^{5}7^{2}.3^{5}467^{3}.5^{6}7$},
\scalebox{0.75}{$1^{7}2^{3}3^{3}4.2^{4}3^{2}5^{7}6.3^{2}4^{4}6^{6}.4^{2}$},
\scalebox{0.75}{$1^{7}234^{2}5^{3}6.2^{6}35^{4}7.3^{5}46^{5}.4^{4}67^{6}$},
\scalebox{0.75}{$1^{7}24^{3}5^{5}.2^{6}3^{2}5^{2}6^{6}.3^{5}4^{3}67^{7}.4$},
\scalebox{0.75}{$1^{7}24^{4}5^{4}6.2^{6}3^{2}5^{2}6^{4}.3^{5}4^{3}6^{2}.5$},
\scalebox{0.75}{$1^{7}23^{5}56.2^{6}3^{2}4^{2}.4^{5}5^{2}6^{3}.5^{4}6^{3}$},
\scalebox{0.75}{$1^{7}2^{2}3^{2}45.2^{5}35^{6}.3^{4}4^{2}6^{5}.4^{4}6^{2}$},
\scalebox{0.75}{$1^{7}2^{3}3^{2}4.2^{4}3^{2}45^{6}.3^{3}4^{2}56^{7}.4^{3}$},
\scalebox{0.75}{$1^{7}24^{4}5^{5}6.2^{6}3^{2}5^{2}6^{5}7.3^{5}4^{3}67^{6}$},
\scalebox{0.75}{$1^{7}24^{3}5^{4}6.2^{6}3^{2}5^{3}6^{3}.3^{5}4^{3}6^{3}.4$},
\scalebox{0.75}{$1^{7}23^{5}6^{3}.2^{6}3^{2}45^{2}.4^{6}56^{2}.5^{4}6^{2}$},
\scalebox{0.75}{$1^{7}23^{6}4^{4}5.2^{6}34^{3}5^{6}6^{2}.6^{5}7^{3}.7^{4}$},
\scalebox{0.75}{$1^{7}2^{2}3^{4}5^{3}.2^{5}3^{3}45^{2}6.4^{6}5^{2}6.6^{5}$},
\scalebox{0.75}{$1^{7}2^{2}3^{3}56^{2}.2^{5}3^{3}4.34^{6}5^{2}.5^{4}6^{5}$},
\scalebox{0.75}{$1^{7}2^{3}34^{3}5^{6}.2^{4}3^{2}6^{7}7^{7}.3^{4}45.4^{3}$},
\scalebox{0.75}{$1^{7}2^{7}35.3^{6}4^{2}5^{6}7^{2}.4^{5}6^{3}7.6^{4}7^{4}$},
\scalebox{0.75}{$1^{7}2^{2}34^{2}5.2^{5}3^{2}4^{2}.3^{4}4^{2}5^{3}.45^{3}$},
\scalebox{0.75}{$1^{7}2^{2}4^{5}6.2^{5}3^{3}45^{2}6^{2}.3^{4}46^{4}.5^{5}$},
\scalebox{0.75}{$1^{7}2^{2}34^{3}5^{2}6^{6}.2^{5}3^{2}.3^{4}4^{3}.45^{5}6$},
\scalebox{0.75}{$1^{7}23^{2}4^{5}5^{3}.2^{6}34^{2}5^{2}6^{6}.3^{4}6.5^{2}$},
\scalebox{0.75}{$1^{7}2^{2}3^{4}5^{2}.2^{5}3^{3}45^{3}6^{2}.4^{6}56^{5}.5$},
\scalebox{0.75}{$1^{7}23^{5}.2^{6}3^{2}45^{2}6^{2}.4^{6}56^{3}.5^{4}6^{2}$},
\scalebox{0.75}{$1^{7}2^{7}3.3^{6}4^{3}5^{2}7^{2}.4^{4}5^{5}67.6^{6}7^{4}$},
\scalebox{0.75}{$1^{7}2^{2}34^{2}56.2^{5}35^{5}6.3^{5}46^{5}7.4^{4}57^{6}$},
\scalebox{0.75}{$1^{7}23^{2}4^{5}5.2^{6}4^{2}6^{3}.3^{5}5^{2}6.5^{4}6^{3}$},
\scalebox{0.75}{$1^{7}2^{7}3^{2}46.3^{5}4^{3}56^{6}.4^{3}5^{3}7^{7}.5^{3}$},
\scalebox{0.75}{$1^{7}2^{7}36^{2}.3^{6}4^{2}56^{2}.4^{5}5^{2}6^{2}.5^{4}6$},
\scalebox{0.75}{$1^{7}2^{2}3^{2}4^{7}6.2^{5}5^{3}6^{5}7^{6}.3^{5}67.5^{4}$},
\scalebox{0.75}{$1^{7}2^{3}34^{2}6^{2}.2^{4}3^{2}5^{6}.3^{4}456^{5}.4^{4}$},
\scalebox{0.75}{$1^{7}2^{2}3^{7}4.2^{5}45^{2}6^{5}.4^{5}6^{2}7^{6}.5^{5}7$},
\scalebox{0.75}{$1^{7}2^{7}346.3^{6}456^{3}7^{5}.4^{5}5^{2}67.5^{4}6^{2}7$},
\scalebox{0.75}{$1^{7}2^{7}34.3^{6}5^{3}6^{4}7^{3}.4^{6}6^{3}7^{3}.5^{4}7$},
\scalebox{0.75}{$1^{7}2^{6}.23^{7}45^{3}6.4^{6}6^{4}7^{3}.5^{4}6^{2}7^{4}$},
\scalebox{0.75}{$1^{7}234^{5}.2^{6}3^{2}45^{2}6^{3}.3^{4}45^{2}6^{4}.5^{3}$},
\scalebox{0.75}{$1^{7}2^{4}3^{2}4^{3}5^{4}.2^{3}3^{2}4^{4}5^{3}6^{7}.3^{3}$},
\scalebox{0.75}{$1^{7}2345^{6}7.2^{6}456^{7}78.3^{6}7^{4}8^{2}.4^{5}78^{4}$},
\scalebox{0.75}{$1^{7}23^{4}5^{5}67.2^{6}3^{2}456^{4}7.34^{5}56^{2}7^{5}.4$},
\scalebox{0.75}{$1^{7}23^{2}45^{4}.2^{6}35^{2}6^{4}.3^{4}4^{2}5.4^{4}6^{3}$},
\scalebox{0.75}{$1^{7}2^{3}5^{5}6.2^{4}3^{4}456^{4}.3^{3}4^{3}56^{2}.4^{3}$},
\scalebox{0.75}{$1^{7}2^{2}3^{3}4^{2}6.2^{5}3^{3}45^{4}.34^{3}5^{3}6^{6}.4$},
\scalebox{0.75}{$1^{7}234^{4}5^{2}.2^{6}3^{2}45^{3}6.3^{4}5^{2}6^{6}.4^{2}$},
\scalebox{0.75}{$1^{7}2^{2}345^{2}6.2^{5}345^{3}6^{3}.3^{5}456.4^{4}56^{2}$},
\scalebox{0.75}{$1^{7}2345^{2}6^{4}.2^{6}4^{2}5^{2}.3^{6}56^{3}.4^{4}5^{2}$},
\scalebox{0.75}{$1^{7}24^{2}5^{5}6^{4}.2^{6}3^{2}56^{3}7^{7}.3^{5}4.4^{4}5$},
\scalebox{0.75}{$1^{7}2^{2}3^{5}4.2^{5}3^{2}46^{6}7.4^{5}5^{3}7^{6}.5^{4}6$},
\scalebox{0.75}{$1^{7}24^{2}5^{7}6.2^{6}3^{2}6^{4}7^{5}.3^{5}6^{2}7.4^{5}7$},
\scalebox{0.75}{$1^{7}2^{2}3^{2}46^{5}.2^{5}3^{4}5^{3}.34^{5}5^{3}6^{2}.45$},
\scalebox{0.75}{$1^{7}2345^{5}6^{3}7.2^{6}456^{2}7^{4}.3^{6}567^{2}.4^{5}6$},
\scalebox{0.75}{$1^{7}24^{6}.2^{6}3^{2}5^{6}.3^{5}46^{2}7^{6}.56^{5}78^{7}$},
\scalebox{0.75}{$1^{7}234^{5}57.2^{6}4^{2}6^{7}.3^{6}5^{2}7^{2}.5^{4}7^{4}$},
\scalebox{0.75}{$1^{7}234^{3}5^{4}67.2^{6}3^{2}6^{6}.3^{4}4^{4}7^{6}.5^{3}$},
\scalebox{0.75}{$1^{7}2^{2}3^{2}4^{5}5^{2}.2^{5}3^{2}4^{2}5^{2}.3^{3}.5^{3}$},
\scalebox{0.75}{$1^{7}2^{2}4^{3}5^{5}6.2^{5}3^{3}5^{2}6^{4}.3^{4}4^{4}6^{2}$},
\scalebox{0.75}{$1^{7}2^{2}3^{7}4.2^{5}45^{2}6^{6}8.4^{5}67^{6}.5^{5}78^{6}$},
\scalebox{0.75}{$1^{7}234^{2}5^{2}6^{4}.2^{6}45^{4}.3^{6}56^{3}7.4^{4}7^{6}$},
\scalebox{0.75}{$1^{7}2^{2}3^{4}4^{2}5^{2}.2^{5}3^{3}4^{5}.5^{5}6^{4}.6^{3}$},
\scalebox{0.75}{$1^{7}2^{7}356.3^{6}4^{2}56^{2}7.4^{5}6^{4}7^{2}.5^{5}7^{4}$},
\scalebox{0.75}{$1^{7}2346^{4}.2^{6}3^{2}45.3^{4}4^{2}5^{3}.4^{3}5^{3}6^{3}$},
\scalebox{0.75}{$1^{7}23^{7}457^{3}.2^{6}5^{5}6^{2}7^{4}.4^{6}6^{4}8^{7}.56$},
\scalebox{0.75}{$1^{7}23^{2}4^{2}5^{2}6^{5}.2^{6}34^{3}5.3^{4}4^{2}5^{4}6.6$},
\scalebox{0.75}{$1^{7}234^{4}6^{2}.2^{6}34^{2}6^{2}.3^{5}45^{2}6.5^{5}6^{2}$},
\scalebox{0.75}{$1^{7}2^{4}35^{3}.2^{3}3^{3}4^{3}5^{2}.3^{3}4^{2}5^{2}.4^{2}$},
\scalebox{0.75}{$1^{7}23^{2}45^{2}6^{4}7.2^{6}5^{5}67^{5}.3^{5}46^{2}7.4^{5}$},
\scalebox{0.75}{$1^{7}23^{4}46.2^{6}3^{2}4^{2}6.34^{3}5^{4}6^{2}.45^{3}6^{3}$},
\scalebox{0.75}{$1^{7}235^{3}6^{6}7.2^{6}4^{2}57^{4}.3^{6}5^{2}6.4^{5}57^{2}$},
\scalebox{0.75}{$1^{7}2^{2}35^{6}67.2^{5}34^{2}6^{6}.3^{5}57^{6}8.4^{5}8^{6}$},
\scalebox{0.75}{$1^{7}2^{4}3^{2}4^{3}5^{4}.2^{3}3^{2}45^{3}6^{7}.3^{3}.4^{3}$},
\scalebox{0.75}{$1^{7}2^{2}3^{2}4^{4}5^{3}.2^{5}34^{3}5^{4}6^{5}.3^{4}.6^{2}$},
\scalebox{0.75}{$1^{7}245^{6}68^{3}.2^{6}3^{2}6^{6}.3^{5}47^{7}8.4^{5}58^{3}$},
\scalebox{0.75}{$1^{7}234^{4}5^{2}6^{2}.2^{6}4^{3}6^{3}7^{5}.3^{6}67.5^{5}67$},
\scalebox{0.75}{$1^{7}234^{6}7.2^{6}35^{2}6^{3}7^{3}.3^{5}46^{3}.5^{5}67^{3}$},
\scalebox{0.75}{$1^{7}23^{6}4.2^{6}34^{2}56^{4}.4^{4}5^{2}6^{2}7^{6}.5^{4}67$},
\scalebox{0.75}{$1^{7}2^{2}3^{2}4^{3}6^{2}.2^{5}4^{4}5^{2}.3^{5}56^{5}.5^{4}$},
\scalebox{0.75}{$1^{7}2^{5}3^{3}45^{2}.2^{2}3^{2}4^{2}5^{5}.3^{2}4^{2}.4^{2}$},
\scalebox{0.75}{$1^{7}2^{5}4^{6}5^{2}6^{3}7^{2}.2^{2}3^{6}45^{5}6^{4}7^{5}.3$},
\scalebox{0.75}{$1^{7}23^{2}4^{5}.2^{6}3^{2}4^{2}56^{4}.3^{3}5^{3}6^{3}.5^{3}$},
\scalebox{0.75}{$1^{7}23^{4}56^{6}.2^{6}3^{2}45^{3}7^{6}.34^{5}5^{3}678^{7}.4$},
\scalebox{0.75}{$1^{7}2^{4}56.2^{3}3^{5}4^{2}567.3^{2}4^{5}6^{4}7.5^{5}67^{5}$},
\scalebox{0.75}{$1^{7}24^{4}5^{2}6.2^{6}3^{2}5^{2}6^{4}.3^{5}5^{3}6^{2}.4^{3}$},
\scalebox{0.75}{$1^{7}23^{3}45^{3}6^{2}.2^{6}5^{4}6^{3}.3^{4}4^{2}.4^{4}6^{2}$},
\scalebox{0.75}{$1^{7}2345^{5}6^{2}7.2^{6}346^{4}7^{3}.3^{5}5^{2}6.4^{5}7^{3}$},
\scalebox{0.75}{$1^{7}2^{5}4^{3}5^{2}6^{4}.2^{2}3^{6}4^{3}5^{3}6^{3}.34.5^{2}$},
\scalebox{0.75}{$1^{7}2^{7}36^{2}7.3^{6}4^{2}5^{2}7^{2}.4^{5}6^{5}.5^{5}7^{4}$},
\scalebox{0.75}{$1^{7}2345^{3}6.2^{6}45^{2}6^{4}.3^{6}56^{2}7^{4}.4^{5}57^{3}$},
\scalebox{0.75}{$1^{7}234^{2}5^{4}6.2^{6}45^{3}67^{5}.3^{6}6^{3}7.4^{4}6^{2}7$},
\scalebox{0.75}{$1^{7}24^{4}5^{4}6^{2}.2^{6}3^{2}5^{2}6^{5}.3^{5}4^{3}7^{7}.5$},
\scalebox{0.75}{$1^{7}245^{6}.2^{6}3^{2}6^{5}7.3^{5}46^{2}7^{5}8.4^{5}578^{6}$},
\scalebox{0.75}{$1^{7}2^{4}3^{6}4^{3}5^{3}6^{5}.2^{3}34.4^{3}5^{2}.5^{2}6^{2}$},
\scalebox{0.75}{$1^{7}2345^{5}.2^{6}35^{2}6^{3}7.3^{5}46^{3}7^{3}.4^{5}67^{3}$},
\scalebox{0.75}{$1^{7}24^{2}5^{3}.2^{6}3^{2}6^{5}.3^{5}45^{2}.4^{4}5^{2}6^{2}$},
\scalebox{0.75}{$1^{7}23^{5}5^{3}6^{5}.2^{6}3^{2}4^{3}56^{2}7^{7}.4^{4}.5^{3}$},
\scalebox{0.75}{$1^{7}2^{2}3^{2}4^{3}.2^{5}34^{4}5.3^{4}5^{2}6^{5}.5^{4}6^{2}$},
\scalebox{0.75}{$1^{7}234^{5}6^{3}.2^{6}3^{2}45^{2}6^{3}7.3^{4}4567^{6}.5^{4}$},
\scalebox{0.75}{$1^{7}2345^{5}7^{3}.2^{6}5^{2}6^{7}.3^{6}7^{4}8^{4}.4^{6}8^{3}$},
\scalebox{0.75}{$1^{7}234^{3}5^{3}6^{2}.2^{6}4^{3}5^{2}6^{2}.3^{6}5^{2}6^{3}.4$},
\scalebox{0.75}{$1^{7}24^{7}567.2^{6}3^{2}5^{3}6^{3}7^{3}.3^{5}5^{3}6^{3}7^{3}$},
\scalebox{0.75}{$1^{7}23^{3}45^{2}6^{3}.2^{6}4^{2}5^{3}.3^{4}45^{2}6^{4}.4^{3}$},
\scalebox{0.75}{$1^{7}234^{2}5^{5}67.2^{6}36^{6}7^{2}.3^{5}457^{2}.4^{4}57^{2}$},
\scalebox{0.75}{$1^{7}23^{5}6^{7}.2^{6}3^{2}45^{4}7^{5}.4^{6}5^{2}7^{2}8^{7}.5$},
\scalebox{0.75}{$1^{7}234^{4}5^{2}6^{4}.2^{6}345^{4}67^{2}.3^{5}46^{2}7^{5}.45$},
\scalebox{0.75}{$1^{7}2^{3}3.2^{4}3^{2}4^{3}5^{2}.3^{4}45^{4}6^{2}.4^{3}56^{5}$},
\scalebox{0.75}{$1^{7}24^{5}5^{3}7.2^{6}3^{2}6^{6}7^{2}.3^{5}4^{2}67^{4}.5^{4}$},
\scalebox{0.75}{$1^{7}24^{2}5^{6}.2^{6}3^{2}6^{7}7.3^{5}57^{6}8^{2}.4^{5}8^{5}$},
\scalebox{0.75}{$1^{7}2^{7}35.3^{6}4^{2}56^{4}7^{2}.4^{5}6^{3}7^{3}.5^{5}7^{2}$},
\scalebox{0.75}{$1^{7}2345^{3}6^{3}7^{2}.2^{6}35^{2}6^{3}.3^{5}457^{5}.4^{5}56$},
\scalebox{0.75}{$1^{7}2^{2}3^{5}5.2^{5}3^{2}4^{7}7.5^{6}6^{2}7^{6}8.6^{5}8^{6}$},
\scalebox{0.75}{$1^{7}2^{2}3^{5}5^{3}6.2^{5}3^{2}46^{5}7^{5}.4^{6}567^{2}.5^{3}$},
\scalebox{0.75}{$1^{7}2345^{2}67^{5}.2^{6}345^{2}6^{2}7.3^{5}5^{3}67.4^{5}6^{3}$},
\scalebox{0.75}{$1^{7}2345^{4}6^{4}.2^{6}3567^{6}.3^{5}4^{2}68^{7}.4^{4}5^{2}67$},
\scalebox{0.75}{$1^{7}23^{2}5^{3}.2^{6}4^{2}5^{2}6^{2}.3^{5}456^{3}.4^{4}56^{2}$},
\scalebox{0.75}{$1^{7}2^{2}345^{6}.2^{5}34^{2}6^{5}7^{4}.3^{5}56^{2}7^{3}.4^{4}$},
\scalebox{0.75}{$1^{7}2^{2}3^{5}46^{3}.2^{5}3^{2}456^{4}.4^{5}5^{2}7^{6}.5^{4}7$},
\scalebox{0.75}{$1^{7}234^{4}6^{7}8.2^{6}4^{2}5^{5}7^{4}.3^{6}5^{2}7^{3}8^{6}.4$},
\scalebox{0.75}{$1^{7}2^{2}3^{2}4^{4}56^{3}7.2^{5}4^{3}56^{4}7^{6}.3^{5}5.5^{4}$},
\scalebox{0.75}{$1^{7}23^{2}4^{4}56^{3}.2^{6}4^{3}56^{2}7^{6}.3^{5}6^{2}7.5^{5}$},
\scalebox{0.75}{$1^{7}23^{2}4^{4}6^{6}.2^{6}34^{2}5^{3}7^{6}.3^{4}45^{4}678^{7}$},
\scalebox{0.75}{$1^{7}2^{2}45^{6}6.2^{5}3^{3}46^{6}7^{2}.3^{4}4^{2}57^{5}.4^{3}$},
\scalebox{0.75}{$1^{7}24^{3}5^{6}7.2^{6}3^{2}56^{3}7^{5}.3^{5}6^{4}78^{7}.4^{4}$},
\scalebox{0.75}{$1^{7}2^{2}3^{5}5^{4}6.2^{5}3^{2}45^{3}7^{2}.4^{6}67^{4}.6^{5}7$},
\scalebox{0.75}{$1^{7}24^{3}5^{4}6^{2}.2^{6}3^{2}5^{3}6^{4}7.3^{5}4^{3}67^{6}.4$},
\scalebox{0.75}{$1^{7}234^{3}56^{4}.2^{6}5^{6}7^{5}.3^{6}6^{3}7^{2}8^{6}.4^{4}8$},
\scalebox{0.75}{$1^{7}24^{4}5^{4}6^{2}.2^{6}3^{2}5^{3}6^{4}7.3^{5}4^{2}67^{6}.4$},
\scalebox{0.75}{$1^{7}2^{2}3^{6}45^{2}6^{4}7.2^{5}34^{6}5^{3}67^{6}.5^{2}.6^{2}$},
\scalebox{0.75}{$1^{7}234^{2}5^{2}67^{2}.2^{6}5^{4}6^{2}7.3^{6}6^{4}.4^{5}57^{4}$},
\scalebox{0.75}{$1^{7}23^{2}4^{3}56^{2}.2^{6}3^{2}45^{4}6.3^{3}45^{2}6^{4}.4^{2}$},
\scalebox{0.75}{$1^{7}24^{5}5^{5}6.2^{6}3^{2}56^{2}7^{3}.3^{5}4^{2}57^{3}.6^{4}7$},
\scalebox{0.75}{$1^{7}2^{2}4^{5}5^{3}6.2^{5}3^{3}45^{4}67^{3}.3^{4}6^{4}7^{4}.46$},
\scalebox{0.75}{$1^{7}23^{7}4^{6}.2^{6}45^{3}6^{2}78^{6}.5^{4}67^{6}89^{7}.6^{4}$},
\scalebox{0.75}{$1^{7}23^{2}45^{2}6^{3}.2^{6}35^{3}6^{2}.3^{4}4^{2}5.4^{4}56^{2}$},
\scalebox{0.75}{$1^{7}2345^{6}.2^{6}3^{2}6^{6}.3^{4}4^{2}567^{5}.4^{4}7^{2}8^{7}$},
\scalebox{0.75}{$1^{7}245^{5}6^{2}7^{3}.2^{6}3^{2}6^{4}7^{3}.3^{5}45^{2}7.4^{5}6$},
\scalebox{0.75}{$1^{7}2^{2}34^{2}5^{7}8.2^{5}34^{3}6^{7}8^{3}.3^{5}47^{7}8^{3}.4$},
\scalebox{0.75}{$1^{7}2^{2}3456^{2}.2^{5}3^{4}5^{2}67.3^{2}4^{5}6^{4}.45^{4}7^{6}$},
\scalebox{0.75}{$1^{7}23^{2}45^{4}.2^{6}45^{3}6^{3}7.3^{5}46^{3}7^{4}.4^{4}67^{2}$},
\scalebox{0.75}{$1^{7}2^{2}345^{2}7^{3}.2^{5}345^{4}.3^{5}456^{4}.4^{4}6^{3}7^{4}$},
\scalebox{0.75}{$1^{7}2^{3}3^{4}4^{3}5^{2}6^{3}.2^{4}45^{5}6^{4}7^{7}.3^{3}.4^{3}$},
\scalebox{0.75}{$1^{7}23^{7}47.2^{6}45^{2}6^{2}7^{5}.4^{5}56^{3}8^{7}.5^{4}6^{2}7$},
\scalebox{0.75}{$1^{7}234^{4}67^{3}.2^{6}4^{3}567.3^{6}5^{2}6^{2}.5^{4}6^{3}7^{3}$},
\scalebox{0.75}{$1^{7}24^{6}56^{2}.2^{6}3^{2}6^{3}7^{2}.3^{5}4567^{3}.5^{5}67^{2}$},
\scalebox{0.75}{$1^{7}2^{2}3^{2}5^{5}6.2^{5}34^{2}56^{5}7.3^{4}4^{2}567^{6}.4^{3}$},
\scalebox{0.75}{$1^{7}2^{3}35^{7}.2^{4}3^{3}4^{3}6^{5}.3^{3}4^{2}6^{2}7^{7}.4^{2}$},
\scalebox{0.75}{$1^{7}24^{2}5^{6}.2^{6}3^{2}6^{5}7.3^{5}567^{5}8^{2}.4^{5}678^{5}$},
\scalebox{0.75}{$1^{7}23^{3}45^{2}6^{5}.2^{6}4^{4}567^{6}8.3^{4}4^{2}5^{4}678^{6}$},
\scalebox{0.75}{$1^{7}234^{6}6^{4}.2^{6}456^{2}7^{5}8^{2}.3^{6}567^{2}.5^{5}8^{5}$},
\scalebox{0.75}{$1^{7}23^{2}45^{4}7^{2}.2^{6}45^{2}6^{5}.3^{5}456^{2}7.4^{4}7^{4}$},
\scalebox{0.75}{$1^{7}2^{3}3^{2}4^{4}6^{5}.2^{4}34^{3}5^{3}6^{2}7^{7}.3^{4}.5^{4}$},
\scalebox{0.75}{$1^{7}2^{2}3^{3}4^{2}5^{2}.2^{5}34^{5}6^{5}.3^{3}5^{2}6^{2}.5^{3}$},
\scalebox{0.75}{$1^{7}2^{2}3^{5}4^{3}.2^{5}3^{2}4^{4}5.5^{6}6^{2}7^{4}.6^{5}7^{3}$},
\scalebox{0.75}{$1^{7}2^{3}4^{5}.2^{4}3^{5}56^{5}.3^{2}4^{2}5^{2}6^{2}7^{7}.5^{4}$},
\scalebox{0.75}{$1^{7}23456^{6}.2^{6}3^{2}5^{2}7^{6}.3^{4}4^{2}5^{2}.4^{4}5^{2}67$},
\scalebox{0.75}{$1^{7}24^{2}5^{6}78.2^{6}3^{2}6^{7}.3^{5}47^{5}8^{2}.4^{4}578^{4}$},
\scalebox{0.75}{$1^{7}2^{2}34^{3}5^{3}6^{2}.2^{5}34^{3}5^{2}6^{2}.3^{5}45^{2}6^{3}$},
\scalebox{0.75}{$1^{7}25^{5}6^{5}.2^{6}3^{2}6^{2}7^{4}.3^{5}4^{3}7^{2}.4^{4}5^{2}7$},
\scalebox{0.75}{$1^{7}2^{2}3^{2}45^{3}6^{3}.2^{5}35^{4}6^{4}.3^{4}4^{2}7^{7}.4^{4}$},
\scalebox{0.75}{$1^{7}2^{7}36^{2}7.3^{6}4^{2}5^{2}6^{3}7^{4}.4^{5}6^{2}.5^{5}7^{2}$},
\scalebox{0.75}{$1^{7}24^{2}5^{5}7^{2}.2^{6}3^{2}5^{2}6^{3}7.3^{5}6^{4}.4^{5}7^{4}$},
\scalebox{0.75}{$1^{7}23^{3}5^{3}6^{3}.2^{6}4^{2}5^{2}6^{2}.3^{4}45^{2}6^{2}.4^{4}$},
\scalebox{0.75}{$1^{7}2^{2}3^{2}4^{4}5^{3}6.2^{5}3^{2}4^{3}5^{4}6^{3}.3^{3}6.6^{2}$},
\scalebox{0.75}{$1^{7}2^{2}3^{4}46^{2}.2^{5}4^{2}5^{5}6.3^{3}4^{2}5^{2}6^{4}.4^{2}$},
\scalebox{0.75}{$1^{7}23^{2}4^{6}8^{2}.2^{6}5^{7}.3^{5}46^{2}7^{5}.6^{5}7^{2}8^{5}$},
\scalebox{0.75}{$1^{7}2^{2}3^{2}4^{5}.2^{5}4^{2}5^{2}6^{5}7.3^{5}56^{2}7^{6}.5^{4}$},
\scalebox{0.75}{$1^{7}24^{2}5^{6}67^{2}.2^{6}3^{2}6^{5}78^{5}.3^{5}57^{4}8.4^{5}68$},
\scalebox{0.75}{$1^{7}2^{4}3.2^{3}3^{3}4^{2}5^{3}.3^{3}4^{2}5^{3}6^{3}.4^{3}56^{4}$},
\scalebox{0.75}{$1^{7}2^{2}3^{5}4^{2}6.2^{5}3^{2}4^{2}56^{4}.4^{3}5^{3}6^{2}.5^{3}$},
\scalebox{0.75}{$1^{7}24^{2}5^{7}6^{2}.2^{6}3^{2}6^{4}7^{6}.3^{5}4^{2}678^{7}.4^{3}$},
\scalebox{0.75}{$1^{7}23^{2}4^{2}5^{5}.2^{6}5^{2}6^{4}7^{5}.3^{5}46^{2}7^{2}.4^{4}6$},
\scalebox{0.75}{$1^{7}245^{6}6^{2}.2^{6}3^{2}6^{5}7^{2}.3^{5}4^{2}57^{3}.4^{4}7^{2}$},
\scalebox{0.75}{$1^{7}2^{2}35^{2}6.2^{5}3^{3}46^{2}.3^{3}4^{3}5^{3}.4^{3}5^{2}6^{4}$},
\scalebox{0.75}{$1^{7}234^{5}6^{2}7^{6}.2^{6}34^{2}5^{4}78^{7}.3^{5}5^{3}6^{5}9^{7}$},
\scalebox{0.75}{$1^{7}2^{2}45^{2}6^{4}.2^{5}3^{3}56.3^{4}4^{2}5^{2}.4^{4}5^{2}6^{2}$},
\scalebox{0.75}{$1^{7}2^{7}34^{4}.3^{6}4^{3}56^{2}7^{5}8.5^{6}6^{3}7^{2}8^{6}.6^{2}$},
\scalebox{0.75}{$1^{7}234^{3}5^{2}6.2^{6}5^{3}6^{4}7^{2}.3^{6}6^{2}7^{5}.4^{4}5^{2}$},
\scalebox{0.75}{$1^{7}2345^{2}6^{3}7.2^{6}4^{2}5^{2}67.3^{6}56^{3}7.4^{4}5^{2}7^{4}$},
\scalebox{0.75}{$1^{7}2^{2}4^{4}5^{3}6.2^{5}3^{4}5^{4}6^{3}7.3^{3}4^{3}6^{2}7^{6}.6$},
\scalebox{0.75}{$1^{7}2^{7}3.3^{6}4^{3}56^{3}7^{2}.4^{4}5^{3}67^{5}.5^{3}6^{3}8^{7}$},
\scalebox{0.75}{$1^{7}2345^{3}6^{2}.2^{6}345^{2}7^{4}.3^{5}456^{3}.4^{4}56^{2}7^{3}$},
\scalebox{0.75}{$1^{7}2^{2}345^{5}6.2^{5}3^{2}4^{3}6^{5}7^{2}.3^{4}4^{3}5^{2}67^{5}$},
\scalebox{0.75}{$1^{7}24^{2}5^{2}6^{2}.2^{6}3^{2}5^{2}.3^{5}45^{2}6^{2}.4^{4}56^{3}$},
\scalebox{0.75}{$1^{7}23^{2}4^{4}6^{2}7^{2}.2^{6}4^{2}6^{5}.3^{5}45^{2}7^{4}.5^{5}7$},
\scalebox{0.75}{$1^{7}2345^{4}6^{2}7.2^{6}3^{2}56^{4}7.3^{4}4^{2}5^{2}67^{4}.4^{4}7$},
\scalebox{0.75}{$1^{7}234^{5}5^{2}7.2^{6}4^{2}5^{3}6^{4}.3^{6}5^{2}6^{3}7^{4}.7^{2}$},
\scalebox{0.75}{$1^{7}2345^{2}.2^{6}3^{2}5^{2}6^{2}.3^{4}4^{2}5^{2}6^{2}.4^{4}56^{3}$},
\scalebox{0.75}{$1^{7}245^{6}6^{3}.2^{6}3^{2}6^{2}7^{4}.3^{5}4^{2}57.4^{4}6^{2}7^{2}$},
\scalebox{0.75}{$1^{7}23^{3}4^{2}56^{5}.2^{6}35^{5}6^{2}7^{5}.3^{3}4^{3}57^{2}.4^{2}$},
\scalebox{0.75}{$1^{7}23^{7}6.2^{6}4^{3}57^{6}.4^{4}5^{2}6^{4}8^{5}.5^{4}6^{2}78^{2}$},
\scalebox{0.75}{$1^{7}2^{7}36.3^{6}4^{2}5^{2}6^{2}7^{2}.4^{5}56^{3}7^{3}.5^{4}67^{2}$},
\scalebox{0.75}{$1^{7}24^{6}6^{3}7.2^{6}3^{2}5^{4}6^{2}7^{3}.3^{5}45^{2}6^{2}7^{3}.5$},
\scalebox{0.75}{$1^{7}23^{6}4.2^{6}35^{2}6^{2}7^{3}.4^{6}56^{2}7^{2}.5^{4}6^{3}7^{2}$},
\scalebox{0.75}{$1^{7}2345^{3}6^{2}7^{3}.2^{6}4^{2}5^{2}6.3^{6}6^{4}.4^{4}5^{2}7^{4}$},
\scalebox{0.75}{$1^{7}234^{2}5^{3}6^{3}.2^{6}35^{3}6^{4}7^{3}.3^{5}457^{2}.4^{4}7^{2}$},
\scalebox{0.75}{$1^{7}23^{2}45^{3}6^{3}7.2^{6}45^{2}6^{3}7^{3}.3^{5}5^{2}67^{3}.4^{5}$},
\scalebox{0.75}{$1^{7}24^{2}5^{5}7.2^{6}3^{2}5^{2}67^{3}.3^{5}46^{4}7.4^{4}6^{2}7^{2}$},
\scalebox{0.75}{$1^{7}24^{2}5^{6}6.2^{6}3^{2}56^{3}7^{5}.3^{5}467^{2}8^{7}.4^{4}6^{2}$},
\scalebox{0.75}{$1^{7}245^{3}6^{5}78.2^{6}3^{2}57^{5}8^{4}.3^{5}45^{2}68^{2}.4^{5}567$},
\scalebox{0.75}{$1^{7}2^{2}34^{6}.2^{5}35^{2}6^{5}7^{2}.3^{5}467^{5}8^{3}.5^{5}68^{4}$},
\scalebox{0.75}{$1^{7}23456^{2}.2^{6}3^{2}5^{2}6.3^{4}4^{2}5^{2}6^{2}.4^{4}5^{2}6^{2}$},
\scalebox{0.75}{$1^{7}2345^{4}7^{3}.2^{6}3^{2}5^{2}6^{4}.3^{4}4^{2}56^{3}7.4^{4}7^{3}$},
\scalebox{0.75}{$1^{7}23^{2}4^{3}5^{2}6^{2}.2^{6}345^{2}6^{5}7.3^{4}45^{3}7^{6}.4^{2}$},
\scalebox{0.75}{$1^{7}234^{2}5^{4}68.2^{6}35^{3}67^{5}.3^{5}46^{3}78^{5}.4^{4}6^{2}78$},
\scalebox{0.75}{$1^{7}2345^{3}6.2^{6}4^{2}5^{2}6^{2}7.3^{6}6^{4}7^{2}.4^{4}5^{2}7^{4}$},
\scalebox{0.75}{$1^{7}24^{2}5^{3}6^{4}7^{3}.2^{6}3^{2}5^{2}67^{4}.3^{5}45.4^{4}56^{2}$},
\scalebox{0.75}{$1^{7}234^{2}6^{4}7^{4}8.2^{6}3^{2}457^{3}8^{6}.3^{4}4^{4}6.5^{6}6^{2}$},
\scalebox{0.75}{$1^{7}24^{5}6^{2}7^{6}.2^{6}3^{2}5^{5}78^{6}9.3^{5}45^{2}6^{5}89^{6}.4$},
\scalebox{0.75}{$1^{7}235^{6}8^{2}.2^{6}34^{2}6^{4}.3^{5}456^{2}7^{4}.4^{4}67^{3}8^{5}$},
\scalebox{0.75}{$1^{7}24^{2}5^{3}6^{2}.2^{6}3^{2}56^{4}7^{2}.3^{5}45^{2}67.4^{4}57^{4}$},
\scalebox{0.75}{$1^{7}2^{2}34^{2}5^{5}6^{2}.2^{5}3^{2}4^{3}6^{3}.3^{4}4^{2}5^{2}.6^{2}$},
\scalebox{0.75}{$1^{7}2^{7}3^{2}.3^{5}4^{3}56^{5}7^{2}.4^{4}5^{2}6^{2}7^{3}.5^{4}7^{2}$},
\scalebox{0.75}{$1^{7}235^{6}6^{3}.2^{6}4^{2}6^{3}7^{4}8.3^{6}567^{2}8^{3}.4^{5}78^{3}$},
\scalebox{0.75}{$1^{7}234^{3}5^{2}6^{5}.2^{6}4^{3}5^{3}67^{5}8.3^{6}5^{2}67^{2}8^{6}.4$},
\scalebox{0.75}{$1^{7}235^{6}6^{2}8.2^{6}4^{2}6^{3}7^{3}8.3^{6}57^{4}8.4^{5}6^{2}8^{4}$},
\scalebox{0.75}{$1^{7}2^{2}45^{5}7.2^{5}3^{3}6^{5}.3^{4}4^{2}6^{2}7^{4}.4^{4}5^{2}7^{2}$},
\scalebox{0.75}{$1^{7}23^{2}4^{7}5^{4}.2^{6}5^{2}6^{5}7^{2}8^{6}.3^{5}56^{2}7^{5}89^{7}$},
\scalebox{0.75}{$1^{7}2^{2}3^{2}45^{4}6^{3}.2^{5}35^{3}6^{4}7^{5}.3^{4}4^{2}.4^{4}7^{2}$},
\scalebox{0.75}{$1^{7}235^{6}6^{2}8.2^{6}4^{2}6^{4}7^{2}8^{3}.3^{6}67^{4}8^{2}.4^{5}578$},
\scalebox{0.75}{$1^{7}24^{5}6^{3}7^{5}9.2^{6}3^{2}5^{5}78^{6}.3^{5}45^{2}6^{4}789^{6}.4$},
\scalebox{0.75}{$1^{7}23^{2}5^{4}6^{2}7^{2}.2^{6}4^{2}56^{3}7^{3}.3^{5}456^{2}7.4^{4}57$},
\scalebox{0.75}{$1^{7}235^{6}7^{3}.2^{6}34^{2}6^{3}78^{5}.3^{5}56^{3}7.4^{5}67^{2}8^{2}$},
\scalebox{0.75}{$1^{7}23^{3}45^{2}6^{5}.2^{6}4^{2}5^{2}7^{7}.3^{4}5^{3}6^{2}8^{7}.4^{4}$},
\scalebox{0.75}{$1^{7}23^{2}5^{5}6^{2}.2^{6}4^{2}5^{2}6^{2}7.3^{5}6^{3}7^{4}.4^{5}7^{2}$},
\scalebox{0.75}{$1^{7}235^{4}6^{3}7.2^{6}34^{2}67^{5}8^{2}.3^{5}45^{3}8^{5}.4^{4}6^{3}7$},
\scalebox{0.75}{$1^{7}2345^{5}67^{3}.2^{6}4^{2}56^{3}7.3^{6}567^{3}8^{2}.4^{4}6^{2}8^{5}$},
\scalebox{0.75}{$1^{7}23^{2}45^{2}6^{2}.2^{6}45^{3}6^{3}.3^{5}5^{2}6^{2}7^{4}.4^{5}7^{3}$},
\scalebox{0.75}{$1^{7}24^{2}5^{4}6^{2}7^{3}.2^{6}3^{2}56^{3}7.3^{5}5^{2}7^{3}.4^{5}6^{2}$},
\scalebox{0.75}{$1^{7}23^{2}4^{4}56^{4}8^{3}.2^{6}34^{3}56^{3}7^{3}.3^{4}5^{5}7^{4}8^{4}$},
\scalebox{0.75}{$1^{7}2^{2}5^{4}6^{2}7.2^{5}3^{3}46^{4}7^{2}.3^{4}4^{2}5^{2}67^{4}.4^{4}5$},
\scalebox{0.75}{$1^{7}2^{2}3^{6}6^{2}8.2^{5}34^{2}6^{5}7^{3}8.4^{5}5^{3}7^{3}8^{5}.5^{4}7$},
\scalebox{0.75}{$1^{7}24^{2}5^{6}6^{2}7^{2}.2^{6}3^{2}6^{5}8^{2}.3^{5}57^{4}8.4^{5}78^{4}$},
\scalebox{0.75}{$1^{7}234^{6}5.2^{6}3^{2}46^{2}7^{2}.3^{4}5^{2}6^{3}7^{2}.5^{4}6^{2}7^{3}$},
\scalebox{0.75}{$1^{7}24^{4}5^{3}67^{5}9.2^{6}3^{2}5^{4}678^{6}9.3^{5}4^{2}6^{5}789^{5}.4$},
\scalebox{0.75}{$1^{7}2^{2}3^{7}67.2^{5}4^{3}56^{2}7.4^{4}5^{2}6^{2}7^{3}.5^{4}6^{2}7^{2}$},
\scalebox{0.75}{$1^{7}23^{7}.2^{6}4^{4}5^{2}6^{2}7.4^{3}5^{2}6^{3}7^{5}8.5^{3}6^{2}78^{6}$},
\scalebox{0.75}{$1^{7}2345^{4}6^{3}.2^{6}3456^{2}7^{4}8.3^{5}5^{2}6^{2}7^{2}8.4^{5}78^{5}$},
\scalebox{0.75}{$1^{7}24^{2}5^{4}6^{5}7^{2}.2^{6}3^{2}5^{2}7^{5}8^{6}.3^{5}45.4^{4}6^{2}8$},
\scalebox{0.75}{$1^{7}23^{2}45^{3}.2^{6}35^{4}6^{2}7.3^{4}4^{2}6^{3}7^{3}.4^{4}6^{2}7^{3}$},
\scalebox{0.75}{$1^{7}245^{2}6^{5}.2^{6}3^{2}567^{6}.3^{5}4^{2}578^{6}9.4^{4}5^{3}689^{6}$},
\scalebox{0.75}{$1^{7}234^{7}6^{2}7.2^{6}5^{2}6^{3}7^{2}8^{6}.3^{6}6^{2}7^{2}.5^{5}7^{2}8$},
\scalebox{0.75}{$1^{7}24^{3}5^{6}78^{5}.2^{6}3^{2}56^{6}789^{6}.3^{5}4^{3}67^{5}89A^{7}.4$},
\scalebox{0.75}{$1^{7}2^{3}345^{4}6.2^{4}3^{2}45^{3}6^{3}7^{3}.3^{4}4^{2}6^{3}7^{4}.4^{3}$},
\scalebox{0.75}{$1^{7}24^{2}5^{3}6^{5}7.2^{6}3^{2}6^{2}7^{5}.3^{5}45^{2}8^{7}.4^{4}5^{2}7$},
\scalebox{0.75}{$1^{7}2^{2}345^{2}6^{4}.2^{5}345^{3}7^{7}.3^{5}5^{2}6^{3}8^{5}.4^{5}8^{2}$},
\scalebox{0.75}{$1^{7}2^{2}34^{4}5^{2}.2^{5}345^{5}7^{3}.3^{5}46^{2}7^{4}8^{2}.46^{5}8^{5}$},
\scalebox{0.75}{$1^{7}2345^{2}6^{2}.2^{6}4^{2}5^{2}6^{2}7^{2}.3^{6}56^{3}7.4^{4}5^{2}7^{4}$},
\scalebox{0.75}{$1^{7}235^{4}67^{2}.2^{6}4^{2}6^{4}78^{3}.3^{6}5^{2}7^{4}.4^{5}56^{2}8^{4}$},
\scalebox{0.75}{$1^{7}23^{5}4^{6}56^{3}7^{3}8^{7}9.2^{6}3^{2}5^{5}6^{4}7^{4}9^{6}A^{7}.4.5$},
\scalebox{0.75}{$1^{7}2^{2}345^{2}6.2^{5}3^{2}45^{2}6^{2}.3^{4}4^{2}56^{2}.4^{3}5^{2}6^{2}$},
\scalebox{0.75}{$1^{7}2^{2}35^{4}67^{4}.2^{5}34^{2}56^{4}8^{6}.3^{5}5^{2}6^{2}.4^{5}7^{3}8$},
\scalebox{0.75}{$1^{7}24^{2}5^{5}67.2^{6}3^{2}5^{2}6^{4}78^{2}.3^{5}6^{2}7^{5}8.4^{5}8^{4}$},
\scalebox{0.75}{$1^{7}23^{2}45^{3}67^{3}.2^{6}35^{3}67^{2}.3^{4}4^{2}56^{2}.4^{4}6^{3}7^{2}$},
\scalebox{0.75}{$1^{7}235^{4}6^{2}7^{3}.2^{6}3456^{4}78^{4}.3^{5}4^{2}57^{2}8.4^{4}5678^{2}$},
\scalebox{0.75}{$1^{7}2^{2}35^{6}67.2^{5}3^{2}46^{4}8^{2}.3^{4}4^{2}57^{5}.4^{4}6^{2}78^{5}$},
\scalebox{0.75}{$1^{7}23^{2}4^{2}5^{3}6^{3}8.2^{6}35^{3}67^{6}.3^{4}4^{2}56^{3}78^{6}.4^{3}$},
\scalebox{0.75}{$1^{7}23^{2}4^{2}5^{3}67^{3}.2^{6}34^{2}5^{2}6^{4}.3^{4}4^{3}5^{2}6^{2}7^{4}$},
\scalebox{0.75}{$1^{7}245^{6}7^{3}.2^{6}3^{2}6^{4}78^{4}.3^{5}4^{2}67^{2}8^{2}.4^{4}56^{2}78$},
\scalebox{0.75}{$1^{7}24^{4}5^{4}67^{3}.2^{6}3^{2}5^{3}6^{3}7^{3}8^{2}.3^{5}4^{3}6^{3}78^{5}$},
\scalebox{0.75}{$1^{7}2^{2}35^{3}6^{4}.2^{5}3^{2}46^{3}7^{5}.3^{4}4^{2}5^{2}7^{2}.4^{4}5^{2}$},
\scalebox{0.75}{$1^{7}23^{4}456^{2}7^{3}.2^{6}3456^{5}8.3^{2}4^{3}5^{2}7^{4}.4^{2}5^{3}8^{6}$},
\scalebox{0.75}{$1^{7}245^{2}6^{6}89^{2}.2^{6}3^{2}57^{7}8.3^{5}4^{2}58^{5}.4^{4}5^{3}69^{5}$},
\scalebox{0.75}{$1^{7}2345^{6}6^{2}7.2^{6}456^{3}7^{3}8^{5}.3^{6}6^{2}7^{2}8^{2}9^{7}.4^{5}7$},
\scalebox{0.75}{$1^{7}24^{2}5^{3}6^{3}.2^{6}3^{2}5^{2}6^{2}.3^{5}45^{2}6^{2}7^{2}.4^{4}7^{5}$},
\scalebox{0.75}{$1^{7}24^{2}5^{5}678^{2}.2^{6}3^{2}5^{2}6^{3}.3^{5}6^{3}7^{4}.4^{5}7^{2}8^{5}$},
\scalebox{0.75}{$1^{7}23^{2}45^{2}6^{3}7^{3}.2^{6}45^{3}6^{2}7^{2}.3^{5}5^{2}7^{2}.4^{5}6^{2}$},
\scalebox{0.75}{$1^{7}245^{3}6^{4}78.2^{6}3^{2}567^{6}.3^{5}4^{2}6^{2}8^{5}9.4^{4}5^{3}89^{6}$},
\scalebox{0.75}{$1^{7}2^{7}3^{7}4A.4^{6}5^{3}7^{6}8^{5}9^{2}.5^{4}6^{4}78^{2}9^{5}A^{6}.6^{3}$},
\scalebox{0.75}{$1^{7}245^{3}6^{6}78.2^{6}3^{2}567^{5}.3^{5}4^{2}78^{5}9^{2}.4^{4}5^{3}89^{5}$},
\scalebox{0.75}{$1^{7}235^{3}6^{3}7^{2}.2^{6}4^{2}567^{5}9.3^{6}5^{2}68^{6}.4^{5}56^{2}89^{6}$},
\scalebox{0.75}{$1^{7}2^{3}35^{3}6^{7}.2^{4}3^{2}4^{4}7^{6}8^{4}.3^{4}45^{4}78^{3}9^{7}.4^{2}$},
\scalebox{0.75}{$1^{7}24^{2}5^{7}68^{2}.2^{6}3^{2}6^{3}7^{4}.3^{5}6^{3}7^{2}8^{3}.4^{5}78^{2}$},
\scalebox{0.75}{$1^{7}24^{5}5^{3}67^{2}.2^{6}3^{2}56^{6}7^{2}8^{2}.3^{5}4^{2}57^{3}8^{5}.5^{2}$},
\scalebox{0.75}{$1^{7}25^{6}6^{2}8.2^{6}3^{2}6^{4}7^{3}.3^{5}4^{3}7^{2}8^{4}.4^{4}567^{2}8^{2}$},
\scalebox{0.75}{$1^{7}2^{2}3^{2}4^{5}5^{3}6^{3}7^{5}.2^{5}3^{3}4^{2}5^{4}6^{4}7^{2}8^{7}.3^{2}$},
\scalebox{0.75}{$1^{7}24^{2}5^{2}6^{4}7^{2}.2^{6}3^{2}5^{2}67^{2}.3^{5}45^{2}7^{3}.4^{4}56^{2}$},
\scalebox{0.75}{$1^{7}234^{2}6^{5}78.2^{6}45^{4}7^{4}8^{2}9.3^{6}567^{2}8^{4}.4^{4}5^{2}69^{6}$},
\scalebox{0.75}{$1^{7}2^{2}35^{3}6^{5}8^{2}.2^{5}34^{2}5^{2}67^{3}.3^{5}5^{2}7^{4}.4^{5}68^{5}$},
\scalebox{0.75}{$1^{7}2345^{5}7^{4}.2^{6}3^{2}56^{3}7^{3}89.3^{4}4^{2}568^{6}9.4^{4}6^{3}9^{5}$},
\scalebox{0.75}{$1^{7}2^{2}345^{2}6^{3}7^{2}8^{2}.2^{5}345^{3}67^{4}.3^{5}456^{3}8^{5}.4^{4}57$},
\scalebox{0.75}{$1^{7}24^{2}5^{5}68.2^{6}3^{2}6^{3}7^{3}8.3^{5}4567^{3}8^{2}.4^{4}56^{2}78^{3}$},
\scalebox{0.75}{$1^{7}2^{2}4^{3}5^{7}8.2^{5}3^{3}67^{4}8^{5}.3^{4}4^{4}79^{7}A.6^{6}7^{2}8A^{6}$},
\scalebox{0.75}{$1^{7}2^{2}5^{4}6^{5}8^{3}.2^{5}3^{3}467^{4}8.3^{4}4^{2}57^{3}.4^{4}5^{2}68^{3}$},
\scalebox{0.75}{$1^{7}2^{5}34^{3}5^{4}6^{3}7^{3}8^{5}.2^{2}3^{5}4^{4}5^{3}6^{4}7^{4}8^{2}9^{7}.3$},
\scalebox{0.75}{$1^{7}2^{2}4^{3}5^{3}6^{2}.2^{5}3^{3}5^{2}6^{4}7^{3}.3^{4}4^{2}5^{2}67^{4}.4^{2}$},
\scalebox{0.75}{$1^{7}235^{5}6^{2}8^{2}.2^{6}4^{3}6^{2}7^{3}8^{2}.3^{6}567^{3}.4^{4}56^{2}78^{3}$},
\scalebox{0.75}{$1^{7}2345^{3}6^{7}9.2^{6}45^{3}7^{3}8^{2}.3^{6}7^{4}8^{2}9^{3}.4^{5}58^{3}9^{3}$},
\scalebox{0.75}{$1^{7}2^{2}35^{5}67.2^{5}3^{2}45^{2}67^{2}8^{4}.3^{4}4^{2}6^{4}7.4^{4}67^{3}8^{3}$},
\scalebox{0.75}{$1^{7}24^{2}5^{3}6^{2}.2^{6}3^{2}5^{2}6^{2}7^{2}.3^{5}46^{3}7^{3}.4^{4}5^{2}7^{2}$},
\scalebox{0.75}{$1^{7}245^{5}6^{2}7^{2}.2^{6}3^{2}56^{3}7^{3}8^{2}.3^{5}4^{2}567^{2}8^{4}.4^{4}68$},
\scalebox{0.75}{$1^{7}245^{4}6^{4}7.2^{6}3^{2}6^{2}7^{3}8^{4}.3^{5}4^{2}57^{2}8.4^{4}5^{2}678^{2}$},
\scalebox{0.75}{$1^{7}2^{2}4^{2}5^{4}6^{3}.2^{5}3^{3}5^{3}6^{2}7^{5}.3^{4}4^{4}6^{2}7^{2}8^{6}.48$},
\scalebox{0.75}{$1^{7}24^{7}7^{6}9^{4}A^{2}.2^{6}3^{2}5^{6}78^{5}A^{5}.3^{5}56^{7}8^{2}9^{3}B^{7}$},
\scalebox{0.75}{$1^{7}25^{5}6^{2}7^{2}8^{2}.2^{6}3^{2}6^{5}7^{2}8.3^{5}4^{3}7^{3}.4^{4}5^{2}8^{4}$},
\scalebox{0.75}{$1^{7}245^{4}6^{3}7.2^{6}3^{2}6^{3}7^{2}8^{3}.3^{5}45^{2}7^{3}8^{2}.4^{5}5678^{2}$},
\scalebox{0.75}{$1^{7}23^{5}4^{7}.2^{6}3^{2}57^{4}8^{2}9^{4}.5^{6}6^{2}78^{3}.6^{5}7^{2}8^{2}9^{3}$},
\scalebox{0.75}{$1^{7}24^{5}5^{2}67^{2}8^{7}.2^{6}3^{2}5^{4}6^{3}79^{7}.3^{5}4^{2}56^{3}7^{4}A^{7}$},
\scalebox{0.75}{$1^{7}234^{3}5^{3}6^{2}7^{2}8^{2}.2^{6}3^{2}5^{4}67^{4}8^{2}.3^{4}4^{4}6^{4}78^{3}$},
\scalebox{0.75}{$1^{7}25^{4}6^{3}78^{2}.2^{6}3^{2}6^{3}7^{4}.3^{5}4^{3}7^{2}8^{5}.4^{4}5^{3}69^{7}$},
\scalebox{0.75}{$1^{7}245^{3}6^{7}.2^{6}3^{2}57^{4}8^{2}9^{4}.3^{5}45^{2}78^{4}9.4^{5}57^{2}89^{2}$},
\scalebox{0.75}{$1^{7}23^{2}46^{2}7^{5}89.2^{6}45^{4}8^{6}.3^{5}456^{4}9^{4}.4^{4}5^{2}67^{2}9^{2}$},
\scalebox{0.75}{$1^{7}25^{3}6^{4}7^{2}.2^{6}3^{2}6^{3}7^{4}8.3^{5}4^{3}78^{5}9^{2}.4^{4}5^{4}89^{5}$},
\scalebox{0.75}{$1^{7}2^{3}3^{2}5^{2}6^{3}7.2^{4}34^{3}5^{2}6^{2}7^{5}8.3^{4}45^{3}6^{2}78^{6}.4^{3}$},
\scalebox{0.75}{$1^{7}2345^{3}6^{5}.2^{6}3^{2}5^{2}6^{2}7^{4}9.3^{4}4^{2}7^{3}8^{6}.4^{4}5^{2}89^{6}$},
\scalebox{0.75}{$1^{7}2^{2}34^{5}67^{2}9.2^{5}3456^{4}7^{4}8.3^{5}46^{2}8^{5}9^{4}.5^{6}789^{2}A^{7}$},
\scalebox{0.75}{$1^{7}235^{3}6^{5}7^{2}.2^{6}4^{3}7^{3}8^{5}.3^{6}5^{2}68^{2}9^{6}.4^{4}5^{2}67^{2}9$},
\scalebox{0.75}{$1^{7}23^{7}8^{2}.2^{6}4^{2}5^{2}67^{4}9^{2}.4^{5}56^{4}8^{4}9.5^{4}6^{2}7^{3}89^{4}$},
\scalebox{0.75}{$1^{7}2^{2}35^{4}6^{4}8.2^{5}34^{2}6^{3}7^{2}8.3^{5}45^{2}7^{2}8^{3}.4^{4}57^{3}8^{2}$},
\scalebox{0.75}{$1^{7}2345^{5}6^{3}8.2^{6}35^{2}6^{2}7^{2}8^{2}.3^{5}46^{2}7^{2}8^{2}.4^{5}7^{3}8^{2}$},
\scalebox{0.75}{$1^{7}234^{5}6^{3}9.2^{6}345^{2}6^{4}8^{3}.3^{5}45^{2}7^{4}8^{2}.5^{3}7^{3}8^{2}9^{6}$},
\scalebox{0.75}{$1^{7}2345^{5}7^{2}89.2^{6}3^{2}6^{4}789^{4}.3^{4}4^{2}567^{4}8.4^{4}56^{2}8^{4}9^{2}$},
\scalebox{0.75}{$1^{7}235^{6}789^{2}.2^{6}3456^{3}7^{2}8.3^{5}4^{2}6^{2}78^{4}9.4^{4}6^{2}7^{3}89^{4}$},
\scalebox{0.75}{$1^{7}245^{3}6^{2}7^{3}9.2^{6}3^{2}6^{4}8^{5}.3^{5}4^{2}57^{4}9^{5}.4^{4}5^{3}68^{2}9$},
\scalebox{0.75}{$1^{7}235^{4}7^{6}A.2^{6}4^{3}6^{3}8^{5}A.3^{6}5^{2}68^{2}9^{5}.4^{4}56^{3}79^{2}A^{5}$},
\scalebox{0.75}{$1^{7}2345^{4}689^{5}.2^{6}3456^{3}7^{3}9A^{5}.3^{5}5^{2}67^{4}89.4^{5}6^{2}8^{5}A^{2}$},
\scalebox{0.75}{$1^{7}245^{3}6^{5}89.2^{6}3^{2}6^{2}7^{4}8^{2}9.3^{5}4^{2}57^{2}8^{4}.4^{4}5^{3}79^{5}$},
\scalebox{0.75}{$1^{7}23^{3}5^{5}6^{4}.2^{6}4^{2}5^{2}6^{2}7^{3}8^{5}.3^{4}4^{2}67^{4}8^{2}9^{7}.4^{3}$},
\scalebox{0.75}{$1^{7}2^{2}345^{5}7^{3}8^{5}.2^{5}3^{2}4^{3}56^{4}8^{2}9^{6}.3^{4}4^{3}56^{3}7^{4}9A^{7}$},
\scalebox{0.75}{$1^{7}235^{4}6^{3}7^{2}9.2^{6}4^{2}56^{2}7^{5}A^{3}.3^{6}568^{6}9^{2}.4^{5}5689^{4}A^{4}$},
\scalebox{0.75}{$1^{7}2^{2}5^{5}6^{2}7^{3}.2^{5}3^{3}46^{2}7^{3}8.3^{4}4^{2}5^{2}68^{4}.4^{4}6^{2}78^{2}$},
\scalebox{0.75}{$1^{7}234^{2}5^{2}67^{3}8^{2}.2^{6}45^{4}67^{2}9^{5}.3^{6}56^{2}7^{2}89.4^{4}6^{3}8^{4}9$},
\scalebox{0.75}{$1^{7}245^{7}7.2^{6}3^{2}6^{4}7^{2}8^{3}.3^{5}46^{2}7^{2}8^{2}9^{5}.4^{5}67^{2}8^{2}9^{2}$},
\scalebox{0.75}{$1^{7}24^{2}5^{2}67^{5}8.2^{6}3^{2}5^{2}6^{2}8^{6}.3^{5}45^{2}679^{6}A.4^{4}56^{3}79A^{6}$},
\scalebox{0.75}{$1^{7}2^{2}5^{3}6^{6}9.2^{5}3^{3}467^{4}8^{2}.3^{4}4^{2}5^{2}78^{5}9.4^{4}5^{2}7^{2}9^{5}$},
\scalebox{0.75}{$1^{7}23^{2}4^{2}5^{3}6^{2}7^{4}.2^{6}34^{2}56^{4}7^{3}8^{2}9.3^{4}5^{3}68^{5}9^{6}.4^{3}$},
\scalebox{0.75}{$1^{7}25^{4}6^{3}78.2^{6}3^{2}6^{3}7^{3}8^{2}9.3^{5}4^{3}7^{3}8^{3}9^{2}.4^{4}5^{3}689^{4}$},
\scalebox{0.75}{$1^{7}23^{5}45^{6}.2^{6}3^{2}467^{7}9^{2}.4^{5}56^{2}8^{3}9^{2}A^{6}.6^{4}8^{4}9^{3}AB^{7}$},
\scalebox{0.75}{$1^{7}23^{2}4^{2}56^{2}7^{4}9.2^{6}45^{3}6^{3}78^{5}.3^{5}5^{2}6^{2}7^{2}8^{2}9^{6}.4^{4}5$},
\scalebox{0.75}{$1^{7}23^{2}5^{3}7^{5}.2^{6}4^{2}5^{2}678^{5}9.3^{5}45^{2}6^{2}79^{6}A.4^{4}6^{4}8^{2}A^{6}$},
\scalebox{0.75}{$1^{7}2^{2}345^{3}7^{7}.2^{5}3^{2}5^{4}8^{2}9^{7}B.3^{4}4^{2}6^{5}A^{7}.4^{4}6^{2}8^{5}B^{6}$},
\scalebox{0.75}{$1^{7}2345^{2}6^{2}7^{5}9.2^{6}45^{3}678^{6}A^{2}.3^{6}56^{3}9^{6}AB^{3}.4^{5}5678A^{4}B^{4}$},
\scalebox{0.75}{$1^{7}23^{2}45^{3}6^{5}78.2^{6}34^{3}56^{2}7^{2}8^{4}9^{2}.3^{4}4^{2}5^{3}7^{4}8^{2}9^{4}.49$},
\scalebox{0.75}{$1^{7}2^{3}3^{3}46^{3}7^{7}8^{2}.2^{4}3^{3}5^{6}8^{5}9^{2}A^{5}.34^{5}56^{4}9^{5}A^{2}B^{7}.4$},
\scalebox{0.75}{$1^{7}25^{4}6^{4}79.2^{6}3^{2}6^{3}7^{3}8^{2}.3^{5}4^{3}7^{2}8^{3}9^{3}.4^{4}5^{3}78^{2}9^{3}$},
\scalebox{0.75}{$1^{7}2^{2}345^{4}67^{3}8^{4}.2^{5}3^{2}4^{3}56^{4}78^{3}9^{2}.3^{4}4^{3}5^{2}6^{2}7^{3}9^{5}$},
\scalebox{0.75}{$1^{7}2345^{3}6^{7}.2^{6}45^{2}7^{5}9^{4}A^{2}.3^{6}57^{2}8^{4}9A^{4}B.4^{5}58^{3}9^{2}AB^{6}$},
\scalebox{0.75}{$1^{7}25^{4}6^{2}7^{3}8^{2}A.2^{6}3^{2}6^{4}8^{5}.3^{5}4^{3}7^{4}9^{4}A.4^{4}5^{3}69^{3}A^{5}$},
\scalebox{0.75}{$1^{7}2^{2}35^{2}6^{3}7^{3}.2^{5}34^{2}5^{2}67^{3}8^{2}9^{2}.3^{5}5^{3}78^{5}9.4^{5}6^{3}9^{4}$},
\scalebox{0.75}{$1^{7}23^{2}4^{2}5^{2}6^{4}7^{4}9.2^{6}3^{2}5^{4}67^{3}8^{6}A.3^{3}4^{3}56^{2}89^{6}A^{6}.4^{2}$},
\scalebox{0.75}{$1^{7}2^{2}35^{2}6^{2}7^{2}.2^{5}34^{2}5^{2}6^{2}78^{3}.3^{5}456^{3}78^{2}.4^{4}5^{2}7^{3}8^{2}$},
\scalebox{0.75}{$1^{7}2^{2}35^{5}78^{2}9^{2}.2^{5}34^{2}6^{4}7^{2}9.3^{5}5^{2}67^{3}8^{3}.4^{5}6^{2}78^{2}9^{4}$},
\scalebox{0.75}{$1^{7}25^{7}6^{7}.2^{6}3^{2}7^{6}9^{2}A^{5}.3^{5}4^{3}8^{4}9^{2}B^{7}.4^{4}78^{3}9^{3}A^{2}C^{7}$},
\scalebox{0.75}{$1^{7}235^{3}6^{2}7^{3}8^{2}.2^{6}4^{3}6^{2}78^{4}9^{3}.3^{6}5^{3}7^{2}89^{4}A.4^{4}56^{3}7A^{6}$},
\scalebox{0.75}{$1^{7}2^{6}4^{2}7^{6}8.23^{7}56^{2}8^{5}9^{6}.4^{5}5^{2}6^{2}789A^{2}B^{6}.5^{4}6^{3}A^{5}BC^{7}$},
\scalebox{0.75}{$1^{7}2345^{2}6^{3}78A^{4}.2^{6}35^{3}67^{3}8^{2}9.3^{5}456^{3}8^{3}9^{3}.4^{5}57^{3}89^{3}A^{3}$},
\scalebox{0.75}{$1^{7}24^{4}5^{3}6^{2}78^{4}9A.2^{6}3^{2}5^{3}6^{2}7^{4}9^{5}A.3^{5}4^{2}56^{3}7^{2}8^{3}9A^{5}.4$},
\scalebox{0.75}{$1^{7}2345^{2}6^{2}7^{3}9.2^{6}3^{2}5^{2}6^{2}78^{3}9.3^{4}4^{2}5^{2}67^{2}89^{4}.4^{4}56^{2}78^{3}9$},
\scalebox{0.75}{$1^{7}234^{6}8^{3}A.2^{6}45^{2}6^{3}7^{2}89^{3}.3^{6}56^{2}7^{3}9^{3}A^{3}.5^{4}6^{2}7^{2}8^{3}9A^{3}$},
\scalebox{0.75}{$1^{7}2^{3}4^{2}6^{5}79^{5}.2^{4}3^{5}567^{6}9A^{3}BC.3^{2}4^{5}8^{7}9ABC^{5}D.5^{6}6A^{3}B^{5}CD^{6}$},
\scalebox{0.75}{$1^{7}2345^{3}6^{2}78^{4}.2^{6}4^{2}5^{2}6^{2}78^{3}A^{3}.3^{6}5^{2}7^{4}9^{3}A.4^{4}6^{3}79^{4}A^{3}$},
\scalebox{0.75}{$1^{7}245^{3}6^{2}78^{4}.2^{6}3^{2}56^{3}7^{2}89^{4}.3^{5}45^{2}7^{4}9^{3}A^{3}.4^{5}56^{2}8^{2}A^{4}$},
\scalebox{0.75}{$1^{7}2345^{2}7^{5}9.2^{6}3^{2}5^{2}6^{2}8^{5}A.3^{4}4^{2}5^{2}6^{2}8^{2}9^{5}.4^{4}56^{3}7^{2}9A^{6}$},
\scalebox{0.75}{$1^{7}2^{3}3^{6}4^{3}56^{5}7^{7}8^{5}9^{6}AB^{4}C.2^{4}34^{4}5^{6}6^{2}8^{2}9A^{6}B^{3}C^{6}D^{7}E^{7}$},
\scalebox{0.75}{$1^{7}24^{2}5^{2}6^{2}7^{6}8.2^{6}3^{2}5^{2}6^{2}78^{6}B.3^{5}4^{3}9^{7}A^{4}.4^{2}5^{3}6^{3}A^{3}B^{6}$},
\scalebox{0.75}{$1^{7}2345^{3}67^{2}89^{2}.2^{6}3^{2}5^{2}6^{2}7^{2}9^{2}.3^{4}4^{2}5^{2}7^{3}8^{3}.4^{4}6^{4}8^{3}9^{3}$},
\scalebox{0.75}{$1^{7}24^{5}7^{3}8^{2}9^{2}A.2^{6}3^{2}5^{2}6^{3}8^{2}9^{3}A.3^{5}4^{2}6^{3}78^{2}9^{2}.5^{5}67^{3}8A^{5}$},
\scalebox{0.75}{$1^{7}2^{2}3^{2}4^{2}6^{5}7^{3}9^{4}.2^{5}3^{2}45^{4}67^{3}8^{3}A^{6}.3^{3}4^{2}5^{3}678^{4}9^{3}AB^{7}.4^{2}$},
\scalebox{0.75}{$1^{7}235^{4}7^{2}89^{6}A.2^{6}3^{2}46^{4}8^{2}9A^{4}B^{2}C.3^{4}4^{2}5^{2}7^{5}AB^{5}.4^{4}56^{3}8^{4}AC^{6}$},
\scalebox{0.75}{$1^{7}2^{2}4^{4}6^{3}8A^{6}D.2^{5}3^{3}46^{4}7^{3}AB^{6}.3^{4}4^{2}5^{2}7^{4}8^{3}BC^{6}.5^{5}8^{3}9^{7}CD^{6}$},
\scalebox{0.75}{$1^{7}24^{2}5^{3}6^{2}8^{3}9^{2}B.2^{6}3^{2}5^{2}6^{3}9^{5}AB.3^{5}45^{2}6^{2}7^{3}A^{5}.4^{4}7^{4}8^{4}AB^{5}$},
\scalebox{0.75}{$1^{7}245^{3}6^{3}789^{3}.2^{6}3^{2}56^{2}7^{3}8^{2}A^{4}.3^{5}45^{2}7^{3}8^{2}9^{2}.4^{5}56^{2}8^{2}9^{2}A^{3}$},
\scalebox{0.75}{$1^{7}2345^{4}8^{4}9^{2}AB.2^{6}345^{2}6^{2}7^{2}9^{4}B.3^{5}456^{2}7^{3}89A^{4}.4^{4}6^{3}7^{2}8^{2}A^{2}B^{5}$},
\scalebox{0.75}{$1^{7}234^{5}78^{6}9C.2^{6}3^{2}45^{3}7^{2}9^{3}A^{3}C.3^{4}45^{3}6^{2}789^{2}A^{3}B^{3}.56^{5}7^{3}9AB^{4}C^{5}$},
\scalebox{0.75}{$1^{7}23^{2}456^{2}7^{3}8^{3}9.2^{6}35^{3}67^{2}8^{3}9^{2}A.3^{4}4^{3}6^{3}9^{4}A^{4}B.4^{3}5^{3}67^{2}8A^{2}B^{6}$},
\scalebox{0.75}{$1^{7}23^{2}5^{3}7^{2}8^{3}9^{2}AB.2^{6}4^{2}6^{5}9^{4}A^{2}B^{2}.3^{5}45^{2}7^{5}A^{4}.4^{4}5^{2}6^{2}8^{4}9B^{4}$},
\scalebox{0.75}{$1^{7}2^{3}5^{2}6^{3}8^{2}9^{4}.2^{4}3^{4}456^{2}7^{2}89A^{5}.3^{3}4^{3}56^{2}7^{3}8A^{2}B^{6}.4^{3}5^{3}7^{2}8^{3}9^{2}BC^{7}$},
\scalebox{0.75}{$1^{7}2^{4}35679^{5}AB^{4}.2^{3}3^{4}4^{2}6^{4}78A^{5}BC^{4}.3^{2}4^{5}6^{2}7^{3}8^{2}9BC^{3}D^{5}.5^{6}7^{2}8^{4}9ABD^{2}E^{7}$},
\scalebox{0.75}{$1^{7}234^{3}6^{2}7^{2}9^{4}A^{2}.2^{6}34^{2}6^{3}7^{2}8^{2}A^{3}B^{2}C.3^{5}4^{2}5^{2}7^{2}8^{2}9AB^{4}C^{2}.5^{5}6^{2}78^{3}9^{2}ABC^{4}$},

\end{flushleft}
}

\bibliographystyle{alpha}
\bibliography{lit}

\end{document}